\documentclass[sigconf]{acmart}
\usepackage{amssymb,amsfonts,amsmath,amsthm,amscd,dsfont,mathrsfs,mathtools}

\usepackage{graphicx}
\usepackage{appendix}
\usepackage{float}	
\usepackage{xcolor}
\usepackage{amsthm}
\usepackage{geometry,enumitem}

\usepackage{comment}
\usepackage{graphicx,subcaption}
\usepackage{varwidth}

\newtheorem{thm}{Theorem}[section]
\newtheorem{defn}[thm]{Definition}

\newcommand{\bpf}{\begin{proof}}
\newcommand{\epf}{\end{proof}}

\newcommand{\blue}[1]{{#1}}

\newcommand{\beqa}{\begin{eqnarray}}

\newcommand{\eeqa}{\end{eqnarray}}

\newcommand{\beqan}{\begin{eqnarray*}}

\newcommand{\eeqan}{\end{eqnarray*}}

\newcommand{\beq}{\begin{equation}}

\newcommand{\eeq}{\end{equation}}

\renewcommand{\C}{{\cal C}}

\newcommand{\eps}{\varepsilon}

\newcommand{\singlespacing}{\let\CS=\@currsize\renewcommand{\baselinestretch}{0.95}\tiny\CS}

\newcommand{\oneandahalfspacing}{\let\CS=\@currsize\renewcommand{\baselinestretch}{1.25}\tiny\CS}

\newcommand{\doublespacing}{\let\CS=\@currsize\renewcommand{\baselinestretch}{1.39}\tiny\CS}

\newcommand{\be}{\begin{equation}}

\newcommand{\ee}{\end{equation}}

\newcommand{\bc}{\begin{center}}

\newcommand{\ec}{\end{center}}

\newcommand{\bfl}{\begin{flushleft}}

\newcommand{\efl}{\end{flushleft}}

\newcommand{\scheme}{{\sf Prism} }
\newcommand{\schemenosp}{{\sf Prism}\ignorespaces}

\newcommand{\fruitchains}{{\sf Fruitchains } }
\newcommand{\fruitchainsnosp}{{\sf Fruitchains}\ignorespaces}
\newcommand{\bitcoinNG}{{\sf BitcoinNG } }
\newcommand{\bitcoinNGnosp}{{\sf BitcoinNG}\ignorespaces}

\newcommand{\disccoinnosp}{{\sf DiscCoin}\ignorespaces}

\newcommand{\byzcoinnosp}{{\sf ByzCoin}\ignorespaces}

\newcommand{\thunderella}{{\sf Thunderella }}
\newcommand{\inclusive}{{\sf Inclusive }}
\newcommand{\spectre}{{\sf Spectre }}
\renewcommand{\phantom}{{\sf Phantom }}
\newcommand{\conflux}{{\sf Conflux }}

\newcommand{\bitcoin}{{\sf Bitcoin} }
\newcommand{\bitcoinnosp}{{\sf Bitcoin}\ignorespaces}

\newcommand{\ghost}{{\sf GHOST} }
\newcommand{\ghostnosp}{{\sf GHOST}\ignorespaces}

\newcommand{\pois}{{\sf Poiss}}
\newcommand{\geom}{{\sf Geometric}}
\newcommand{\bin}{{\sf Bin}}

\makeatletter
\renewcommand\subsubsection{\@startsection{subsubsection}{3}{\z@}%
                       {-18\p@ \@plus -4\p@ \@minus -4\p@}%
                       {4\p@ \@plus 2\p@ \@minus 2\p@}%
                       {\normalfont\normalsize\bfseries\boldmath
                        \rightskip=\z@ \@plus 8em\pretolerance=10000 }}

\newcommand{\N}{{\mathcal{N}}}
\newcommand{\Z}{{\mathcal{Z}}}
\newcommand{\fv}{\bar{f}_v}
\newcommand{\fp}{\bar{f}_p}
\newcommand{\Rl}{R_{\ell}}

\renewcommand{\H}{{\mathcal{H}}}
\renewcommand{\C}{{\mathcal{C}}}
\newcommand{\A}{{\mathcal{A}}}
\newcommand{\T}{{\mathcal{T}}}
\newcommand{\R}{{\mathcal{R}}}
\newcommand{\rmax}{r_{\rm max}}
\newcommand{\rmin}{r_{\rm min}}

\newcommand{\kmin}{k_{\rm min}}
\renewcommand{\Pr}{\mathbb{P}}
\newcommand{\E}{\mathbb{E}}
\renewcommand{\P}{{\mathcal P}}
\newcommand{\Dr}{{\Delta_r}}
\newcommand{\Dmax}{{\Delta_{\text{max}}}}
\newcommand{\cone}{\frac{1-2\beta}{16}}

\makeatother

\newcommand{\tE}{\texttt{E}}
\newcommand{\tT}{\texttt{T}}
\newcommand{\tF}{\texttt{F}}
\newcommand{\Rc}{R^{\text{stop}}}
\newcommand{\Rconf}{R^{\text{conf}}}
\newcommand{\Dll}{D_{\ell', \ell}}
\newcommand{\Dllf}{\big\{D_{\ell', \ell}\big\}_{\ell'\leq \ell}}

\newcommand{\LS}{\texttt{LedSeq}}
\renewcommand{\epsilon}{\varepsilon}

\usepackage{algorithm}
\usepackage[noend]{algpseudocode}
\usepackage{xcolor}
\definecolor{azure}{rgb}{0.54, 0.17, 0.89}
\newcommand{\colorcomment}[1]{\Comment{ {\color{azure} #1}} }
\newcommand{\av}[1]{$#1$}
\newcommand{\maincolorcomment}[1]{{\color{azure}// #1 } }

\usepackage{enumitem,kantlipsum}

\begin{document}

\title{{\sf Prism}: Deconstructing the Blockchain \\ to  Approach Physical Limits}

%


\author{Vivek Bagaria}
\email{vbagaria@stanford.edu}
\affiliation{%
  \institution{Stanford University}
}
\author{Sreeram Kannan}
\email{ksreeram@uw.edu}
\affiliation{%
  \institution{University of Washington at Seattle}
}
\author{David Tse}
\email{dntse@stanford.edu}
\affiliation{%
  \institution{Stanford University}
}
\author{Giulia Fanti}
\email{gfanti@andrew.cmu.edu}
\affiliation{%
  \institution{Carnegie Mellon University}
}
\author{Pramod Viswanath}
\email{pramodv@illinois.edu}
\affiliation{%
  \institution{University of Illinois at Urbana-Champaign}
}


\begin{abstract}
The concept of a blockchain was invented by Satoshi Nakamoto to maintain a distributed ledger.  In addition to its security, important performance measures of a blockchain protocol are its transaction throughput and confirmation latency. In a decentralized setting, these measures are limited by two underlying physical network attributes: communication capacity and speed-of-light propagation delay. In this work we introduce \schemenosp, a new proof-of-work blockchain protocol, which can achieve 1) security against up to $50\%$ adversarial hashing power; 2) optimal throughput up to the capacity $C$ of the network; 3) confirmation latency for honest transactions proportional to the propagation delay $D$, with confirmation error probability exponentially small in the bandwidth-delay product $CD$; 4) eventual total ordering of all transactions. Our approach to the design of this protocol is based on {\em deconstructing} Nakamoto's blockchain into its basic functionalities and systematically  scaling up  these functionalities to approach their physical limits.
\end{abstract}
{\def\addcontentsline#1#2#3{}\maketitle}
%
%
%

\section{Introduction}
In 2008, Satoshi Nakamoto invented the concept of a blockchain,  a mechanism to maintain a distributed ledger in a permissionless setting. Honest nodes mine blocks on top of each other by solving Proof-of-Work  (PoW) cryptographic puzzles;  by following a longest chain protocol, they can come to consensus on a transaction ledger that is difficult for an adversary to alter. Since then, many other blockchain protocols have been invented.

\subsection{Performance measures}
The fundamental performance measures of a PoW blockchain protocol are:
\begin{enumerate}[leftmargin=*]
    \item the fraction $\beta$ of hashing power the adversary can control without compromising system security, assuming the rest of the nodes follow protocol; 
    \item  the throughput $\lambda$, number of transactions confirmed per second;
    \item the confirmation latency, $\tau$, in seconds, for a given probability $\eps$ that a confirmed transaction will be removed from the ledger in the future. 
\end{enumerate} 
For example, \bitcoin is secure against an adversary holding up to $50\%$ of the total network hash power ($\beta = 0.5$), has throughput $\lambda$ of a few transactions per seconds and confirmation latency of the order of tens of minutes to  hours. There is a tradeoff between the confirmation latency and the confirmation error probability: the smaller the desired confirmation error probability, the longer the needed latency is in \bitcoin\!\!. For example, Nakamoto's calculations \cite{bitcoin} show that for $\beta = 0.3$, while it takes a latency of $6$ blocks (on the average, $60$ minutes) to achieve an error probability of $0.15$, it takes a latency of $30$ blocks (on the average, $300$ minutes)  to achieve an error probability of $10^{-4}$. 


\subsection{Physical limits}
\bitcoin has strong security guarantees but its throughput and latency performance are poor. In the past decade, much effort has been expended to improve the performance in these metrics. But what are the fundamental bounds that limit the performance of {\em any} blockchain protocol?

Blockchains are protocols that run on a distributed set of nodes connected by a physical network. As such, their performance is limited by the attributes of the underlying network. The two most important  attributes  are $C$, the communication capacity of the network, and $D$, the speed-of-light propagation delay across the network. Propagation delay $D$ is measured in seconds and the capacity $C$ is measured in transactions per second. Nodes participating in a blockchain network need to communicate information with each other to reach consensus; the capacity $C$ and the propagation delay $D$ limit the {\em rate} and {\em speed} at which such information can be communicated. These parameters encapsulate the effects of both fundamental network properties (e.g., hardware, topology), as well as resources consumed by the network's relaying mechanism, such as validity checking of transactions or blocks. 
\footnote{We define confirmation formally in Section \ref{sec:model}, but informally, we say a node $\epsilon$-confirms a transaction if, upon successfully evaluating a \emph{confirmation rule} under parameter $\epsilon$, the transaction has a probability of at most $\epsilon$ of being reverted by any adversary.}
Assuming that each transaction needs to be communicated at least once across the network, it holds that $\lambda$, the number of transactions which can be confirmed per second, is at most $C$, i.e. 
\begin{equation}
    \lambda < C.
\end{equation}
One obvious constraint on the confirmation latency $\tau$ is that 
\begin{equation}
    \tau > D.
    \label{eq:prop_lb}
\end{equation}
Another less obvious constraint on the confirmation latency comes from the network capacity and the reliability requirement $\varepsilon$.  
Indeed, if the confirmation latency is $\tau$ and the block size is $B_v$ transactions, then at most $C/B_v \cdot \tau$ mined blocks can be communicated across the network during the confirmation period for a given transaction. 
 	These mined blocks can be interpreted as confirmation {\em votes} for a particular transaction during this period; i.e. votes are communicated at rate $C/B_v$ and $C\tau/B_v $ votes are accumulated over duration $\tau$. (The parameter $B_v$ can be interpreted as the minimum block size to convey a vote.) 
 	On average, a fraction $\beta < 0.5$ of these blocks are adversarial, but due to the randomness in the mining process, there is a probability, exponentially small in $C\tau/B_v$, that there are more adversarial blocks than honest blocks; if this happens, confirmation cannot be guaranteed. Hence, this probability is a lower bound on the achievable confirmation error probability, i.e. $\eps = exp(- O(C\tau/B_v))$. Turning this equation around, we have the following lower bound on the latency for a given confirmation probability $\eps$:
\begin{equation}
    \tau = \Omega \left ( \frac{B_v}{C}\cdot \log \frac{1}{\eps}\right).
    \label{eq:lb2}
\end{equation}
Comparing the two constraints, we see that if 
$$ \frac{CD}{B_v} \gg \log \frac{1}{\eps},$$
the latency is limited by the propagation delay; otherwise, it is limited by the confirmation reliability requirement. The quantity $CD/B_v$ is analogous to the key notion of {\em bandwidth-delay product} in networking (see eg. \cite{katabi2002congestion}); it is the number of  ``in-flight" votes in the network. 


To evaluate existing blockchain systems with respect to these limits, consider a global network with communication links of capacity $20$ Mbits/second and speed-of-light propagation delay $D$ of $1$ second. If we take a vote block of size $100$ bytes, then the bandwidth-delay product $CD/B_v = 25000$ is very large. Hence, the confirmation latency is limited by the propagation delay of $1$ seconds,  but not by the confirmation reliability requirement unless it is astronomically small.  Real-world blockchains operate far from these physical network limits.  \bitcoin\!\!, for example, has $\lambda$ of the order of $10$ transactions per second, $\tau$ of the order of minutes to hours, and is limited by the confirmation reliability requirement rather than the propagation delay.
Ethereum has $\lambda \approx 15$ transactions per second and $\tau \approx 3$ minutes to achieve an error probability of $0.04$ for $\beta=0.3$ \cite{vitalik_blocktimes}.
\subsection{Main contribution}
The main contribution of this work is a new blockchain protocol, \schemenosp, which, in the face of any powerful adversary\footnote{The powerful adversary will be precisely defined in the formal model.} with power $\beta < 0.5$,  can {\em simultaneously} achieve:
\begin{enumerate}[leftmargin=*]
    \item {\bf Security:} (Theorem \ref{cor:latency_ordering}) a total ordering of the transactions,  with consistency and liveness guarantees.
    \item {\bf Throughput:} (Theorem \ref{cor:throughput})  a throughput 
    \begin{equation}
    \label{eq:thruput_result}
        \lambda = 0.9 (1-\beta) C \quad \mbox{transactions per second.}
    \end{equation}
    
    \item {\bf Latency:} (Theorem \ref{cor:latency_fast}) confirmation of all honest transactions (without public double spends) with an expected latency of
    \begin{equation}
    \label{eq:latency_result}
        \mathbb{E}[\tau] < \max \left \{ c_1(\beta) D,  c_2(\beta)\frac{B_v}{C} \log \frac{1}{\eps} \right \}   \quad \mbox{seconds},
    \end{equation}
    with confirmation reliability at least $1-\eps$ (Figure \ref{fig:fig1}).
    Here, $c_1(\beta)$ and $c_2(\beta)$ are constants depending only on $\beta$ 

\end{enumerate} 

Notice that the worst-case optimal throughput of any protocol with $1-\beta$ fraction of hash  power is $(1-\beta)C$ transactions/second, assuming each transaction needs to be communicated across the network.
Hence, \schemenosp's throughput is near-optimal. 
At the same time, \scheme achieves a confirmation latency for honest transactions matching the two physical limits \eqref{eq:prop_lb} and \eqref{eq:lb2}.
In particular, if the desired security parameter $\log \frac{1}{\eps} \ll CD/B_v$, the confirmation latency is of the order of the propagation delay and {\em independent} of $\log 1/\eps$. Put another way, one can achieve latency close to the propagation delay with a confirmation error probability exponentially small in the bandwidth-delay product $CD/B_v$. Note that the latency is worst-case over all adversarial strategies but averaged over the randomness in the mining process.

To the best of our knowledge, no other existing PoW protocol has guaranteed performance which can match that of \schemenosp. Two novel ideas which enable this performance are 1) a {\em total decoupling} of transaction proposing, validation and confirmation functionalities in the blockchain, allowing performance scaling;  2) the concept of confirming a {\em list} of possible ledgers rather than a unique ledger, enabling honest non-double-spend transactions to be confirmed quickly\footnote{This idea was inspired by the concept of list decoding from information theory.}.

\begin{figure}
    \includegraphics[width=\linewidth]{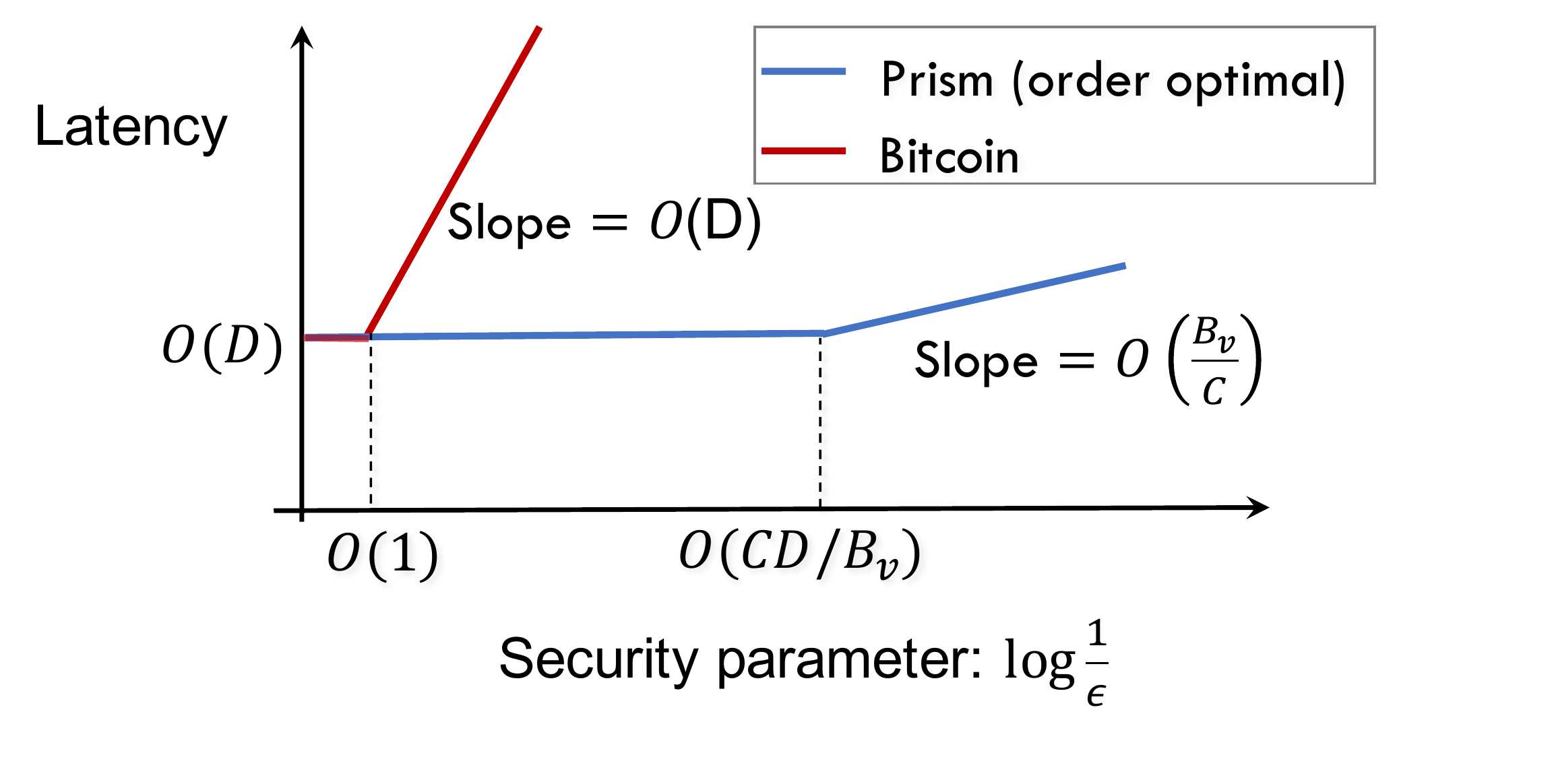}
 \caption{ 
 Confirmation latency vs. security parameter for \scheme. The latency of \scheme is independent of the security parameter value up to order $CD/B_v$ and increases very slowly after that (with slope $B_v/C$).  For \bitcoin, latency increases much more rapidly with the security parameter, with slope proportional to $D$. (Since $CD/B_v \gg 1$, this latter slope is much larger.)
 }
  \label{fig:fig1}
\end{figure}

\subsection{Performance of existing PoW protocols}



\noindent{\bf High forking protocols}. Increasing the mining rate in \bitcoin can decrease latency and improve throughput, however, this comes at the expense of decreased security \cite{ghost}. Thus, unlike \schemenosp, the throughput and latency of \bitcoin is {\em security-limited} rather than {\em communication-limited}.
 To increase the mining rate while maintaining security, one line of work (\ghost\!\! \cite{ghost},  \inclusive\!\! \cite{inclusive},  \spectre\!\! \cite{spectre}, \phantom\!\!~\cite{phantom}, \conflux\!\! \cite{conflux}) in the literature has used more complex fork choice rules and added reference links to convert the blocktree into a directed acyclic graph (DAG). 
 This allows blocks to be  voted on by blocks that are not necessarily its descendants. 

 While \ghost remains secure at low mining rates\cite{kiayias2016trees}, there is a balancing attack by the adversary \cite{ghost_attack,better}, which severely limits the security at high mining rates. Thus, like \bitcoinnosp, the throughput of \ghost is security-limited. The other protocols \inclusive and \conflux that rely on \ghost inherit this drawback. While  \spectre and \phantom improve  latency and throughput,  \spectre cannot provide a total order on all transactions (required for smart contracts) and \phantom does not yet have a formal proof of security. 
 
\noindent{\bf Decoupled consensus}.  Protocols such as \bitcoinNG\!\!\cite{bitcoin-ng} decouple transaction proposal and leader election (which are coupled together in \bitcoinnosp).   \bitcoinNG elects a single leader to propose many transaction blocks till the next leader is elected by PoW. While this achieves high throughput, the latency cannot be reduced using this approach.  
Furthermore, \bitcoinNG is vulnerable to bribery or DDoS attacks, whereby an adversary can corrupt a leader after learning its identity (unlike \bitcoinnosp).
Subchains \cite{subchains} and weak blocks \cite{weakblocks1,weakblocks2} both  employ blocks with lower hash threshold (``weak blocks'') along with regular blocks in an attempt to scale throughput. However, since weak blocks are required to form a chain, it does not achieve the optimal throughput.

\noindent{\bf Hybrid blockchain-BFT consensus}. 
Several protocols combine ideas from Byzantine fault tolerant (BFT) based consensus into a PoW setting \cite{byzcoin,hybrid,thunderella, abraham2016solida}.  \byzcoinnosp~\cite{byzcoin} and its predecessor \disccoinnosp~\cite{disccoin} attempt to address the latency shortcoming of \bitcoinNG  but is proven in a later paper \cite{hybrid} to be insecure when the adversarial fraction $\beta > 0.25$. \emph{Hybrid consensus} uses a combination of proof-of-work based committee selection with Byzantine fault tolerance (BFT) consensus \cite{hybrid}. However, this protocol is secure only till $\beta = 0.33$. While the protocol latency is responsive, i.e., it decreases with network delay linearly, for a known network delay, it has similar non-optimal dependence on $\epsilon$ as \bitcoinnosp.

A closely-related protocol called \thunderella  \cite{thunderella} achieves very low latency under optimistic conditions, i.e., when the leader is honest and $\beta < 0.25$. However even when $\beta$ is very small, a dishonest leader can keep delaying transactions to the \bitcoin latency (since such delaying behavior is detected by a slow PoW blockchain). 

\subsection{Our Approach}


 \begin{figure}
   \centering
   \includegraphics[width=\linewidth]{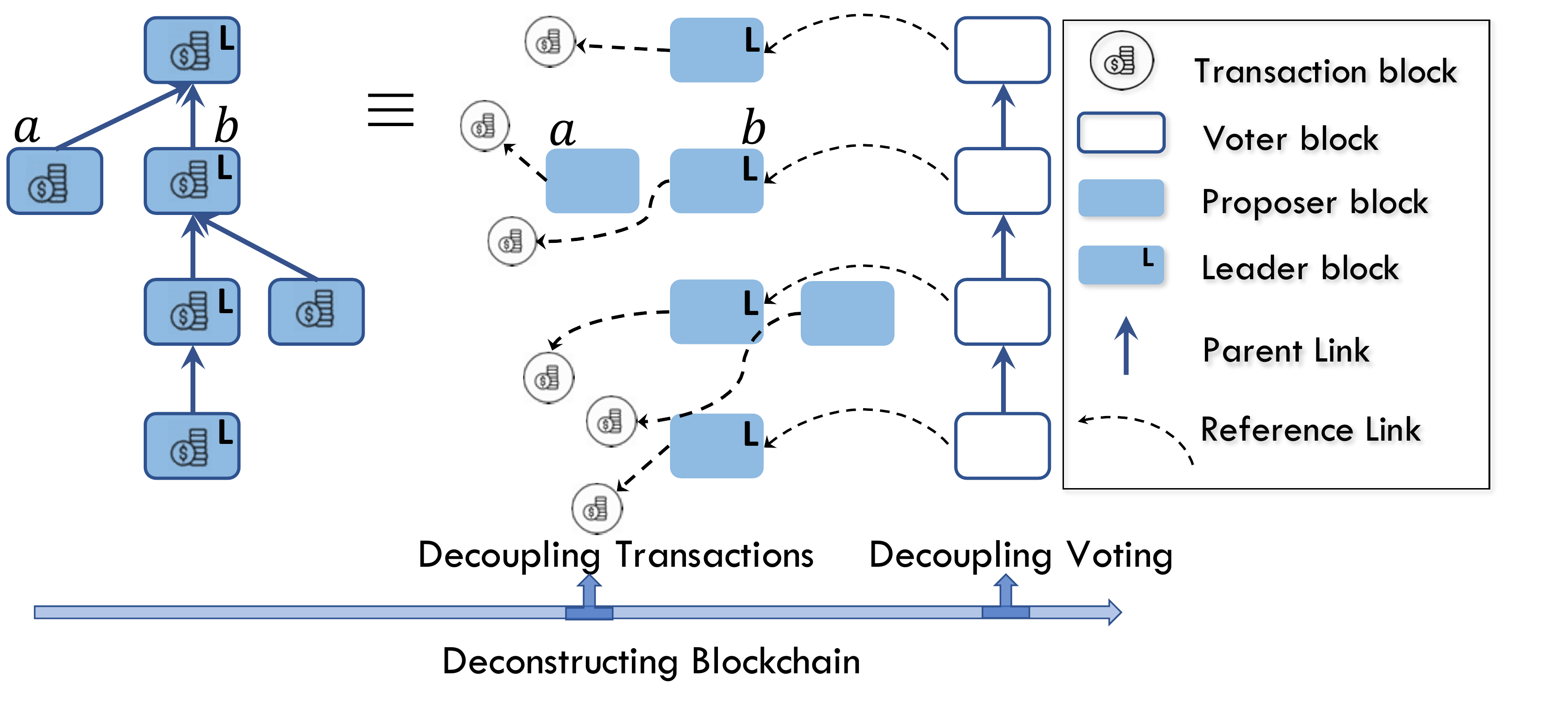} 
   \caption{Deconstructing the blockchain into transaction blocks, partially ordered proposal blocks arranged by level,  and voter blocks organized in a voter tree. The main chain is selected through voter blocks, which vote among the proposal blocks at each level to select a leader block. For example, at level $3$, block $b$ is elected the leader over block $a$. }  
   \label{fig:deconstruct}
 \end{figure}
 
  \begin{figure}
   \centering
   \includegraphics[width=\linewidth]{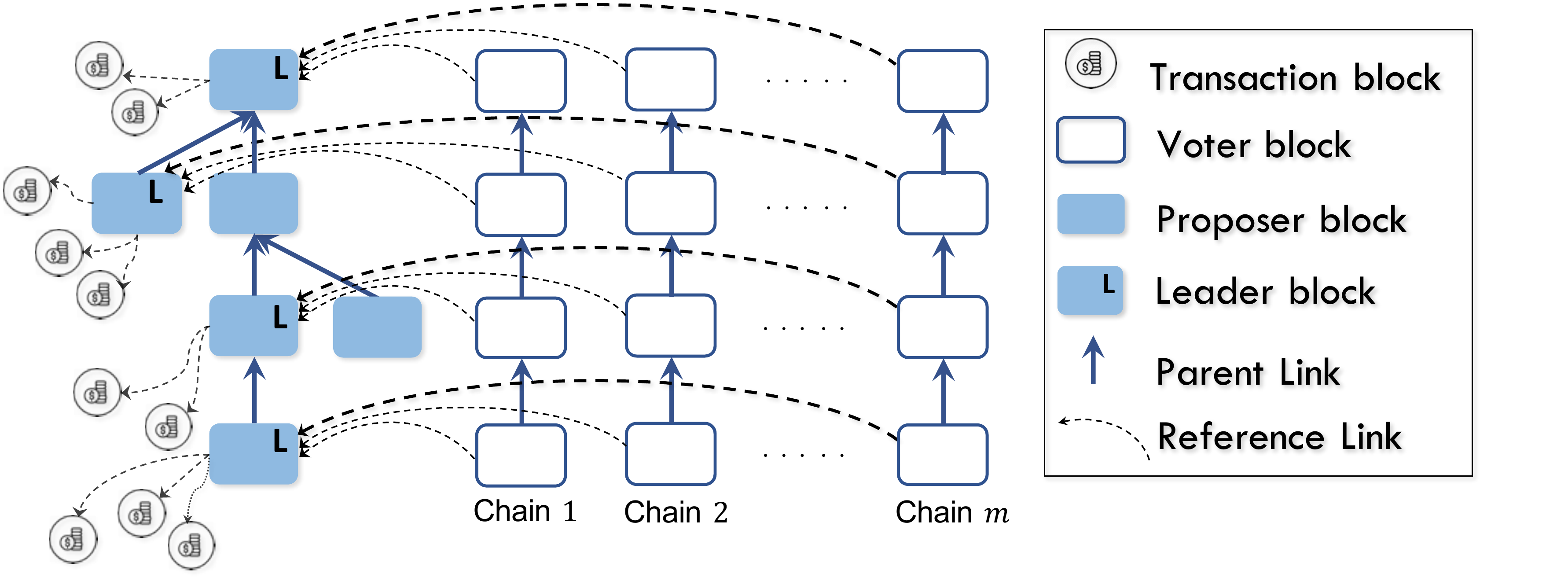} 
   \caption{\schemenosp. Throughput, latency and reliability are scaled to the physical limits by increasing the number of transaction blocks and the number of parallel voting chains. }  
   \label{fig:parallelize}
 \end{figure}
 
Increasing the mining rate is critical to improving the throughput and latency of blockchain protocols.
The challenges facing the DAG approaches arise from the fact that the DAG is {\em unstructured}, due to the excessive random forking when the mining rate is increased. In contrast, \scheme is based on a {\em structured} DAG created by cryptographic sortition of the mined blocks into different types of different functionalities and scaling these functionalities separately.

\noindent
{\bf Deconstruction.}
We start by deconstructing the basic blockchain structure into its atomic functionalities, illustrated in Figure \ref{fig:deconstruct}. The selection of a main chain in a blockchain protocol (e.g., the longest chain in \bitcoin\!\!) can be viewed as electing a leader block among all the blocks at each level of the blocktree, where the level of a block is defined as its distance (in number of blocks) from the genesis block. Blocks in a blockchain then serve three purposes: they stand for election to be leaders, they add transactions to the main chain, and they vote for ancestor blocks through parent link relationships.  We explicitly separate these three functionalities by representing the blocktree in a conceptually equivalent form (Figure \ref{fig:parallelize}).
In this representation, blocks are divided into three types: proposer blocks, transaction blocks and voter blocks.
The voter blocks vote for transactions indirectly by voting for proposer blocks, which in turn link to transaction blocks . Proposer blocks are grouped according to their level in the original blocktree, and each voter block votes among the proposer blocks at the same level to select a leader block among them. The elected leader blocks can then bring in the transactions to form the final ledger.
The valid voter blocks are the ones in the longest chain of the voter tree, and this longest chain maintains the security of the whole system.

\noindent
{\bf Scaling.}
This alternative representation of the traditional blockchain, although seemingly more complex than the original blockchain representation, provides a natural path for scaling performance to approach physical limits (Figure \ref{fig:parallelize}). To increase the transaction throughput, one can simply increase the number of transaction blocks that a proposer block points to without compromising the security of the blockchain. This number is limited only by the physical capacity of the underlying communication network. To provide fast confirmation, one can increase the number of parallel voting trees,  voting on the proposal blocks in parallel to increase the voting rate, until reaching the physical limit of confirming with speed-of-light latency and extremely high reliability. Note that even though the overall  block generation rate has increased tremendously,  the number of proposal blocks per level remains small and manageable, and the voting blocks are organized into many separate voting chains with low block mining rate per chain and hence little forking.  The overall structure, comprising of the different types of blocks and the links between them, is a structured DAG.

\noindent
{\bf Sortition.}
The sortition of blocks into the three types of blocks, and further into blocks of different voting trees, can be accomplished by using the random hash value when a block is successfully mined. This sortition splits the adversary power equally across the structures and does not allow it to focus its power to attack specific structures. This sortition is similar to the 2-for-1 PoW technique used in \cite{backbone}, which is also used in \fruitchains\!\!~\cite{fruitchains} for the purpose of providing fairness in rewards. In fact, the principle of {\em decoupling} functionalities of the blockchain, central to our approach, has already been applied in \fruitchainsnosp, as well as other works such as \bitcoinNGnosp.  The focus of these works is only on decoupling the transactions-carrying functionality. In our work, we broaden this principle to decouple {\em all} functionalities.
\noindent
{\bf Concurrent work.} We were made aware of two independent but related works  \cite{parallel,DBLP:journals/corr/abs-1811-12628} which appeared after we posted this work online. \cite{parallel} proposes two protocols, one achieves high throughput $O(C)$ but \bitcoin latency, and the other achieves low latency $O(1/\sqrt{C})$ but low throughput $O(1)$. In contrast, \scheme achieves {\em simultaneously} high throughput $O(C)$ and even lower latency $O(1/C)$. Although \cite{parallel} also uses the concept of multiple chains, the key difference with \scheme is that there is {\em no} decoupling: the blocks in each chain both carry transactions and vote. Thus, either different transactions are put on the different chains to increase throughput, but the voting rate is low and hence the latency is poor, or the same transaction is repeated across all the chains to increase the voting rate, but the throughput is poor. In contrast, \scheme decouples blocks into transaction blocks and voter blocks, tied together through proposer blocks, and allocate a fraction of the network capacity to each to deliver both low latency and high throughput. The protocol in \cite{DBLP:journals/corr/abs-1811-12628} is similar to first one in \cite{parallel}, achieving  high throughput but only \bitcoin latency.

\subsection{Outline of paper}

Section \ref{sec:model} presents our model.  It is a combination of the synchronous model used in  \cite{backbone} and a network  model that ties the blockchain parameters to physical parameters of the underlying network. In Section \ref{sec:protocol}, we give a pseudocode description of \schemenosp.  The analysis of the security, throughput and latency of \scheme is presented in Section \ref{sec:analysis}, with details of the proofs in the appendices. Section \ref{sec:sims} contains simulation results.


\section{Model}
\label{sec:model}

We consider a synchronous, round-based network model similar to that of Garay \emph{et al.} \cite{backbone}.
We define a blockchain protocol as a pair $(\Pi, g)$, where $\Pi$ is an algorithm that maintains a blockchain data structure $\C$ consisting of a set of \emph{blocks}. 
The function
$g({\sf tx},\C)$ encodes a \emph{ledger inclusion rule}; it takes in a transaction ${\sf tx}$ and a blockchain $\C$, and outputs  $g({\sf tx},\C)=1$ if ${\sf tx}$ 
is contained in the ledger defined by blockchain $\C$ 
and $0$ otherwise. 
For example,  in  \bitcoinnosp, $g({\sf tx},\C)=1$ iff ${\sf tx}$ appears in any block on the longest chain.
If there are multiple longest chains, $g$ can resolve ties deterministically, e.g., by taking the chain with the smallest hash value.

The blockchain protocol proceeds in  rounds of  $\Delta$ seconds each.
Letting $\kappa$ denote a security parameter, the \emph{environment} $\Z(1^\kappa)$ captures all aspects external to the protocol itself, such as inputs to  the protocol (i.e., new  transactions) or interaction with outputs.


Let $\N$ denote the set of participating nodes.
The set of \emph{honest} nodes $\H\subset \N$  strictly follow the blockchain protocol $(\Pi,f)$.
\emph{Corrupt} nodes $\N \setminus \H$ are collectively  controlled by an adversarial party $\A$. 
Both honest and corrupt nodes interact with a random function $H:\{0,1\}^*\to \{0,1\}^\kappa$ through an oracle ${\sf H}(x)$, which  outputs $H(x)$.
In each round, each node $n \in \N$ is allowed to query the oracle ${\sf H}(\cdot)$ at most $q$ times.
The adversary's corrupt nodes are collectively allowed up to $\beta  q |\N|$ sequential queries to oracle ${\sf H}(\cdot)$, where $\beta < 0.5$ denotes the fraction of adversarial hash power, i.e., $1-\frac{|\H|}{|\N|}= \beta$.\footnote{$\beta$ for {\bf b}ad. Like \cite{backbone}, we have assumed all nodes have the same hash power, but this model  can easily be generalized to arbitrary hash power distributions.}  
Like \cite{backbone}, the environment is not allowed to access the oracle. 
These restrictions model the limited hash rate in the system.


In an execution of the blockchain protocol, the environment $\Z$ first initializes all nodes  as either honest or corrupt;
like \cite{backbone}, once the nodes are initialized, the environment can adaptively  change the set $\H$ between rounds, as long as the adversary's total hash power remains bounded by $\beta$.
Thereafter,  the protocol proceeds in  rounds.
In each round, the environment first delivers inputs to the appropriate nodes (e.g., new transactions), and the adversary delivers any messages to be delivered in the current round.
Here, delivery means that the message appears on the recipient node's input tape.
Nodes incorporate the inputs and any messages (e.g., new blocks) into their local blockchain data structure according to protocol $\Pi$.
The nodes then  access the random oracle ${\sf H(\cdot)}$ as many  times as their hash  power allocation allows.
Hence, in each round, users call the oracle ${\sf H(\cdot)}$ with different nonces $s$ in an attempt to find a valid proof of work. 
If an oracle call produces a proof of work, then  the node can deliver a new block to the environment. 
Note that the computational constraints on calling oracle ${\sf H(\cdot)}$ include block validation. 
Since each block only needs to be validated once, validation represents a small fraction of computational demands.

Since each node is allowed a finite number  of calls to ${\sf H(x)}$ in each  round,  the number of blocks mined per round is a Binomial random variable. 
To simplify the analysis, we consider a limit of our model as the number of nodes $|\N| \to \infty$.
As $|\N|$ grows, the proof-of-work threshold adjusts such that the expected number of blocks mined per  round remains constant.
Hence, by the Poisson limit theorem, the number of voter blocks mined per round converges to a Poisson random variable. 

All messages broadcast to the environment  are delivered by the adversary. 
The adversary has various capabilities and restrictions.  
(1) Any message broadcast by an honest node in the previous round must be delivered by the adversary at the beginning of the current round to all remaining honest nodes. 
However, during delivery, the adversary can present these messages to each honest node in whatever order it chooses. 
(2) The adversary cannot forge or alter any message sent by an honest node.
(3) The adversary can control the actions of corrupt nodes.
For example, the adversary can choose how corrupt nodes allocate their hash power, decide block content, and release mined blocks. 
Notably,  although honest blocks publish mined blocks immediately,  the adversary may choose to keep blocks they mined private and release in future round.
(4) The adversary can deliver corrupt nodes' messages to some honest nodes in one round, and the remaining honest nodes in the next round.
We consider a ``rushing'' adversary that observes the honest nodes' actions before taking its  own action for a given round.
Notice that we do not model rational users who are not necessarily adversarial but nevertheless may have incentives to deviate from protocol. 
\paragraph{Physical Network Constraints}
To connect to the physical parameters of the network, we assume a simple network model. Let $B$ be the size of a block, in units of number of transactions. The network delay $\Delta$ (in seconds)  is given by:
\begin{equation} 
\label{eq:delay_coarse}
\Delta  = \frac{B}{C} + D
 \end{equation}
i.e. there is a processing delay of $B/C$ followed by a propagation delay of $D$ seconds.  This is the same model used in \cite{ghost}, based on empirical data in \cite{decker}, as well in \cite{subchains}. 
Notice that the network delay $\Delta$ is by definition  equal to the duration of a single round.

In practice, networks cannot transport an infinite number of messages at once. 
We model this by allowing the environment to transport only a finite volume of messages per round.
This volume is parametrized by the \emph{network capacity} $C$, measured in units of transactions per second.
Hence, during each round, the environment can process  a message volume  equivalent to at most $\Delta C$ transactions.
This puts a constraint on the number of blocks mined per unit time in any protocol.
This \emph{stability constraint} differentiates our model from prior work, which has traditionally assumed infinite network capacity; in particular, this gives us a foothold for quantifying physical limits on throughput and latency.

For simplicity, we assume that the dissemination of new transactions consumes no bandwidth. 
Instead, the cost of communicating transaction messages is captured when the environment transmits blocks carrying transactions.
In  other words, we assume that the cost of  transmitting  transactions is counted only once.

\paragraph{Metrics.}
We let random variable ${\sf VIEW}_{\Pi,\A,\Z}$ denote the joint view of all parties over all rounds; here we have suppressed the dependency on security parameter $\kappa$.
The randomness is defined over the choice of function $H(\cdot)$,  as well as any randomness in the adversary $\A$ or environment $\Z$. 
Our goal is to reason about the joint  view for all possible adversaries $\A$ and environments $\Z$. 
In particular, we want to study the evolution of $\C_i^r$, or the blockchain of each honest node $i\in \H$ during round $r$.
Following the \bitcoin backbone protocol model \cite{backbone}, we consider protocols that execute for a finite execution  horizon $\rmax$, polynomial in $\kappa$.  
Our primary concern will be the efficiency of \emph{confirming} transactions.

\begin{defn}
\label{def:epsilon}
We say a transaction ${\sf  tx}$ is  $(\epsilon,\A,\Z,r_0,\kappa)$-cleared iff under an adversary $\A$,  environment  $\Z$, and security parameter $\kappa$, 
$$
 \mathbb P_{{\sf VIEW}_{\Pi,\A,\Z}} \left( \bigcap_{\substack{r \in \{r_0, \ldots, \rmax\}\\  i \in \H }}  \left \{g({\sf tx},\C_i^r)=b\right \}  \right)\geq 1-\epsilon - {\sf negl}(\kappa),
$$
where $b\in\{0,1\}$; $b=1$ corresponds to confirming the transactions and $b=0$ corresponds to rejecting the transaction.
%
\end{defn}
That is, a transaction is considered confirmed (resp. rejected) if all honest party will include (resp. exclude) it from the ledger with probability more than $\epsilon$ plus a term negligible in $\kappa$ resulting from hash collisions, which we ignore in our analysis. 
We suppress the notation  $\kappa$ from here on. 

Our objective is to optimize two properties of a  blockchain protocol: the throughput and latency of confirming transactions. 
We let $|S|$ denote the number of elements in set $S$.
We let $\T$  denote the set of all transactions generated during the execution horizon, and $\T^r$ denote all transactions delivered up to  and  including  round $r$. 

\begin{defn}[Throughput]
We say a blockchain protocol $\Pi$ supports a throughput of $\lambda$ transactions per round if there exists $U_\epsilon$,linear in $\log(1/\epsilon)$, such that for all environments $\Z$ that produce at most $\lambda$ transactions per round, and for $\forall\;r\in [1,\rmax]$, 
\begin{equation}
    \max_{\A}\left|\left\{{\sf tx}\in \T^r ~:~{\sf tx} \text{  is not }(\epsilon,\A,\Z,r)\text{-cleared} \right\}\right| < \lambda  U_\epsilon.
\label{def:throughput}
\end{equation}
The \emph{system throughput} 
is the largest throughput that a blockchain protocol can  support.
\end{defn}
Notice that although $|\T^r|$ grows with $r$, the right-hand side of \eqref{def:throughput} is constant in  $r$; this implies that the system throughput $\lambda$ is the expected \emph{rate} at which we can clear transactions maintaining a bounded transaction queue, taken worst-case over adversary $\A$ and environments $\Z$ producing at most $\lambda\Delta$ transactions per round.

\begin{defn}[Latency]
For a transaction ${\sf  tx}$, let $r({\sf tx})$ denote the round in which the transaction was first introduced by the envioronment, and let random variable $ R_{\epsilon}({\sf tx})$ denote the smallest round $r$ for which ${\sf  tx}$ is $(\epsilon,\A,\Z,r)$-cleared.
The expected $\epsilon$-\emph{latency} of transaction ${\sf  tx}$ is defined as:
\begin{equation}
    \begin{aligned}
 \tau_\epsilon(\textsf{tx}) &\triangleq       &\underset{ \Z, \A}{\max}& & E_{{\sf VIEW}_{\Pi,\A,\Z}} \left [ R_\epsilon({\sf tx}) -  r({\sf tx}) \right]
    \end{aligned}
    \label{def:latency}
\end{equation}
\end{defn}

Note that if all transactions have finite $\epsilon$-latency, it implies that the blockchain has both consistency and liveness properties. 

\section{Protocol description}
\label{sec:protocol}

\begin{figure}
\includegraphics[width=\linewidth]{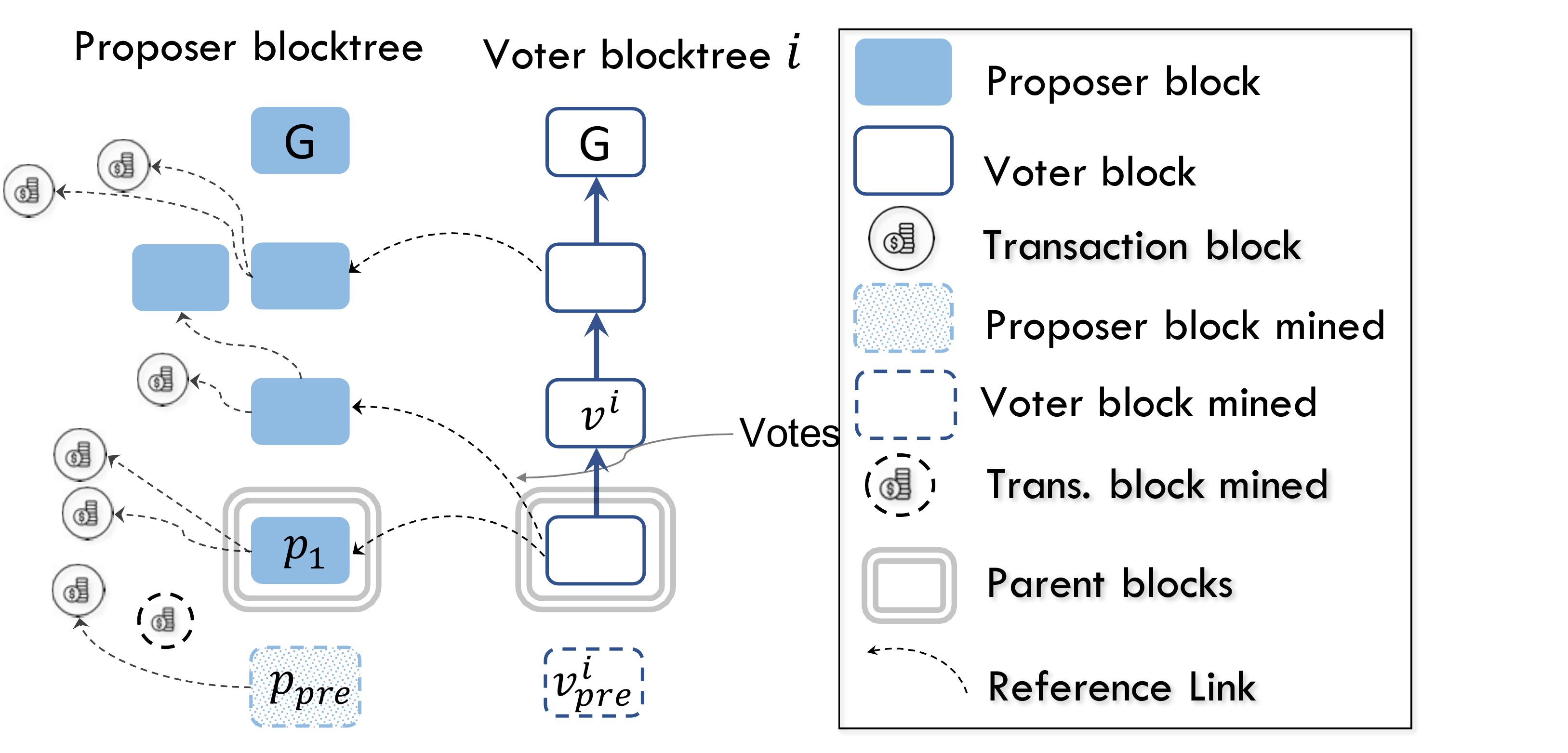}
   \caption{Snapshot of a miner's blocktree: The previously mined blocks have solid boundary whereas blocks which are being mined have dotted-boundary. A miner simultaneously  mines on $p_1$, parent on proposer blocktree,  $v^i$, parent on voter block blocktree $i$($\forall i\in[m]$).}
   \label{fig:prism_mining}
 \end{figure}

We first describe the content and roles of three types of blocks in the \schemenosp$(\Pi, g)$ blockchain. 
We then present Algorithm \ref{alg:prism_minining}, which defines the protocol $\Pi$ and the blockchain data structure $C$. 
We then define the \textit{ledger inclusion rule}, $g$, in Algorithm \ref{alg:prism_con}. 
Due to space constraints, all pseudocode for these algorithms can be found in Appendix \ref{app:pseudocode}.
\schemenosp's blockchain data structure, $C$, has one proposer blocktree and $m$ voter blocktrees, as shown in Figure \ref{fig:parallelize}. 
We use these different blocktrees to maintain three distinct types of blocks:

     \textit{Proposer blocks:} Proposer blocks represent the skeleton of the \scheme blockchain and are mined on the proposer blocktree according to the longest-chain rule.  
    The \textit{level} of a proposer block is defined as its distance from the proposer genesis block.  
    The blocktree structure is only utilized in our protocol as a proof of level of a given proposal block. To construct the ledger, our protocol selects proposal block sequences where one block is chosen at each level. Proposer blocks can \emph{refer} to  transaction blocks and other proposer blocks by including pointers to referred blocks in their payload.
    For example, in Fig \ref{fig:prism_mining},  the proposer blocktree has two proposer blocks mined at level $1$, and one proposer block mined at levels $2$ and $3$, and they point to five transaction blocks in total.
    
     \textit{Voter blocks:}
    Voter blocks are mined on $m$ separate voter blocktrees, each with its own genesis block, according to the longest chain rule. 
    We say a voter block \emph{votes} on a proposer block $B$ if it includes a pointer to $B$ in its payload.
    Note that unlike many BFT consensus protocols, a malicious miner in Prism cannot equivocate when voting because voter blocks are sealed by proof of work.
    Even if a miner mines conflicting voter blocks and tries to send them to disjoint sets of honest users, all users will receive both blocks within one round.
    Each longest chain from each voter blocktree can cast at most one vote for each level in the proposer blocktree.
    More precisely, a voter block votes on all levels in the proposer tree that are unvoted by the voter block's ancestors. 
    Therefore, the voter trees collectively  cast at most $m$ votes on a given level of the proposer blocktree.
    Fig. \ref{fig:parallelize} shows voter blocktree $i$ and its votes (dotted arrows) on each level of the proposer blocktree.
    For each level $\ell$ on the proposer blocktree, the block with the highest number of votes is defined as the \textit{leader} block of level $\ell$. 
    
     \textit{Transaction blocks:} Transaction blocks contain transactions and are mined on the proposer blocktree as in Fig. \ref{fig:parallelize}.
    Although transaction  blocks are  not considered part of the proposer blocktree,  each  transaction block has  a proposer block as its  parent.  

The process by which a  transaction is included in the ledger is as follows: 
(1) the transaction is included in a transaction block $B_T$.
(2) $B_T$ is referred  by a proposer block $B_P$. 
(3) Proposer block $B_P$ is confirmed, either directly (by becoming a leader) or indirectly (e.g., by being referred by a leader). 

\subsection{ Protocol $\Pi$}
Algorithm \ref{alg:prism_minining} presents  \schemenosp's protocol $\Pi$. 
The protocol begins with a trusted setup, in which the environment generates  genesis blocks for the proposer blocktree and each of the $m$ voter blocktrees.
Once the trusted setup completes, the protocol enters the mining loop.  

Whereas \bitcoin miners mine on a single blocktree, \scheme miners simultaneously mine  one proposer block, one transaction block, and $m$ voter blocks, each with its own parent and content. 
This simultaneous mining happens via cryptographic sortition. 
Roughly, a miner first generates a  ``superblock'' that contains enough information for all $m+2$ blocks simultaneously.
It then tries different nonce values; upon mining a block, the output of the hash is deterministically mapped to either a voter block (in one of the $m$ trees), a transaction block, or a proposer block (lines 41-47 in  Algorithm \ref{alg:prism_minining}).
After sortition,  the miner discards unnecessary information and releases the block  to the environment.

More precisely, while mining, each miner maintains outstanding content for each of the  $m+2$ possible mined blocks.  
In \bitcoinnosp, this content would be the transaction memory pool,  but since \scheme  has multiple types of blocks, each  miner stores different content for each block type.
For transaction blocks, the content consists  of all transactions that have not yet been included in a transaction block. 
For proposer blocks, the content is a list of transaction blocks and proposer blocks that have not been  referred by any other proposer block. 
For voter blocks in  the $i$th voter tree, the content is a list of proposer blocks at each level in the  proposer blocktree that has not yet received a vote in the longest chain of the $i$th voter tree.
If a miner observes multiple  proposer blocks at the same  level, it  always votes  on  the first one it received.  
For example, in Figure \ref{fig:prism_mining}, voter block $v^i_{new}$ votes on one proposer block on levels 3 and 4 because its ancestors have voted on level 1 and 2.

Upon collecting this content, the miner generates a block.  
Instead of naively including all the $m+2$ parents\footnote{Proposer and tx block share the same parent and are included twice for simplicity.} and content hashes in the block, \schemenosp's header contains a) the Merkle root of a tree with $m+2$ parent blocks, b) the Merkle root of a tree with $m+2$ contents, and c) a nonce.
Once a valid nonce is found, the block is sortitioned into a proposer block, a transaction block, or a voter block on one  of the $m$ voter trees. 
The mined, sortitioned block consists of the header, the appropriate parent and content, and their respective Merkle proofs.
For instance, if the mined block is a proposer block, it would contain only the proposer parent reference, proposer content,  and appropriate Merkle proofs; 
it would \emph{not}  store transactions or votes.

While mining, nodes may receive blocks from  the network, which are processed in  much the same way as \bitcoinnosp.
Upon receiving a new block,  the  miner first checks validity. 
A block $B$ is valid if it satisfies the PoW inequality and the miner has all the blocks (directly or indirectly) referred by $B$.  
If the  miner lacks some referred blocks, it requests them from the network.
Upon  receiving  a valid transaction block $B$, the miner removes the transactions in  $B$ from its transaction pool and adds $B$ to the unreferred transaction block pool.
Upon receiving a valid voter block, the miner updates the longest chain if needed, and updates the vote counts accordingly.
Upon  receiving a valid proposer block $B$ at a level $\ell$ higher than the previous highest level, the miner makes $B$ the  new  parent proposer block, and updates all $m$ voter trees to vote on  $B$ at level $\ell$.

\subsection{Ledger confirmation rule $g$}
\label{sec:confirmation}
As defined before, the proposer block with the most votes on level $\ell$ is defined as the \textit{leader block} of level $\ell$. The leader block for a fixed level $\ell$ can initially fluctuate when the voter blocktrees start voting on level $\ell$. However, as the voter blocktrees grow, these votes on level $\ell$ are cemented deeper into their respective voter blocktrees and the leader fluctuation ceases and thus we can confirm the leader block at level w.h.p. 
The sequence of leader blocks for each level of the proposer blocktree is defined as the {\em leader sequence}.

\textit{Confirmation and Ordering:} A set of transactions can often be individually confirmed before being ordered among themselves. 
For this reason, confirming transactions is easier than ordering the transactions.
For example, consider the following two transactions a) Alice pays Bob 10\$, and b) Carol pays Drake 10\$. Both these transactions can be individually confirmed without deciding which transaction occurred first.
In \bitcoinnosp, transactions are simultaneously confirmed and ordered; 
however, in \schemenosp, transactions can be confirmed before being ordered.
The procedure \textsc{IsTxConfirmed()} in Algorithm \ref{alg:prism_con} defines the transaction confirmation rule $g$ and the procedure \textsc{GetOrderedConfirmedTxs()} defines the rule for ordering the confirmed transactions. Both these procedures use \textsc{BuildLedger()} which is described next.


\begin{figure}
   \centering
   \includegraphics[width=\linewidth]{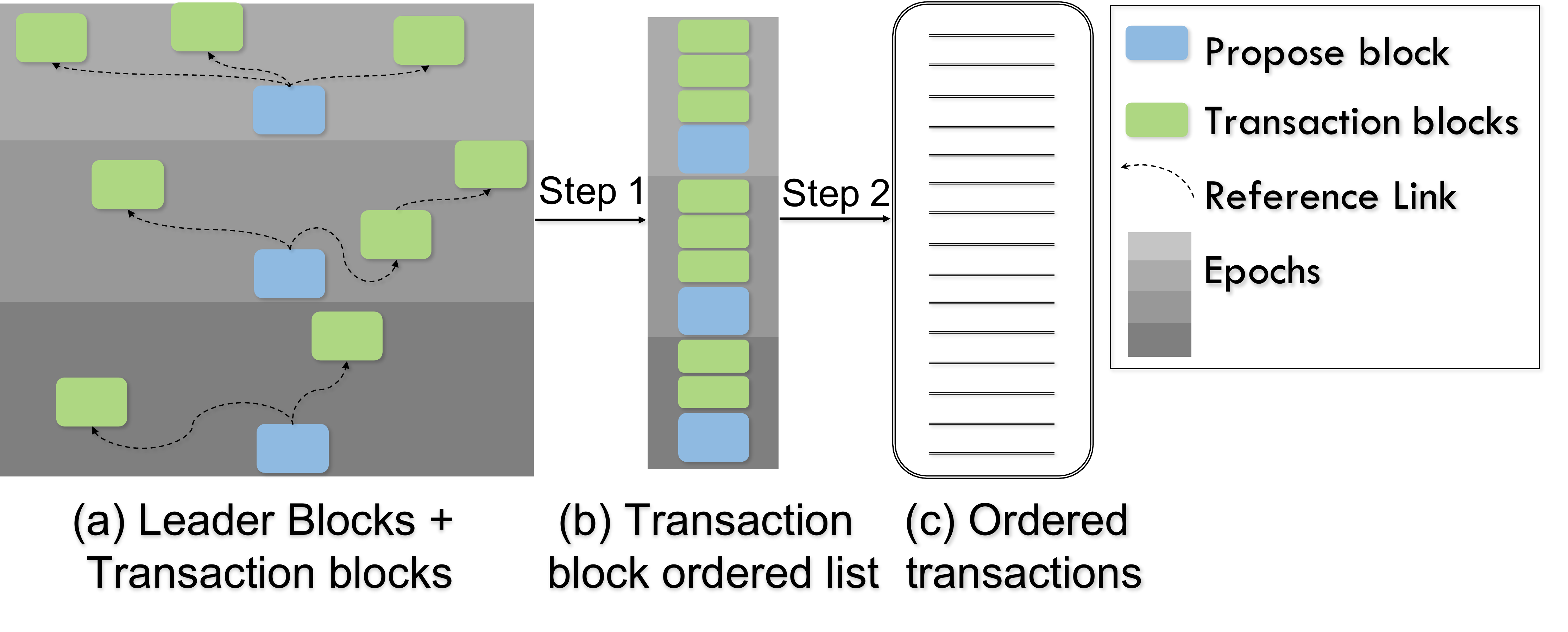}
   \caption{\textsc{BuildLedger():} The proposer blocks for a given proposer block sequence are blue, and the referenced transaction blocks are green. Each shade of gray region is all the tx blocks referred by the proposer block.
}
   \label{fig:leader_ledger}
 \end{figure}
 
\textsc{BuildLedger():} 
Given a proposer block sequence from levels $1$ to $\ell$, $\{p_1,\cdots,p_{\ell}\}$, represented by blue blocks in Fig. \ref{fig:leader_ledger}(a).
Let $L_{p_i}$, represented by green blocks in the grey shaded area in Fig. \ref{fig:leader_ledger}(a), be an ordered list of all the transaction blocks directly or indirectly referred by block $p_i$. 
Note that a transaction block $t$ is indirectly referred by proposer block $p_i$ if $p_i$ includes a reference link to another proposer block $p'$ that directly refers $t$.
Since honest proposer blocks link to any unreferenced transaction blocks and proposer blocks, this ensures that the transaction blocks not referred by the proposer leader sequence are also included in the ledger.
Let $\{L_{p_1}, \cdots, L_{p_{\ell}}\}$ be the \textit{transaction block list} of sequence $\{p_1,\cdots,p_{\ell}\}$ as shown in Fig. \ref{fig:leader_ledger}(b). The procedure then expands this transaction-block list 
and remove all the duplicate and double-spent transactions to output \textit{ordered-transaction list} as shown in Fig. \ref{fig:leader_ledger}(c). 

\textsc{IsTxConfirmed():} While confirming a leader block can take some time\footnote{In absence of an active attack, it will be fast, as described in  Section \ref{sec:analysis}.}, 
we quickly narrow down a set of proposer blocks, defined as \textit{proposer-set}, which is guaranteed to contain the leader block for that level. The proposer-set is realized using  Def. \eqref{defn:propListConfPolicy}. This procedure first gets all the votes from the voter trees and then gets the proposer-set for each 
level from the genesis to the last level for which the proposer-set can be realized (lines:\ref{code:getvotes_start}-\ref{code:appendPrpList}). It then takes the outer product of these proposer-sets and enumerates many proposer block sequences (line:\ref{code:outerproduct}). Note that by design, one of these sequences will be the leader block sequence in the future.
It then builds a ledger for each proposer block sequence and confirms the transaction if it is present in \textit{all} of the ledgers (lines:\ref{code:ledgers_for_start}-\ref{code:fastConfirmTx}).

\textsc{GetOrderedConfirmedTxs():} 
First obtain a leader block for each level on proposer blocktree from genesis up until the level which has a confirmed leader block (line:\ref{code:confirmLeader}). Then return the ledger built from these leader blocks.


\section{Analysis}
\label{sec:analysis}

In this section, we analyze three aspects of \schemenosp: security, throughput, and latency.
\blue{
Before listing the formal guarantees satisfied by Prism, we first describe at an intuitive level why Prism is able to achieve good latency without sacrificing security.}

\subsection{Intuition and Sketch of Proofs}
\label{sec:overview}

\blue{
In the longest-chain protocol, for a fixed block size and network, the maximum tolerable adversarial hash power $\beta$ is governed by the block production rate; the faster one produces blocks, the smaller the tolerable $\beta$ \cite{backbone,pss16}. 
In \schemenosp, we need to be able to tolerate $\beta$ adversarial hash power in each of the voter trees and and the proposer tree.
Hence, following the observations of \cite{backbone,pss16} each of these trees individually must operate at the same rate as a single longest-chain blocktree in \bitcoin in order to be secure. 
}

\blue{
The security of \scheme is provided by the voter trees; a proposer block is confirmed by votes which are on the longest chains of these voter trees.  
Consider a conservative confirmation policy for \schemenosp, where we wait for each vote on each voter tree to reach a confirmation reliability $1-\epsilon$ before counting it. 
This would require us to wait for each vote to reach a depth of $k(\epsilon)$ in its respective tree, where $k(\epsilon)$ denotes the confirmation depth for reliability $1-\epsilon$. 
This conservative confirmation rule immediately implies that \scheme has the same security guarantee as that of each of the voter tree, i.e. that of \bitcoinnosp. However, this rule has as poor a latency as \bitcoinnosp's. For example, for $\epsilon = 10^{-3}$ and the tolerable adversary power $\beta = 0.3$, the vote has to be $24$ blocks deep \cite{bitcoin}.
With a more intelligent transaction confirmation rule, we can do far better. The key insight is that even though each vote individually stabilizes at the same rate at \bitcoin (i.e., slowly), the \emph{aggregate} opinion can converge much faster because there are many voter trees. 
}
\begin{figure}[h]
\begin{centering}
\includegraphics[width=\linewidth]{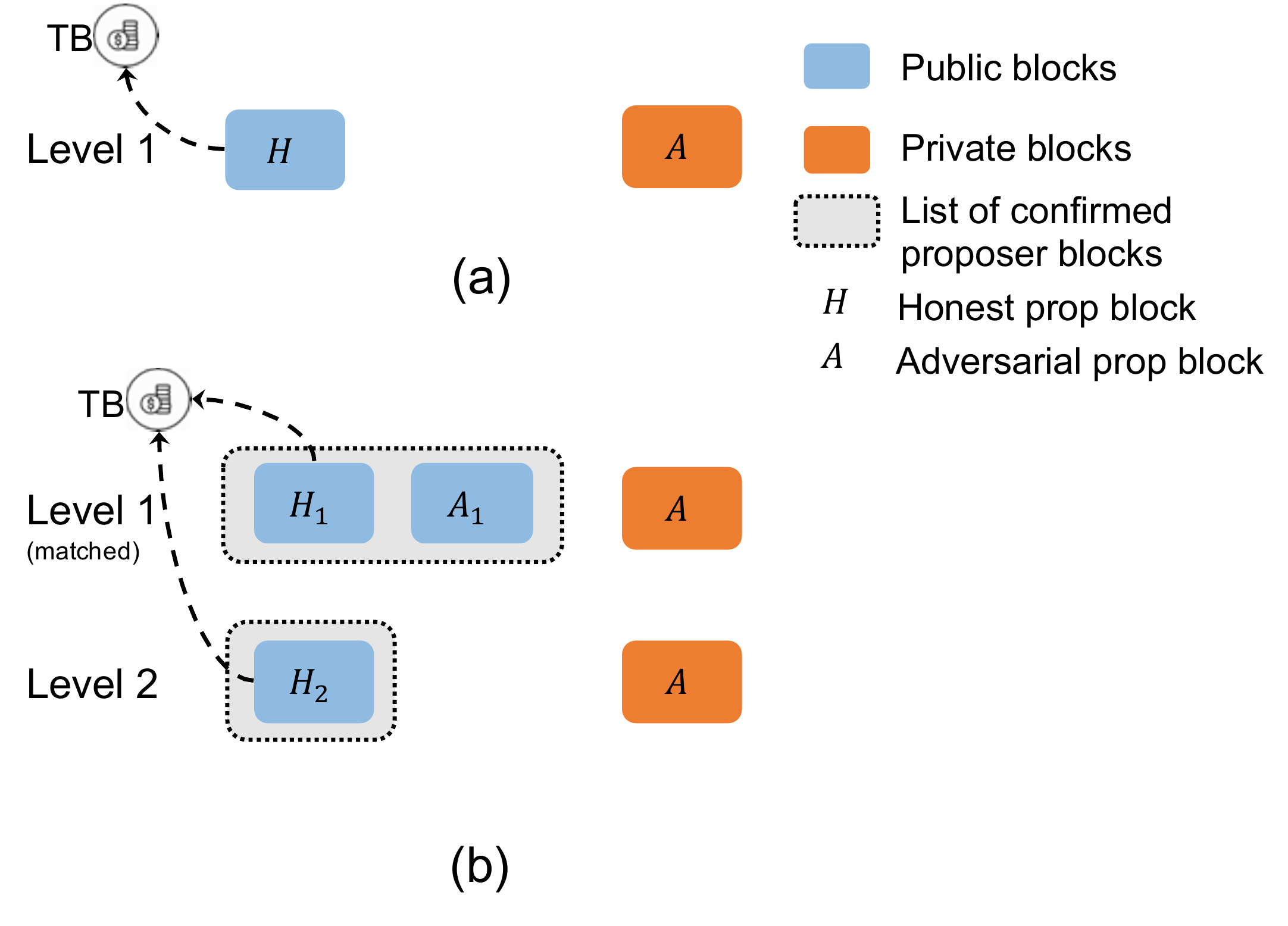}
\caption{(a) Transaction block is referred to by an isolated honest proposer block. (b) Transaction block is referred to by a non-isolated proposer block but on the next level there is an isolated proposer block. 
Note that the link from $H_2$ to \texttt{TB} is implicit; since $H_2$ is honest, it refers to all unreferenced transaction and proposer blocks, i.e., $H_1$ and $A_1$. Since $H_1$ refers \texttt{TB}, $H_2$ implicitly does too  (Section \ref{sec:confirmation})}.
\label{fig:example_1}
\end{centering}
\end{figure}
\subsubsection{Case 1: Isolated Proposer Block}
\label{sec:isolated}
\blue{
Consider first the situation when a transaction block \texttt{TB} is referred to by a honest proposer block $H$ which is currently isolated at its level, i.e. no other public proposal block exists at the same level for a certain fixed number of rounds. See Figure \ref{fig:example_1}(a). 
This case is quite common since the mining rate of the proposer blocks is chosen such that there is little forking in the proposer tree. Block $H$ will start collecting votes, each of which is on the longest chain of its respective voter tree. Over time, each of these votes will become deeper in its voter chain. An attack by the adversary is to mine a private proposal block $A$ at the same level, and on each of the voter trees fork off and mine a private alternate chain and send its vote to the block $A$. After leader block $H$ is confirmed, the adversary continues to mine on each of the voter alternate chains to attempt to overtake the public longest chain and shift the vote from $H$ to $A$. If the adversary can thereby get more votes on $A$ than on $H$, then its attack is successful.
}

\blue{
This can be viewed as the $m$-chain analog to Nakamoto's private attack on a single chain \cite{bitcoin}, where instead of having one race between the honest chain and the attack chain we have $m$ such races.  In fact, Nakamoto's calculations on the success probability of an attack on a single chain can help us determine how deep we need to wait for the votes to become to confirm the proposer block $H$. At tolerable adversary power $\beta = 0.3$, the reversal probability in a single chain is $0.45$ when a block is $2$-deep \cite{bitcoin}. With $m=1000$ voter chains and each vote being $2$-deep, the expected number of chains that can be reversed by the adversary is $450$. The probability that the adversary got lucky and can reverse more than half the votes, i.e. $500$, is about $10^{-3}$. Hence to achieve $\epsilon = 10^{-3}$, we only need to wait for $1000$ votes each $2$-deep. This incurs much shorter latency than the $24$ block depth needed for {\em each} vote to be reversed with probability $10^{-3}$. This reduction in latency is conceptually similar to averaging many unreliable  classifiers to form a strong aggregate classifier:  the more voter chains there are, the less certainty of permanence each individual vote needs to be, thereby reducing confirmation time. This gain comes without sacrificing security: each voter chain is operating slowly enough to tolerate $\beta$ adversarial hash power. 
}

\blue{Just like Nakamoto's private attack, the attack considered here is a particular attack. Our formal security analysis, sketched in Section \ref{sec:proofsketch}, consider all possible attacks in the model. In particular, the attacker can correlate its actions on the different voter chains. However, the confirmation latency behaves similarly to the latency under this attack.
}



\subsubsection{Case 2: Non-isolated Proposer Block}
\label{sec:non-isolated}
\blue{Consider now the case when the the transaction block \texttt{TB} is referred to by a honest proposal block $H_1$ which is not isolated at its level, i.e. $H_1$ is matched by an adversarial public proposer block $A_1$ (the competing proposer block could also be honest). This matching could persist for $L$ levels until reaching a level when there is an isolated honest proposal block. See Figure \ref{fig:example_1}(b) for the special case of $L=1$. 
Let us separately consider the life cycle of an honest transaction vs. a double-spent one.
}

\blue{
\textbf{Honest Transaction:}
A naive approach for confirming \texttt{TB} would be to wait until we can definitively confirm $H_1$ or $A_1$.
However, this may be slow because of adversarial attacks that try to balance votes.
A key insight is that for honest (non-double-spent) transactions, we do not need to know \emph{which} of $H_1$ and $A_1$ is confirmed---only that one of them will be confirmed.
This weaker form of \emph{list confirmation} works because if $A_1$ eventually gets confirmed, a later honest proposer block can still refer to $H_1$ and include \texttt{TB} (Section \ref{sec:confirmation}). 
To confirm an honest transaction at level $i$, we need two events: (1) list confirmation of all levels up  to $i$; (2) an isolated honest proposer at level $i$.
Once we have list-confirmed a set of proposer blocks at level $i$ referring \texttt{TB} (e.g., either $H_1$ or $A_1$ will be the leader), we know that no other block can be the leader at that level.
However, list confirmation alone is not enough for honest transaction confirmation if the transaction is not present in all ledgers. 
In that case, we also need to wait for an isolated honest proposer level, where the proposer block will implicitly or explicitly include \texttt{TB} in the ledger. 
Once this isolated honest proposer level is confirmed \emph{and} all the preceding levels are list-confirmed, we can be sure that \texttt{TB} will appear in the final ledger. 
The confirmation latency is thus the maximum of two parts:
}

\blue{
 \textit{(1) List confirmation.} We fast confirm that the adversary cannot produce a private block $A$ with more votes than the votes of public blocks $H_1$ and $A_1$. 
The logic is similar to the case of isolated honest proposer block discussed above, viewing the situation as a race between honest nodes voting for the public blocks $H_1$ or $A_1$ and adversary voting for $A$. 
Adversarial actions (e.g., presenting first $H_1$ to half the honest nodes and $A_1$ to the other half) can cause the number of votes to be evenly split between $H_1$ and $A_1$, which can slow down list confirmation, albeit not significantly.
}

\blue{
\textit{(2) Isolated honest proposer level.}
In Figure \ref{fig:example_1}(a), if we wait until level $2$, we see an isolated public proposer block $H_2$ which can be fast confirmed (Section \ref{sec:isolated}). 
At this point, we know that the final leader sequence at levels $1,2$ is either $H_1,H_2$ or $A_1,H_2$, both of which contain our honest transaction since $H_2$ refers to all previous unreferred proposer blocks. 
Since isolated honest proposer blocks happen frequently (Section \ref{sec:proofsketch}), this step is fast. 
}

\blue{
\textbf{Double-Spent Transaction:}
To confirm double-spent transactions, we need stronger conditions than those listed above: namely, instead of list confirmation, we need \emph{unique block confirmation}, confirming which block at a proposer level will be the ultimate leader.
This is achieved once list confirmation occurs \emph{and} one of the list-confirmed blocks can be reliably declared the winner. 
If one of the public proposer blocks $H_1$ or $A_1$ gathers many more votes than the other block, then we can fast confirm a unique leader, even for double-spent transactions; 
this happens both in the absence of active attacks and under some classes of attacks (Section \ref{sec:sims}).
However, other adversarial attacks (such as balancing the votes on $H_1$ and $A_1$) can cause the number of votes to be evenly split between $H_1$ and $A_1$, so we cannot fast confirm a leader block.
In this case, we must wait until every vote on $H_1$ and $A_1$ stabilizes, in which case either $H_1$ or $A_1$ is confirmed and only one of the double-spent transactions is accepted. 
A content-dependent tie breaking rule can be used to break ties after votes are stabilized.
}
\subsubsection{Sketch  of Security and Latency Proofs}
\label{sec:proofsketch}
\begin{figure}[H]
\begin{centering}
\includegraphics[width=\linewidth]{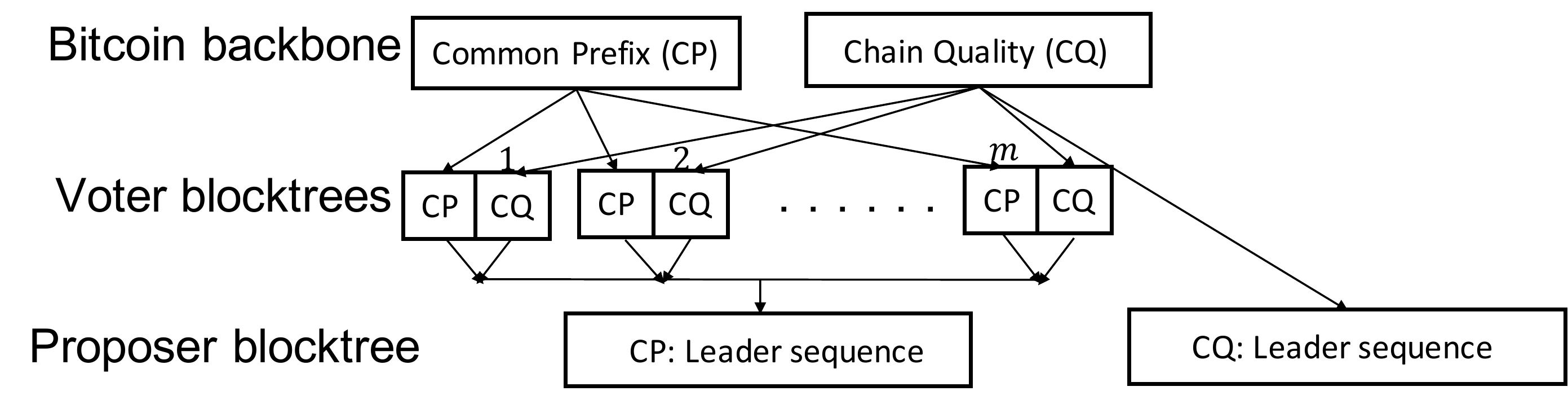}
\caption{Common-prefix and chain-quality properties of voter chains imply common-prefix and chain-quality properties of the proposer leader sequence.}
\label{fig:proof}
\end{centering}
\end{figure}

\blue{
To translate the above intuitive arguments into formal security and latency proofs, 
we borrow key ideas from \cite{backbone}, but also require several  new insights. 
\cite{backbone} proves the consistency and liveness of the Bitcoin backbone protocol by first proving two key properties: common-prefix and chain-quality. Similarly, to show that \scheme achieves consistency and liveness, we need to show that the proposer leader sequence satisfies these properties. 
The results of \cite{backbone} do not directly apply because the proposer leader sequence is not determined by a single longest-chain protocol; rather, it is determined by a combination of the proposer tree and the aggregate voter tree votes. As shown in Figure \ref{fig:proof}, we prove the two properties for the proposer and the voter trees and use them to prove corresponding properties for the leader sequence. Specifically:
}

\blue{
    (1) Each voter tree is constructed according to the backbone protocol, and hence satisfies the chain-quality and common-prefix property. 
    Chain-quality  of the voter trees implies that honest nodes will continually be able to vote for proposer blocks from every voter tree and at every proposer level . Common-prefix implies that all these votes will eventually stabilize.  This further implies that the leader sequence satisfies the common-prefix property (Theorem \ref{thm:common_prefix_1}), since the leader block at each level will eventually stabilize. Hence, the resulting ledger is consistent. The leader-sequence also can be shown to have a certain chain quality (Lemma~\ref{lemma:ls_slow_quality}) and this ensures liveness of the ledger (Theorem \ref{thm:liveness}).
    }
    
    \blue{
    (2) To show fast confirmation of all honest transactions, we follow the intuitive arguments above. We first show that an isolated proposer block, or an honest proposer block that does not have a competing adversarial proposer block for a certain duration of time, appears in constant expected time (independent of $\epsilon)$. 
    Specifically, the honest users are mining proposer blocks at the rate $(1-\beta)\bar{f}_p$ whereas the adversary is mining at  rate $\beta\bar{f}_p$.
Since $\beta < 0.5$, the adversary is mining slower than the honest users, and within the next 
$\frac{1}{1-2\beta}$ levels in expectation, there is a level 
on which the adversary cannot \textit{immediately} create a competing block with the honest block \footnote{Random walk analysis}. Similarly, an isolated level on which the adversary cannot match the honest block for next $R$ rounds after the honest block is mined happens within $\frac{1+2R\bar{f}_v}{1-2\beta}$ levels in expectation. 
}

\blue{
(3) We next show that we can fast confirm an isolated public honest proposer block. The argument has two parts: i) the isolated honest block wins enough votes; 2) the leader block persists, i.e., wins the vote race against a private adversarial proposer block for all time. The first part follows from the chain-quality of the voter chains, which ensures that there is a steady stream of honest votes for the public proposer block until it gathers a sufficiently large fraction of total votes (Lemma~\ref{lem:fast_latency_1}). The second part follows from common-prefix of the voter trees, which ensures that a large fraction of the votes cannot be reversed by the adversary (Lemma~\ref{lem:List_decoding_level_1}).
}
\blue{
(4) Fast {\em list} confirmation of proposer blocks at all previous levels can be proved similarly (see Lemma~\ref{lem:conf_upto_level} and Theorem~\ref{thm:common_prefix_2}). Now, \scheme  ensures that at each proposer level, one of the list-confirmed blocks will remain in the ledger. This, combined with the assurance that every transaction will be either directly or indirectly referred by the isolated proposal block, ensures that all honest transactions are entered into the ledger. This lets \scheme confirm honest transactions within a short time (see Theorem~\ref{thm:honest_tx_latency}). 
}

\blue{
Note that \cite{backbone} proves the $k$-common-prefix property is satisfied with high probability only for large $k$. 
Similarly,  chain-quality is shown to be satisfied with high probability only over a large number of blocks. While this is sufficient to prove (1) and (2) above for the consistency and liveness of the eventual ledger, it is not sufficient to prove (4) and (5) for fast confirmation latency, since we need these two properties over short latencies, i.e. windows of few blocks. In these small time windows, these properties do not hold with high probability {\em microscopically}, for every {\em individual} voter tree. However, since the proposer leader block depends only on the {\em macroscopic} vote counts, we only need to show that these properties hold with high probability macroscopically, for a good fraction of  voter trees.
}

\subsection{Parameter Selection}
We first specify the parameters of \scheme in terms of the parameters of the physical network.  
First, recall that the network delay of a block containing $B$ transactions is given  by $\Delta=\frac{B}{C}+D$.
Let $B_t$, $B_v$, and $B_p$ be the size of transaction, voter, and proposer blocks respectively,  in units of number of transactions. 
The network delays $\Delta_t,\Delta_v$, and $\Delta_p$ for each type of block are thus given by:
\begin{equation} 
\label{eq:delay}
\Delta_t  = \frac{B_t}{C} + D, \quad \Delta_v  = \frac{B_v}{C} + D,  \quad \Delta_p  = \frac{B_p}{C} + D.
 \end{equation}

Given that different block types  have different sizes and  network delays, what is a reasonable choice for $\Delta$, the duration of a round?
Since the synchronous model is used for security analysis, and the security of \scheme depends only on the network delay of the proposer and voter blocks but not of the transaction blocks, we choose: $ \Delta = \max \{ \Delta_p, \Delta_v\}.$
Moreover, the voter blocks and the proposer blocks contain only reference links and no transactions,  so their sizes are expected to be small.  Assuming the bandwidth-delay product $CD/\max\{B_v,B_p\} \gg 1$, we have that the network delay $\Delta = \max\{\frac{B_v}{C}, \frac{B_p}{C}\} + D \approx D$,
the smallest possible.




To provide security, we set the mining rate $\bar{f_v} := f_v D$ on each voter tree such that each voter tree is secure under the longest chain rule. According to \cite{backbone} it should satisfy  
    \begin{equation}
    \label{eq:voter-security}
        \bar{f_v} < \frac{1}{1-\beta}\log \frac{1-\beta}{\beta}.
    \end{equation} 
    We also set the proposer and voter mining rates to be the same, i.e. $f_p = f_v$.  This is not necessary but simplifies the notation.
    
    Third, to utilize $90\%$ of the communication bandwidth for carrying transaction blocks, we set $f_t B_t = 0.9 C$.
    The individual choices of $f_t$ and $B_t$ are not very important, but choosing large $B_t$ and small $f_t$ is preferable to ensure that the number of reference links to transaction blocks per proposer block is small, thus giving a small proposer block size $B_p$.
    
    Finally, speed up voting, we maximize the number of voter chains subject to the \emph{stability constraint} of Sec. \ref{sec:model}:
    $ f_p B_p + m f_v B_v + f_t B_t < C.$
    Substituting the  values of $f_v, f_p$ and $f_tB_t$, we get
           \begin{eqnarray}
            m \;=\frac{0.1CD}{\bar{f_v}B_v} - \frac{B_p}{B_v} \nonumber  \;\ge   \frac{(1- \beta)}{\log (\frac{1-\beta}{\beta} )} \cdot \frac{CD}{B_v} - \frac{B_p}{B_v}
                        \label{eqn:choice_of_m}
        \end{eqnarray}
    i.e. the number of voter trees is at least proportional to $CD/B_v$, the bandwidth-delay product in unit of voting blocks.  This number is expected to be large, which is a key advantage.
    The only degree of freedom left is the choice of $\bar{f_v}$, subject to  \eqref{eq:voter-security}. We will return to this issue in Section \ref{sec:fast_latency} when we discuss  fast confirmation latency. 

\subsection{Total Ordering}
\label{sec:total}

In this subsection, we  show that \scheme can achieve total transaction ordering for any $\beta < 0.5$ using the procedure \textsc{GetOrderedConfirmedTxs()} in Algorithm \ref{alg:prism_con}. That is, as long as the adversary's  hash power is less than $50\%$, transactions can be ordered with consistency and liveness guarantees.  
Following  \cite{backbone}, we  do so by first establishing two backbone properties: common-prefix and chain quality of the {\em proposer leader sequence}.
Let $\P(r)$ denote the set of proposer blocks 
mined by round $r$. Let $\P_{\ell}(r) \subseteq \P(r)$ denote the set of proposer blocks mined on level $\ell$ by round $r$. 
Let the first proposer block on level $\ell$ be mined in round $R_{\ell}$.
Let $V_p(r)$ denote the number of votes on proposer block $p \in \P(r)$ at round $r$. Recall that only votes from  the main chains of the voter trees are counted. The \textit{leader block} on level $\ell$  at round $r$, denoted by $p^*_{\ell}(r)$, is the proposer block with maximum number of votes in the set $\P_{\ell}(r)$ i.e,
$
p^*_{\ell}(r)  := \underset{p \in \P_{\ell}(r)}{\text{argmax}}\;V_p(r),
$
where tie-breaking is done in a hash-dependent way.
The \textit{leader sequence}  up to level $\ell$ at round $r$, denoted by  $\texttt{LedSeq}_{\;\ell}(r)$ is:
\begin{equation}
\texttt{LedSeq}_{\;\ell}(r) := [p^*_{1}(r), p^*_{2}(r), \cdots, p^*_{\ell}(r)]. \label{eqn:leader_block_seq}    
\end{equation}

The leader sequence at the end of round $\rmax$, the end of the horizon,  is the \textit{final leader sequence}, $\texttt{LedSeq}_{\;\ell}(\rmax)$.


\begin{thm}[Leader sequence common-prefix property] \label{thm:common_prefix_1}
Suppose $\beta < 0.5$. For a fixed level $\ell$, we have 
\begin{equation}
    \texttt{LedSeq}_{\;\ell}(r) = \texttt{LedSeq}_{\;\ell}(\rmax) \quad \forall r\geq R_{\ell}+r(\eps)
\end{equation}
with probability  $1-\eps$, where $r(\epsilon)
= \frac{1024}{\fv(1-2\beta)^3 }\log \frac{8m\rmax}{\epsilon}$,
and $R_{\ell}$ is the round in which the first proposer block on level $\ell$ was mined. 
\end{thm}
\begin{proof}
See  Appendix \ref{sec:slow_latency_app}.
\end{proof}

\begin{thm}[Liveness] \label{thm:liveness}
Assume $\beta < 0.5$. Once a transaction enters into a transaction block, w.p $1-\eps$ it will eventually be pointed to by a permanent leader sequence block  after  $$O\left(\frac{\log(1/\epsilon)}{(1-2\beta)^4}\right) \;\;
\text{rounds}.$$  
\end{thm}
\begin{proof}
See Appendix \ref{sec:slow_latency_app}.
\end{proof}


Theorems \ref{thm:common_prefix_1} and \ref{thm:liveness} yield the following:
\begin{thm} \label{cor:latency_ordering}
 The ledger has consistency and liveness with the expected $\epsilon$-\emph{latency} for all transactions (Def. \ref{def:latency}) to be at most $O\left(\log(1/\epsilon)\right)$ for $\beta<0.5$.
\end{thm}

\begin{figure}
\begin{centering}
\includegraphics[width=\linewidth]{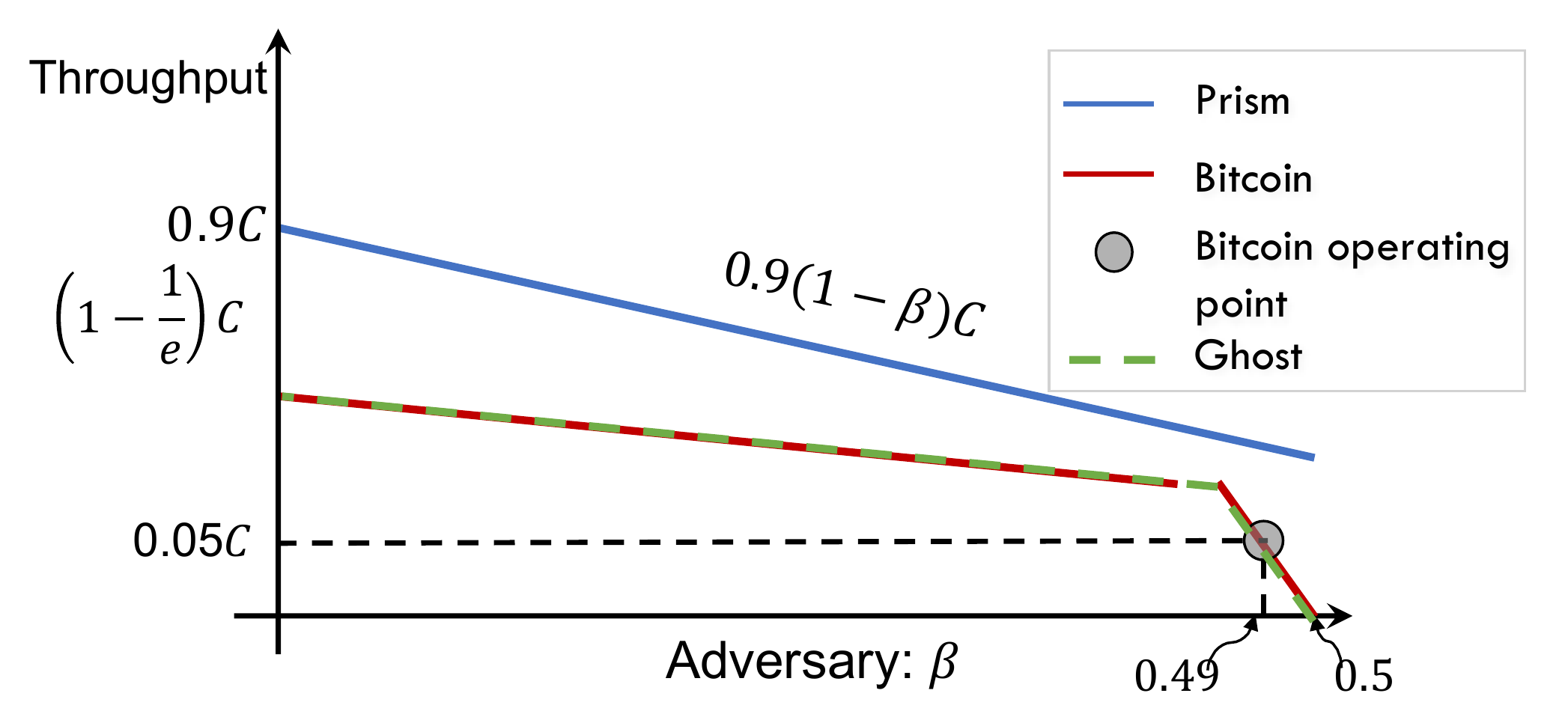}
\caption{Throughput versus $\beta$ tradeoffs of \schemenosp, \bitcoin and \ghostnosp. The tradeoffs for the baseline protocols are upper bounds, while that for \scheme is exact.}
\label{fig:throughput_comparision}
\end{centering}
\end{figure}
\subsection{Throughput}

We now analyze the transaction throughput of \schemenosp. The leader sequence blocks of \scheme orders all the transactions in the transaction blocks they refer to. Due to liveness, all transaction blocks are referred to by some proposer blocks. Since the transaction block generation rate $f_t B_t$ is chosen to be $0.9C$ transactions per second, assuming a worst case that only honest blocks carry transactions yield a throughput of $0.9 (1-\beta) C$ transactions per seconds. 

This seems to give the advertised goal, but there is a catch: blocks mined in the same round may contain the same transactions, since transactions are broadcast to the entire network. 
To achieve the full throughput, one can minimize the transaction redundancy in the blocks by {\bf scheduling} different transactions to different blocks. Concretely, we split the transactions randomly into $q$ 
queues, and each honest block is created from transactions drawn from one randomly chosen queue. 
Thinking of each transaction queue as a color, we have transaction blocks of $q$ different colors.   

	We will only have honest blocks with redundant transactions if two or more blocks of the same color are mined in the same round.  The number of honest blocks of the same color mined at the same round is distributed as Poisson with mean $(1-\beta)f_t \Delta/q$, and so the throughput of non-redundant blocks of a given color is the probability that at least one such block is mined in a round, i.e.
 $1 - e^{-(1-\beta)f_t \Delta/q} \;\; \mbox{blocks per round}.$
The total throughput of non-redundant honest blocks of all colors is
\beq
\label{eq:prism_thruput}
q \left [1 - e^{-(1-\beta)f_t\Delta/q}\right] \; \;  \mbox{blocks per round}.
\eeq
For large $q$, this approaches 
$ (1-\beta) f_t \Delta \;\;\mbox{blocks per round},$
which equals 
$0.9 (1-\beta)C$ transactions per second when we set $f_t = 0.9 C/B_t$. Thus, we achieve the claimed result \eqref{eq:thruput_result}.
Note that the throughput as a fraction of capacity does not vanish as $\beta \rightarrow 0.5$, unlike \bitcoin and \ghost (App. \ref{app:btc_thruput} and \ref{app:ghost}).


\begin{thm}[Throughput] \label{cor:throughput}
Theorem \ref{cor:latency_fast} guarantees that all the transactions proposed before round $r-O\left(\log(1/\epsilon)\right)$
are confirmed by round $r$. Therefore \scheme can support $\lambda=0.9(1-\beta)C$ throughput (Def. \ref{def:throughput}), where $U_\epsilon=O\left(\log(1/\epsilon)\right)$.
\end{thm}


\subsection{Fast confirmation latency}
\label{sec:fast_latency}

\subsubsection{List confirmation latency}
\label{subsubsec:fast_list_conf}
We convert the intuition from Section \ref{sec:overview} to a formal rule for fast confirming a {\em set} of proposer blocks, which enables confirming  a list of proposer sequences. The idea is to have {\em confidence intervals} around the number of votes cast on each proposer block. Figure \ref{fig:confidence} gives an example where there are $5$ proposal blocks in public at a given level, and we are currently at round $r$. The confidence interval $[\underbar{V}_n(r), \overline{V}_n(r)]$ for the votes on proposer block $p_n$ bounds the maximum number of votes the block can lose or gain from uncast votes and votes reversed by the adversary. 
We also consider a potential private proposer block, with an upper bound on the maximum number of votes it can accumulate in the future. We can fast confirm a set of proposal blocks whenever the upper confidence bound of the private block is below the lower confidence bound of the public proposal block with the largest lower confidence bound. 

\begin{figure}
  \centering
  \includegraphics[width=\linewidth]{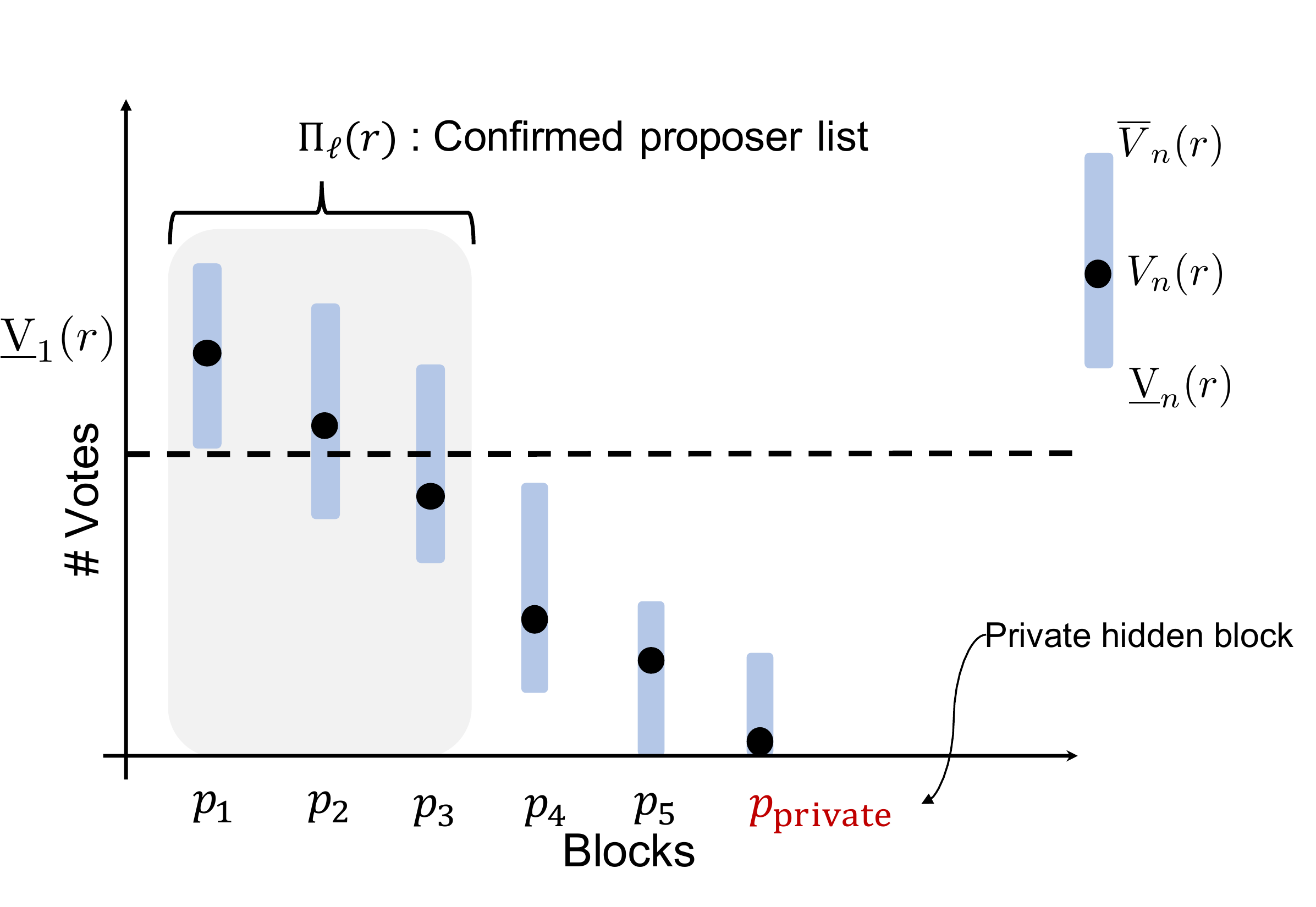}
  \caption{Public proposer block $p_1$ has the largest lower confidence bound, which is larger than the upper confidence bound of the private block. So list confirmation is possible and the set confirmed is $\Pi_{\ell}(r) =\{p_1,p_2,p_3\}$.} 
  \label{fig:confidence}
 \end{figure}
More formally: As defined earlier, $\P_{\ell}(r) = \{p_1,p_2...\}$ is the set of proposer blocks at level $\ell$ at round $r$.
Let $V^d_n(r)$ be the number of votes at depth $d$ or greater for proposer block $p_n$ at round $r$. Define:
\begin{align*}
\delta_d := \max\left( \frac{1}{4\bar{f_v}d}, \frac{1-2\beta}{8\log m} \right), \epsilon' = 1-\rmax^2 e^{-\frac{(1-2\beta)m}{16\log m}},\end{align*}
\begin{align*}\underbar{V}_n(r) := \max_{d \geq 0} \left (V_n^{ d}(r) -2\delta_d m \right )_{+},\end{align*}

\begin{align*}    \overline{V}_n(r) := m - \sum_{p_{n'} \in \mathcal{P}_{\ell}(r) \setminus \{p_n\}} \underbar{V}_{n'}(r),
\end{align*}

\begin{align*}
    \underbar{V}_{\text{private}}(r) & := 0,\quad
    \overline{V}_{\text{private}}(r) := m-\sum_{p_{n'} \in \mathcal{P}_{\ell}(r)} \underbar{V}_{n'}(r).
\end{align*}
\begin{defn}\label{defn:propListConfPolicy}
{\bf Proposer set confirmation policy}: If 
\begin{equation}
     \max_n \underbar{V}_n(r) > \overline{V}_{\text{private}}(r), \label{eqn:fast_confirmation_policy}
\end{equation}
then we confirm the the set of proposer blocks $\Pi_{\ell}(r)$, where
\begin{equation}\label{eqn:proposerlist}
\Pi_{\ell}(r) := \{p_n: \overline{V}_n(r) > \max_i \underbar{V}_i(r) \}.
\end{equation}
\end{defn}

Next, we show that we can confirm proposer sets up to level $\ell$ with an expected latency independent of $\eps$, and the final leader sequence is contained in the outer product of the confirmed sets. 

 
\begin{thm}[List common-prefix property ] \label{thm:common_prefix_2}
Suppose $\beta < 0.5$. Suppose the first proposer block at level $\ell$ appears at round $R_\ell$. Then w.p.  $\epsilon'$, we can confirm proposer sets $\Pi_1(r)), \ldots, \Pi_\ell(r)$ for all rounds  $r \geq R_\ell + R_\ell^{\text{conf}}$, where 
\begin{equation}
    \mathbb{E}[R_{\ell}^{\text{conf}}] \leq  \frac{2808}{(1-2\beta)^3\fv}\log \frac{50}{(1-2\beta)}+\frac{256}{(1-2\beta)^6\fv m^2}, \label{eq:list_latency}
\end{equation}
and
$
     p^*_{\ell'}(\rmax)\in \Pi_{\ell'}(r) \quad \mbox{$\forall \ell' \le \ell$ and $ r\geq R_{\ell}+R_{\ell}^{\text{conf}}$}.
$
\end{thm}
\begin{proof}
See Appendix \ref{sec:fast_list_app}.
\end{proof}

Let us express the latency  bound \eqref{eq:list_latency} in terms of physical parameters. If we set the voting rate $\bar{f_v}$ equal to the largest possible given the security constraint \eqref{eq:voter-security}:
$\bar{f_v} = \frac{1}{1-\beta} \log \frac{1-\beta}{\beta},$
then according to \eqref{eqn:choice_of_m}, we have
$$ m = \frac{0.1(1- \beta)}{\log (\frac{1-\beta}{\beta} )} \cdot \frac{CD}{B_v} - \frac{B_p}{B_v}.$$
With this choice of parameters, and in the regime where the bandwidth-delay product $CD/B_v$ is large so that the second term in \eqref{eq:list_latency} can be neglected, the expected latency for list confirmation is bounded by 
$ c_1(\beta) D \quad \mbox{seconds},$
i.e. proportional to the propagation delay. Here, 
$$ c_1(\beta) : = \frac{2808(1-\beta)}{(1-2\beta)^3\log\frac{1-\beta}{\beta}}\log \frac{50}{(1-2\beta)}$$
and is positive for $\beta < 0.5$. 
The confirmation error probability is exponentially small in $CD/B_v$. This is the constant part of the latency versus security parameter tradeoff of \scheme in Fig. \ref{fig:fig1}.
Since $CD/B_v$ is very large in typical networks, a confirmation error probability exponentially small in $CD/B_v$ is already very small. To achieve an even smaller error probability $\epsilon$ 
we can reduce the voting rate $\bar{f_v}$ smaller below the security constraint \eqref{eq:voter-security} and increase the number of voter chains. More specifically, we set 
\begin{equation}
    \label{eq:fv}
    \bar{f_v} = \frac{0.1CD}{B_v\log \frac{1}{\epsilon}},
\end{equation}
resulting in 
$ m = \log \frac{1}{\epsilon} - \frac{B_p}{B_v} \approx \log \frac{1}{\eps}$,
yielding the desired security parameter. 
Again neglecting the second term in \eqref{eq:list_latency}, the corresponding latency bound is
$$ \frac{c_2(\beta)B_v}{C} \log \frac{1}{\epsilon} \quad \mbox{seconds},$$
where
 $c_2(\beta) := \frac{54000}{(1-2\beta)^3}\log \frac{50}{(1-2\beta)}$.
This is the linearly increasing part of the \scheme  curve in Figure \ref{fig:fig1}, with slope inversely proportional to the network capacity $C/B_v$. 
\begin{figure}
   \includegraphics[width=\linewidth]{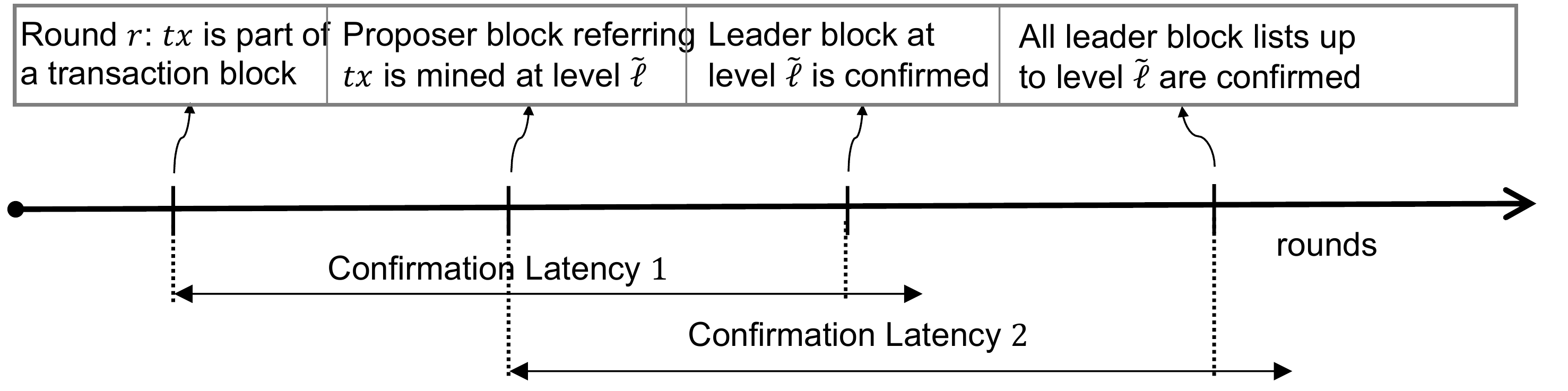}
   \caption{Components of the latency: a) Confirmation latency $1$ is analyzed in Theorem \ref{thm:honest_tx_latency}, and b) Confirmation latency $2$ is analyzed in Theorem \ref{thm:common_prefix_2}.
   \label{fig:latency_timeline}
   }
 \end{figure}
\subsubsection{Fast confirmation of honest transactions} 
In the previous subsection we have shown that one can fast confirm a set of proposer block sequences which is guaranteed to contain the prefix of the final totally ordered leader sequence. As discussed in Section \ref{sec:protocol}, each of these proposer block sequence creates an ordered ledger of transactions using the reference links to the transaction blocks. In each of these ledgers, double-spends are removed to sanitize the ledger. 
If a transaction appears in {\em all} of the sanitized ledgers in the list, then it is guaranteed to be in the final total ordered sanitized ledger, and the transaction can be fast confirmed. 
All honest transactions without double-spends eventually have this {\em list-liveness} property; 
when only a single honest proposer block appears in a level and becomes the leader, it adds any honest transactions that have not already appeared in at least one ledger in the list. 
Due to the positive chain-quality of the leader sequence (Theorem \ref{thm:liveness}), an isolated honest level eventually occurs. The latency of confirming honest transactions is therefore bounded by the sum of the latency of list confirmation in Theorem \ref{thm:common_prefix_2} plus the latency of waiting for this event to occur (Fig. \ref{fig:latency_timeline}). The latter is given by:

\begin{thm}[List-liveness] \label{thm:honest_tx_latency}
Assume $\beta < 0.5$. If an honest transaction without double spends is mined in a transaction block in round $r$, then w.p. $1-\rmax^2e^{-\frac{m}{16\log m}}$ it will appear in all of the  ledgers corresponding to proposer block sequences after an expected latency no more than 
$$  \frac{2592}{(1-2\beta)^3\fv}\log \frac{50}{(1-2\beta)} \quad \mbox{rounds}.$$
\end{thm}
\begin{proof}
Appendix \ref{thm:honest_tx_latency_proof}.
\end{proof}

 Figure \ref{fig:latency_timeline} shows the various components of the overall latency we analyzed. We can see that the confirmation latency from the time an honest transaction enters a blocks to the time it is confirmed is bounded by the sum of the latencies in  Theorem \ref{thm:common_prefix_2} and \ref{thm:honest_tx_latency}. 
 
 Repeating the analysis of Thm. \ref{cor:latency_ordering}, we get the following:
 \begin{thm}[Latency] \label{cor:latency_fast}
Theorems \ref{thm:common_prefix_2} and \ref{thm:honest_tx_latency} guarantee that the expected $\epsilon$-\emph{latency} for all \textbf{honest} transactions (Def. \ref{def:latency}) is at most $r(\beta)$ rounds for $\beta<0.5$, where 

 $$ r(\beta) := \max \left\{ c_1(\beta), c_2(\beta)\frac{B_v}{DC} \log \frac{1}{\epsilon} \right\},$$
 where
 \begin{eqnarray*}
 c_1(\beta) & := &  \frac{5400(1-\beta)}{(1-2\beta)^3\log\frac{1-\beta}{\beta}}\log \frac{50}{(1-2\beta)} \label{eq:a1}\\
 c_2(\beta) & := &  \frac{54000}{(1-2\beta)^3}\log \frac{50}{(1-2\beta)},\label{eq:a2}
\end{eqnarray*}
\end{thm}

\noindent Therefore the honest transactions are confirmed in
$$ \max \left\{ c_1(\beta)D,\; c_2(\beta)\frac{B_v}{C} \log \frac{1}{\epsilon} \right\} \quad \mbox{seconds.}$$

\begin{figure*}[!htb]
\minipage{0.32\textwidth}
\includegraphics[width=\linewidth]{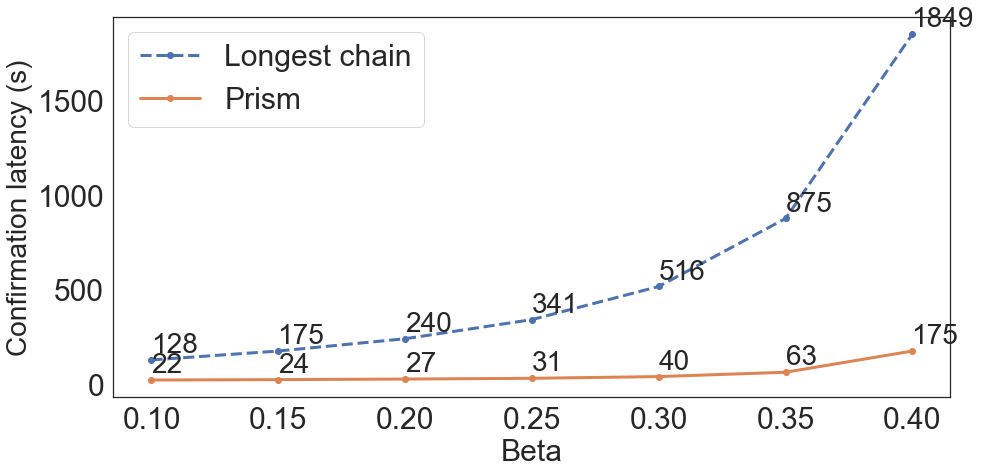}
\put(-80,-8){(a)}
\endminipage\hfill
\minipage{0.32\textwidth}
\includegraphics[width=\linewidth]{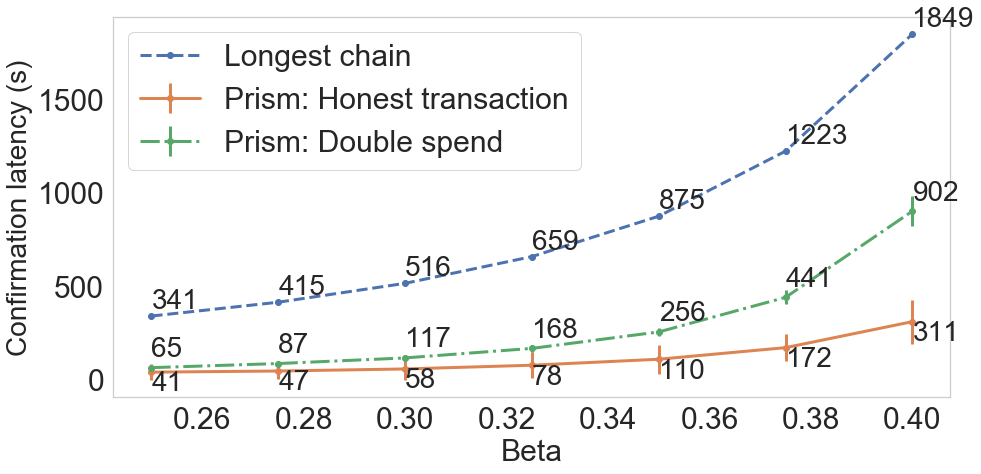}
\put(-80,-8){(b)}
\endminipage\hfill
\minipage{0.32\textwidth}%
\includegraphics[width=\linewidth]{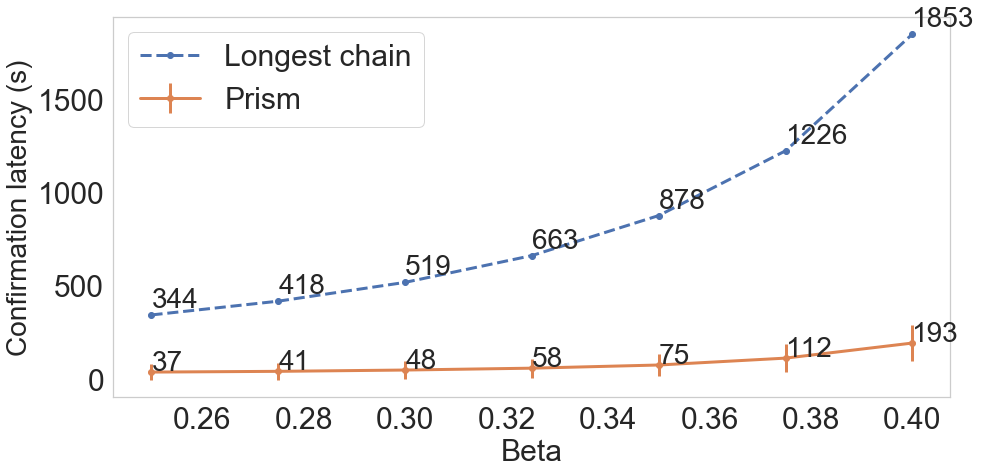}
    \put(-80,-8){(c)}
\endminipage
\caption{\textbf{(a)} Confirmation latency of honest transactions with no attack. 
  The x-axis denotes the maximum tolearble fraction of adversarial hash power $\beta$.
  \textbf{(b)} Transaction latency in the presence of an adversarial balancing attack from $\tilde \beta=0.25$ active hash power,
    for honest and double-spent transactions.
  \textbf{(c)} Confirmation latency under a censorship attack with $\tilde \beta=0.25$ hash power.
    Honest and double-spent transactions have the same latency, both for \scheme and for longest chain.
  }
  \label{fig:simulation}
\end{figure*}
 
\section{Simulations}
\label{sec:sims}

\blue{
Theorem \ref{cor:latency_fast} provides a theoretical upper bound on the expected latency, which matches the physical limit of propagation time up to constant factors. 
Characterizing the exact constants 
is an interesting research direction, but outside the scope of this paper. 
  On the other hand, one can empirically estimate the average latency values by simulating the \scheme protocol and its confirmation rule. 
  The purpose of this section is to conduct such a simulation in the honest setting as well as a variety of adversarial settings.
}

\blue{
\textbf{Setup.} We simulate a network with $m=1,000$ voter chains, in which $D\approx \Delta = 1$ sec. 
We run our proposer tree and each voter tree at a rate of $\bar{f}=1$ block $/ 10$ sec.
Our simulations measure the latency for transaction confirmation under three scenarios: no attack, a balancing attack, and a censorship attack.
By design, our confirmation rule is simultaneously robust against the common private Nakamoto attack \cite{bitcoin}, where the adversary withholds a proposer block as well as corresponding forked voter blocks in order to reverse a confirmed proposal block.
In this section, we show figures for an adversary deploying $\tilde \beta=0.25$ fraction of total hash power, where $\tilde \beta$ denotes the fraction of hash power being actually used for the attack (whereas $\beta$ is the maximum tolerable fraction of adversarial hash power, without losing consistency and liveness). 
We set the confirmation reliability conservatively at $\epsilon=e^{-20}$.
Experiments for additional parameter settings can be found in Appendix \ref{app:simulations}. 
We compare against the longest-chain protocol, for the same block generation rate of 1 block per 10 seconds.
}

\textbf{No Attack.}
\label{sec:honest_adv}
\blue{
We start by considering a setting where \schemenosp's parameters are chosen to withstand an attacker of hash power $\beta$, but the adversary is not actively conducting any attack. 
Since the confirmation rule must still defend against $\beta$ adversarial hash power, latency depends on $\beta$. 
Honest nodes vote on the earliest-seen proposer block, with results shown in Figure \ref{fig:simulation}(a). 
In \bitcoinnosp, a confirmed transaction has to be deeper in the chain  for larger $\beta$; in \schemenosp, the voter blocks have to be deeper.  
We see that \schemenosp's latency is significantly smaller than that of Nakamoto's longest chain protocol, and much closer to the physical limit. 
Note that since there is no active adversary, 
double-spend transactions can be resolved with the same latency as honest transactions. 
}

\textbf{Balancing Attack.} 
\label{sec:balancing}
\blue{
In a balancing attack, the goal of the adversary is to prevent confirmation by casting all of its votes so as to compete with the current proposer leader block. 
We begin this attack with two competing proposer blocks at the same level (say level 0), $A$ and $B$.  
Consider an honest (non-double-spent) transaction that is referred by at least one of the two proposer blocks.
The adversary's goal is to prevent the system from confirming this transaction by balancing votes on the two proposer blocks. 
That is, if block $A$ currently has the majority of votes and the adversary mines a voter block in the $i$th voter tree:
    (1) If voter tree $i$ has not yet voted on level 0, the adversary votes on  the minority block, $B$.
    (2) If voter tree $i$ voted on level 0 for block $B$, the adversary appends its block to the longest chain, thereby reinforcing the vote for the losing proposer block. 
    (3) If voter tree $i$ voted on level 0 for block $A$, the adversary tries to fork the $i$th voter tree to vote for $B$ instead. If there is no vote for $B$ in the voter tree, the adversary creates one. If there is already a fork  voting for $B$, the adversary appends to this fork.
The balancing attack is one of the  most severe and natural attacks on \schemenosp.
The results of this simulation are shown in Figure \ref{fig:simulation}(b). 
Notice that the latency of honest transaction confirmation increases by a factor of about 2x under a balancing attack, but does not affect the longest-chain protocol. 
Despite this, \schemenosp's latency is still far lower than that of the longest-chain protocol.
}

\blue{
Next, we consider double-spent transactions.
The latency for double-spent transactions is the same as honest transactions in the longest-chain protocol, so the blue curve does not change. 
However, the double-spent transaction latency for \scheme grows substantially, approaching that of the longest-chain protocol. 
Indeed, as the active $\tilde \beta$ fraction approaches 0.5, \schemenosp's latency on double-spent transactions in the presence of attacks on the confirmation process actually exceeds that of the longest-chain protocol, as discussed in Section \ref{sec:overview} and seen in Figures \ref{fig:active_e_10} and \ref{fig:varying_epsilon}. 
}

\textbf{Censorship Attack}
\blue{
Finally, we consider an attacker whose goal is simply to slow down the confirmation of blocks by proposing empty proposer and voter blocks. 
This has two effects: (1) it delays the creation of a proposer block referencing the transaction block containing the transaction, and 
(2) it delays the confirmation of such a proposer block by delaying the creation of votes on the proposer tree.
The results of this attack are shown in Figure \ref{fig:simulation}(c). 
The censorship attack adds a delay of between 15-20 seconds to \schemenosp's confirmation delay compared to the non-adversarial setting. 
The effect is smaller for the longest-chain protocol, since the only delay comes from delaying the insertion of a transaction into a block.
Under a censorship attack, double-spent transactions have the same latency as honest ones. 
}

\section*{Acknowledgement}
We thank the Distributed Technologies Research Foundation,  the Army Research Office under grant W911NF-18-1-0332-(73198-NS), the National Science Foundation under grants 1705007 and 1651236 for supporting their  research program on blockchain technologies. We  thank Applied Protocol Research Inc.\ for support and for providing a conducive environment that fostered this collaborative research. We also thank Andrew Miller and Mohammad Alizadeh  for their comments on an earlier draft. We also thank Soubhik Deb for helping us with the simulations.

\bibliographystyle{plain}
\bibliography{references}
\onecolumn

\appendices
\newgeometry{left=38mm, right=38mm} 
\section{Pseudocode}
\label{app:pseudocode}
 
\begin{algorithm}[H]
{\fontsize{8pt}{8pt}\selectfont \caption{Prism: Mining}\label{alg:prism_minining}
\begin{algorithmic}[1]

\Procedure{Main}{ }
    \State \textsc{Initialize}()
    \While{True}
        \State $header, Ppf, Cpf$ =  \textsc{PowMining}() \label{code:blockBody1}
        \State \maincolorcomment{Block contains header, parent, content and merkle proofs}
        \If{header is a \textit{tx block}}
        \State $block \gets \langle header, txParent, txPool, Ppf, Cpf\rangle$
        \ElsIf{header is a \textit{prop block}}
        \State $block \gets \langle header, prpParent, unRfTxBkPool, Ppf, Cpf\rangle$\label{code:blockBody2}
        \ElsIf{header is a \textit{block in voter} blocktree $i$}
        \State $block \gets \langle header, vtParent[i], votesOnPrpBks[i], Ppf,\label{code:blockBody3} Cpf\rangle$
        \EndIf
        \State \textsc{BroadcastMessage}($block$)  \colorcomment{Broadcast to peers}

    \EndWhile
\EndProcedure

\Procedure{Initialize}{ } \colorcomment{ All variables are global}
\vspace{0.5mm} \State \maincolorcomment{Blockchain data structure $C  = (prpTree, vtTree) $}
\State $prpTree \gets genesisP$ \label{code:prpGenesis} \colorcomment{Proposer Blocktree}
\For{$i \gets 1 \; to \; m$}
\State $vtTree[i] \gets genesisM\_i$ \colorcomment{Voter $i$ blocktree} \label{code:voterGenesis}
\EndFor
\State \maincolorcomment{Parent blocks to mine on} \label{code:VarprpParent}
\State $prpParent$ $\gets genesisP $ \colorcomment{ Proposer block to mine on} 
\For{$i \gets 1 \; to \; m$}
\State $vtParent[i]$ $\gets genesisM\_i $ \colorcomment{Voter tree $i$ block to mine on}
\EndFor \label{code:VarvtParent}
\vspace{00.4mm}\State \maincolorcomment{Block content} \label{code:allcontent}
\State  $txPool$ $\gets \phi$ \colorcomment{Tx block content: Txs to add in tx bks} \label{code:txblockcontent}
\State $unRfTxBkPool$ $ \gets \phi$ \colorcomment{Prop bk content1: Unreferred tx bks} \label{code:prblockcontent1}
\State $unRfPrpBkPool$ $ \gets \phi$ \colorcomment{Prop bk content2: Unreferred prp bks} \label{code:prblockcontent2}
\For{$i \gets 1 \; to \; m$}
\State $votesOnPrpBks(i) \gets \phi$ \colorcomment{Voter tree $i$ bkbk content } \label{code:vtblockcontent}
\EndFor

\EndProcedure
\vspace{1mm}
\Procedure{PowMining}{ }
\While{True}
\State  $txParent \gets prpParent$ \label{code:initParent}
\State \maincolorcomment{Assign content for all block types/trees}
\For{$i \gets 1 \;  to \; m$} $vtContent[i] \gets votesOnPrpBks$[i] \label{code:voteIneffcient}
\EndFor
\State $txContent \gets txPool$
\State $prContent \gets (unRfTxBkPool, unRfPrpBkPool)$
\vspace{00.4mm}\State \maincolorcomment{Define parents and content Merkle trees}
\State $parentMT\gets$MerklTree($vtParent,txParent,prpParent$) 
\State $contentMT\gets$MerklTree($vtContent,txContent,prContent$) \label{code:endContent}
\State nonce $ \gets $ RandomString($1^\kappa$)
\State \maincolorcomment{Header is similar to Bitcoin}
\State header $\gets \langle$ $parentMT.$root, $contentMT.$root, nonce $\rangle$
\vspace{00.4mm} \State \maincolorcomment{Sortition into different block types/trees } \label{code:sortitionStart}
\If {Hash(header)  $\leq mf_v$} \colorcomment{Voter block mined}
    \State $i\gets  \lfloor $Hash(header)/$f_v\rfloor$
    \textbf{and} \textit{break} \colorcomment{on tree $i$ }
\ElsIf{ $mf_v < $Hash(header)  $ \leq mf_v+f_t$} 
    \State  $i \gets m+1$ \textbf{ and} \textit{break}\colorcomment{Tx block mined}
\ElsIf{ $mf_v+f_t< $ Hash(header) $\leq mf_v+f_t+f_p$} 
    \State  $i \gets m+2$ \textbf{ and} \textit{break}\colorcomment{Prop block mined}
\EndIf
\EndWhile \label{code:sortitionEnd}
\maincolorcomment{Return header along with Merkle proofs}
\State \Return $\langle header, parentMT.$proof($i$),  $contentMT.$proof($i) \rangle$ \label{code:minedBlocm}
\EndProcedure
\\

\Procedure{ReceiveBlock}{\textsf{B}} \colorcomment{Get block from peers}    
\If {\textsf{B} is a valid \textit{transaction block}} \label{code:txblk_r1}
    \State $txPool$.removeTxFrom(\textsf{B}) \label{code:updateTxContent} 
    \State $unRfTxBkPool$.append(\textsf{B})\label{code:txblk_r2}

\ElsIf{\textsf{B} is a valid \textit{block on $i^{\text{th}}$ voter tree}} \label{code:vtblk_r1}
    \State $vtTree[i]$.append(\textsf{B}) \textbf{and} $vtTree[i]$.append(\textsf{B}.ancestors())    
    \If{\textsf{B}.chainlen $> vtParent[i]$.chainlen} 
    \State $vtParent[i] \gets \textsf{B} $ \textbf{and} $votesOnPrpBks$($i$).update(\textsf{B}) \label{code:updateVoteMine}
    \EndIf\label{code:vtblk_r2}

\ElsIf{\textsf{B} is a valid \textit{prop block}} \label{code:prpblk_r1}
    \If{\textsf{B}.level $==prpParent$.level+$1$}
    \State $prpParent \gets \textsf{B}$ \label{code:updatePropToMine}
    \For{$i \gets 1 \;  to \; m$}  
    \colorcomment{Add vote on level $\ell$ on all $m$ trees}
    \State $votesOnPrpBks($i$)[\textsf{B}.level] \gets \textsf{B}$   \label{code:addVote} 
    \EndFor
    \ElsIf{\textsf{B}.level $ > prpParent$.level+$1$}
        \State \maincolorcomment{Miner doesnt have block at level $prpParent$.level+$1$}
        \State \textsc{RequestNetwork(\textsf{B}.parent)} 
    \EndIf
    \State  $prpTree[\textsf{B} $.level].append(\textsf{B}),\; $unRfPrpBkPool$.append(\textsf{B})
    \State $unRfTxBkPool$.removeTxBkRefsFrom(\textsf{B}) \label{code:updateTxBkContent}
    \State $unRfPrpBkPool$.removePrpBkRefsFrom(\textsf{B}) \label{code:updatePrpBkContent}`

\EndIf
\vspace{1mm}
\EndProcedure

\Procedure{ReceiveTx}{\textsf{tx}} 
\If {\textsf{tx} has valid signature }   $txPool$.append(\textsf{B})
\EndIf    
\EndProcedure
\end{algorithmic}
}
\end{algorithm}

\begin{algorithm}[H]
{\fontsize{8pt}{8pt}\selectfont \caption{Prism: Tx confirmation}\label{alg:prism_con}
\begin{algorithmic}[1]

\Procedure{IsTxConfirmed}{$tx$}\label{code:fastConf}
    \State $\Pi \gets \phi$         \colorcomment{Array of set of proposer blocks}
    \For{$\ell \gets 1 \;  to \; prpTree.$maxLevel}  
        \State $votesNdepth \gets \phi$ 
        \For{$i$ in $1\;to\;m$} \label{code:getvotes_start}
        \State  $votesNdepth[i] \gets \textsc{GetVoteNDepth}(i, \ell)$ 
        \EndFor
            \If{IsPropSetConfirmed($votesNdepth$)}\colorcomment{Refer Def. \ref{defn:propListConfPolicy}}
        \State $\Pi[\ell] \gets$ {GetProposerSet}$(votesNdepth)$\colorcomment{Refer Eq. \ref{eqn:proposerlist}} \label{code:confirmPrpList}
        \Else $\;$ \textbf{break}\label{code:appendPrpList}
        \EndIf
    \EndFor
    \State \maincolorcomment{Ledger list decoding: Check if tx is confirmed in all ledgers}
    \State $prpBksSeqs \gets \Pi[1]\times \Pi[2]\times\cdots\times\Pi[\ell]$  \colorcomment{outer product} \label{code:outerproduct}
\For{$prpBks$ in $prpBksSeqs$} \label{code:ledgers_for_start}
                \State $ledger$ = \textsc{BuildLedger}($prpBks$)
    \If{$tx$ is \textbf{not confirmed} in \av{ledger}} \Return False
    \EndIf
    \EndFor
    \Return True \colorcomment{Return true if tx is confirmed in all ledgers} \label{code:fastConfirmTx}
\EndProcedure
\vspace{1mm}
\State \maincolorcomment{Return the vote of voter blocktree $i$ at level $\ell$ and  depth of the vote}
\Procedure{GetVoteNDepth}{$i, \ell$}
    \State $voterMC \gets vtTree[i].LongestChain()$
    \For{$voterBk$ in $voterMC$} \label{code:voteCountingStart}
        \For{$prpBk$ in $voterBk$.votes}
            \If{$prpBk$.level $ = \ell$} 
                    \State \maincolorcomment{Depth is \#of children bks of voter bk on main chain}
\State  \Return ($prpBk$, $voterBk$.\text{depth}) \label{code:voteCountingEnd}
            \EndIf
        \EndFor
    \EndFor
\EndProcedure
\vspace{1mm}
\Procedure{BuildLedger}{\av{propBlocks}} \colorcomment{Input: list of prop blocks}\label{code:buildLedgerFunc}
\State \av{ledger} $\gets []$ \colorcomment{List of valid transactions}
\For{\av{prpBk} in \av{propBlocks}} 
    \State $refPrpBks \gets prpBk.$getReferredPrpBks() \label{code:epochStarts}
    \State \maincolorcomment{Get all directly and indirectly referred transaction blocks.}
    \State  $txBks \gets\;  ${GetOrderedTxBks}$(prpBk, refPrpBks)$ \label{code:txbksOrdering}
    \For{\av{txBk} in \av{txBks}}
        \State $txs \gets txBk$.getTxs() \colorcomment{Txs are ordered in \av{txBk}}
        \For{$tx$ in $txs$}
        \State \maincolorcomment{Check for double spends and duplicate txs} 
            \If{$tx$ is  \textbf{valid} w.r.t to \av{ledger}} \av{ledger}.append($tx$) \label{code:sanitizeLedger}
            \EndIf
        \EndFor
    \EndFor    
\EndFor
\State \Return \av{ledger}
\EndProcedure
\vspace{1mm}
\State \maincolorcomment{Return ordered list of confirmed transactions}
\Procedure{GetOrderedConfirmedTxs()}{} \label{code:slowConf}
   \State $\texttt{L} \gets \phi$          \colorcomment{Ordered list of leader blocks}
    \For{$\ell \gets 1 \;  to \; prpTree.$maxLevel}  
        \State $votes \gets \phi$ \colorcomment{Stores votes from all $m$ voter trees on level $\ell$}
        \For{$i$ in $\gets\;1\;to\;m$} 
        \State  $votesNDepth[i] \gets \textsc{GetVotes}(i, \ell)$ 
        \EndFor
        \If{IsLeaderConfirmed($votesNDepth$)} \colorcomment{Refer  \ref{thm:common_prefix_1}} \label{code:confirmLeader}
        \State \maincolorcomment{Proposer block with maximum votes on level $\ell$}\label{code:appendLeader}
        \State $\texttt{L}[\ell] \gets$ {GetLeader}$(votesNDepth)$        
        \Else $\;$ \textbf{break}
        \EndIf
    \EndFor
    \Return \textsc{BuildLedger}(\texttt{L}) \label{code:slowBuildLedger}
    \EndProcedure

\end{algorithmic}
}
\end{algorithm}
\section{Notation}
\label{app:notation}

Let $H_i[r]$ and $Z_i[r]$ be the number of voter blocks mined by the honest nodes and by the adversarial node in round $r$ on the $i$-th voting tree respectively, where $i=1,2,..,m$.  
$H_i[r], Z_i[r]$ are Poisson random variables with means $(1-\beta)f_v\Delta$ and $\beta f_v \Delta$ respectively. Similarly, $H^p[r], Z^p[r]$ are the numbers of proposer blocks mined by the honest nodes and by the adversarial node in round $r$ respectively; they are also Poisson, with means $(1-\beta)f_p\Delta$ and $\beta f_p \Delta$ respectively. Finally, $H^t[r], Z^t[r]$ are the numbers of transaction blocks mined by the honest nodes and by the adversarial node in round $r$ respectively; they are also Poisson, with means $(1-\beta)f_t \Delta$ and $\beta f_t \Delta$ respectively. All the random variables are independent of each other. 

\section{\bitcoin backbone properties revisited}
\label{sec:app_bb}

\cite{backbone} defines three important properties of the \bitcoin backbone: common-prefix, chain-quality and chain-growth. 
It was shown that, under a certain {\em typical execution} of the mining process, these properties hold, and the properties are then used to prove the persistence and liveness of the \bitcoin  transaction ledger. These three properties, as well as the notion of a typical execution, were {\em global}, and defined over the {\em entire} time horizon. While this is appropriate when  averaging over time to achieve reliable confirmation, as for \bitcoinnosp, it turns out that for the analysis of fast latency of \schemenosp, where the averaging is over voter chains, we need to formulate finer-grained, {\em local} versions of these properties, localized at a particular round. Correspondingly, the event under which these local backbone properties are proved is also local, in contrast to the event of typical execution.

In this section, we will focus on a single \bitcoin blocktree, with a mining rate of $\bar{f}$ per round, and we will use the model and notations introduced in Section \ref{sec:model}. In addition, we will use the following notation from \cite{backbone}: if $\C$ is a chain of blocks, then $\C^{\lceil k}$ is the $k$-deep prefix of $\C$, i.e. the  chain of blocks of $\C$ with the last $k$ blocks removed. 
Additionally, given two chains $\C$ and $\C'$, we say that $\C\preceq \C'$ if $\C$ is a \emph{prefix} of chain $\C'$.
\begin{definition}[Common-prefix property] The $k$-deep common-prefix property holds at round $r$ if the $k$-deep prefix of the longest chain at round $r$ remains a prefix of any longest chain in any future round. 
\end{definition}

Note that while the common-prefix property in \cite{backbone} is parameterized by a single parameter $k$, the property defined here is parameterized by two parameters $k$ and $r$. It is a property that the prefix of the main chain at round $r$  remains permanently in the main chain in the future.

\begin{definition}[Chain-quality property] \label{def:chain_quality_single_chain}The $(\mu,k)$-chain-quality property  holds at round $r$ if at most $\mu$ fraction of the last $k$ consecutive blocks on the longest chain $\C$ at round $r$ are mined by the adversary.
\end{definition}

The chain-quality property in \cite{backbone} is parameterized by two parameters $\mu$ and $k$, however, the property defined here is parameterized by three parameters $\mu$, $k$ and $r$.

\begin{definition}[Chain-growth property] The chain-growth property with parameters $\phi$ and $s$ states that for any $s$ rounds there are at least $\phi s$ blocks added to the main chain during this time interval.
\end{definition}

We will now show that these three properties hold regardless of adversarial action, provided that certain events on the honest and adversarial mining processes hold.

\subsection{Modelling PoW block generation} \label{sec:modelling_block_generation}
In section \ref{sec:model}, the hash computation of the users are modelled as a random oracle. We now further model the PoW generation as follows: 
Let $H[r]$ and $Z[r]$ be the number of blocks mined by the honest nodes and by the adversarial node in round $r$.  From section \ref{sec:model}, we know that $H[r], Z[r]$ are Poisson random variables with means $(1-\beta)f_v\Delta$ and $\beta f_v \Delta$ respectively. Note that random variables $\{H[r]\}_{r\in\{0,r_{max}\}}$, $\{Z[r]\}_{r\in\{0,r_{max}\}}$  are independent of each other.
We now define auxiliary random variables $X[r]$ and $Y[r]$ as follows: If at round $r$ an honest party mines at least one block, then $X[r] = 1$ , otherwise $X[r] = 0$. If at round $r$ an honest party mines exactly one block, then $Y[r] = 1$, otherwise $Y[r] = 0$. Let $r' = \frac{k}{2\bar{f}}$. Define the following events: 

\begin{align}
    \tE_1\left[r - r', r\right] :=& \bigcap_{a,b\geq 0}\left\{ Y[r - r'-a, r+b] - Z[r - r'-a, r+b] > \frac{(1-2\beta)k}{8} \right\}\nonumber\\
    \tE_2\left[r - r', r\right] :=& \left\{ H\left[r - r', r\right] + Z\left[r - r', r\right] <  k \right\}\nonumber\\
    \tE_3\left[r - r', r\right] :=& \left\{ X\left[r - r', r\right] > \frac{k}{6} \right\}\nonumber\\
    \tE\;[r-r',r] := & \;\tE_1[r-r',r] \cap \tE_2[r-r',r]\cap \tE_3[r-r',r]. \label{eqn:single_chain_event}
\end{align}
As defined in Section \ref{sec:model},   $X\left[r - r', r\right]$ and $Y\left[r - r', r\right]$ are the number of  successful and uniquely 
successful rounds respectively in the interval  $[r-r',r]$, and  $Z\left[r - r', r\right]$ is the number of blocks mined by adversary in the interval  $[r-r',r]$. Note that the honest users mine at least one block in a successful round and mine exactly one block in a uniquely successful round. 
Therefore, the event $\tE_1\left[r - r', r\right]$ implies that the number of
uniquely successful rounds exceed the total blocks mined by the adversary by $\frac{(1-2\beta)k}{8}$ blocks  for \textit{all} the intervals containing the interval $\left[r - r', r\right]$. Event $\tE_2\left[r - r', r\right]$ implies that the number successful rounds plus the total number of blocks mined by the adversary in the interval $\left[r - r', r\right]$ is less than $k$. Event $\tE_3\left[r - r', r\right]$ implies that 
the number of successful rounds in the interval $\left[r - r', r\right]$ at least $\frac{k}{6}$. 


To prove the common-prefix, chain-quality and chain-growth properties, we need the following two lemmas from \cite{backbone}:

\begin{lemma}[Lemma 6 \cite{backbone}]\label{lem:lemma6_bb}
Suppose the $k$-th block, $b$, of a longest chain $\C$ was mined by a honest node in a uniquely successful round. Then the $k$-th block of a longest chain $\C'$, at possibly a different round, is either $b$ or has been mined by the adversary.
\end{lemma}

\begin{lemma}[Lemma 7 \cite{backbone}]\label{lem:lemma7_bb}
Suppose that at round $r_1$ the longest chain is  of length $n$. Then by round $r_2 \geq r_1$, the longest chain is of length of least $n+X[r_1,r_2]$.
\end{lemma}

\begin{lemma}\label{lem:block_to_round}
Under the event $\tE [r-r',r]$, the last $k$ consecutive blocks of the longest chain $\C$ at round $r$ are mined in at least $r'$ consecutive rounds.
\end{lemma}
\begin{proof}
By definition we know that $\tE_2[r-r',r] \supseteq \tE[r-r',r]$. Event $\tE_2[r-r',r]$ implies that the total number of blocks mined in interval $[r-r',r]$ is less than $k$.
Therefore, the $k$-th deep block of chain $\C$ was mined on or before round $r-r'$.
\end{proof}

\noindent The chain-growth lemma stated below is the localized version of Theorem 13 from \cite{backbone} and the proof is similar.

\begin{lemma}[Chain-growth]\label{thm:chain_growth_single_chain}
Under event $\tE[r-r',r]$, where $r'=\frac{k}{2\bar{f}}$, the longest chain grows by at least $\frac{k}{6}$ blocks in the interval $[r-r',r]$. 
\end{lemma}
\begin{proof}
From Lemma \ref{lem:lemma7_bb}, we know that the main chain grows by at least $X[r-r',r]$ in the interval $[r-r',r]$. Since $\tE_3[r-r',r] \supseteq \tE[r-r',r]$ implies $X[r-r',r]>\frac{k}{6}$ and  this completes the proof.
\end{proof}
\begin{figure}
  \centering
  \includegraphics[width=\linewidth]{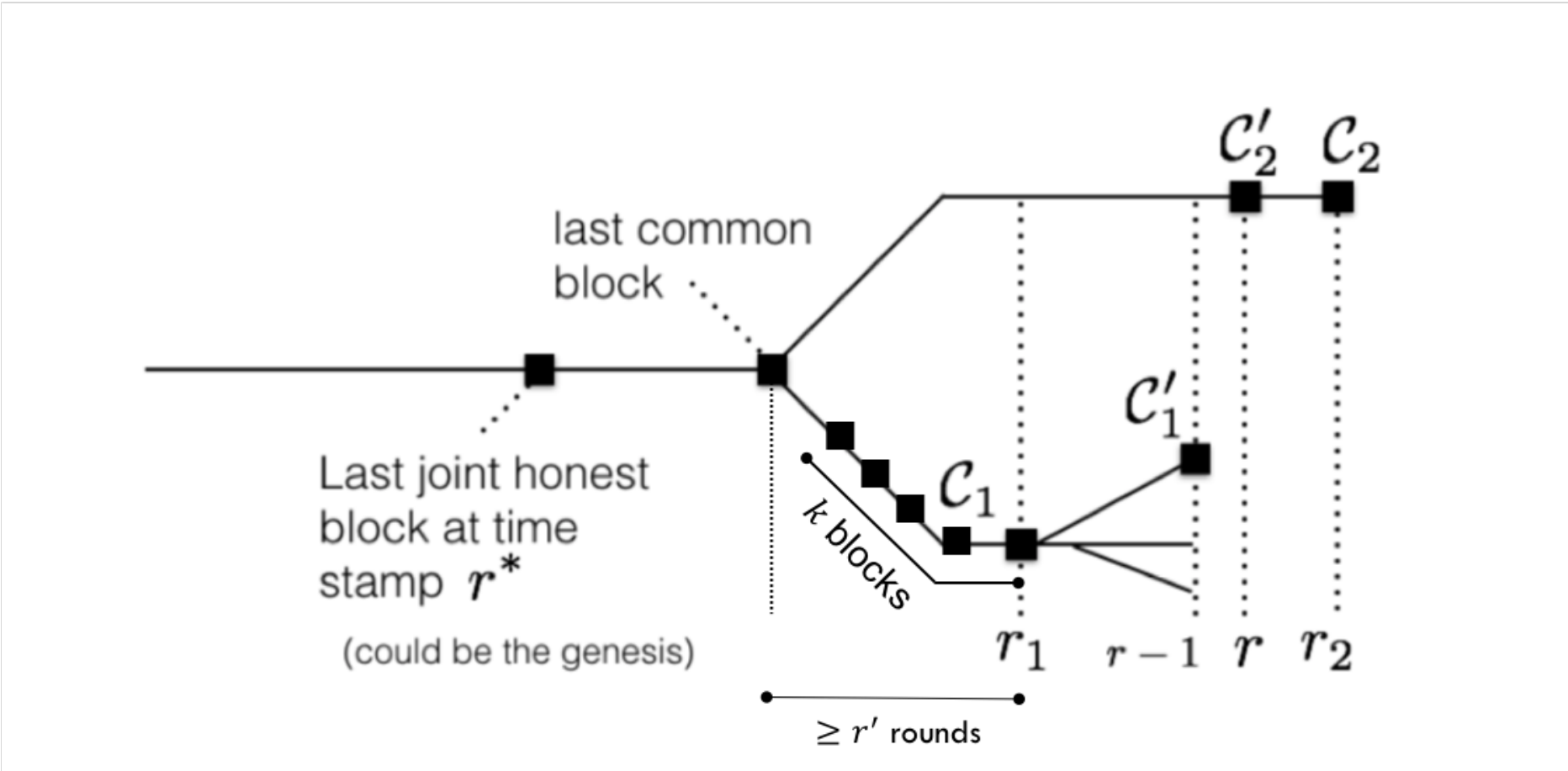}
  \caption{Round $r$ is the first round that the $k$-deep prefix of the longest chain is changed. (This is a slight modification of Figure 3 from \cite{backbone}.)} 
  \label{fig:cp_bb}
\end{figure}
\noindent We modify the proofs of Lemma 14 and Theorem 15 of \cite{backbone} by localizing it to a particular round in order to prove our common-prefix property. 

\begin{lemma}[Common-prefix]\label{thm:common_prefix_single_chain}
Under the event $\tE[r_1-r',r_1]$, where $r' = \frac{k}{2\bar{f}}$, the $k$-deep common-prefix property holds at round $r_1$.
\end{lemma}
\begin{proof}
Consider a longest chain $\C_1$ in the current round $r_1$ and a longest chain  $\C_2$ in a future round $r_2$, which violates  the common-prefix property, i.e., $\C_1^{\lceil k} \not\preceq \C_2$. 
Let $r$ be the smallest round $r_1 \leq r \leq r_2$ such that there is a longest chain $\C_2'$ such that $\C_1^{\lceil k} \not\preceq \C_2'$. If $r=r_1$, define $\C_1'=\C_1$; otherwise, define $\C_1'$ to be a longest chain at round $r - 1$. Note that $\C_1^{\lceil k} \preceq \C_1'$.
Observe that by our assumptions such an $r$ is well-defined (since e.g., $r_2$ is such a round, albeit not necessarily the smallest one); refer to Figure \ref{fig:cp_bb} for an illustration. Consider the last block $b^*$ on the common prefix of $\C_1'$ and $\C_2'$ that was mined by an honest node and let $r^*$ be the round in which it was mined (if no such block exists let $r^*= 0$).
Define the set of rounds $S=\{i:r^* <i \leq r\}$. We claim
\begin{equation}
    Z[r^*,r] \geq Y [r^*,r]. \label{eqn:event_1_contradiction}
\end{equation}

We show this by pairing each uniquely successful round in $S$ with an adversarial block mined in $S$. 
For a uniquely successful round $u\in S$, let $j_u$ be the position 
of the corresponding block i.e., its distance from the genesis block. 
Consider the set
$$J := \left\{j_u : u \text{ is a uniquely successful round in } S \right\}.$$
Note that len$(\C_1') \geq \max J$, because the honest node that mined the chain corresponding to $\max J$ position will broadcast it.  Since $\C_2'$ is adopted at round $r$, it should be at least as long as $\C_1'$, i.e., len$(\C_2')\geq$
len$(\C_1')$. As a result, for every $j \in J$, there is a block in position $j$ of either chain. We now argue that for every $j \in J$ there is an adversarial block in the $j$-th position either in $\C_1'$ or in $\C_2'$ mined after round $r^*$ because $C_1'$ and $C_2'$ contains block $b^*$ which is mined by the honest users: if $j$ lies on the common prefix of $\C_1'$ and $\C_2'$ it is adversarial by the definition of $r^*$; if not, the argument follows from Lemma \ref{lem:lemma6_bb}.

We assume the event $\tE[r_1-r',r_1]$ occurs and under $\tE_2[r_1-r',r_1] \supseteq \tE[r_1-r',r_1]$, from Lemma \ref{lem:block_to_round}, the $k$-deep block of the chain $\C_1$ was mined on or before round $r_1-r'$ and this implies $r^*\leq r_1-r'$. Under the event $E_1[r_1-r',r_1]$ we know that  $ Y[r_1 - r'-a, r_1+b] > Z[r_1 - r'-a, r_1+b]$ for all  $a,b\geq 0$.
Since $r^*<r_1-r'$,  $Y[r^*, r]> Z[r^*, r]$, which contradicts Equation \eqref{eqn:event_1_contradiction}.
\end{proof}

\noindent We again modify the proof of Theorem 16 of \cite{backbone} by localizing it to a particular round in order to prove our chain-quality property. 

\begin{lemma}[Chain-quality]\label{thm:chain_quality_single_chain}
Under the event $\tE[r-r',r]$, where $r' = \frac{k}{2\bar{f}}$, the $(\mu, k)$-chain quality property holds at round $r$ for $\mu=\frac{7+2\beta}{8}$.
\end{lemma}
\begin{proof}
Let $\C$ be the longest chain at round $r$ and denote the last $k$ blocks in the chain $\C$ by  $\C[-k]:=[b_{k}, b_{k-1},\cdots,b_{2}, b_{1}]$.
Now define $N \geq k $ as the the least number of consecutive blocks $\C[-N]:=[b_{N}, b_{N-1},\cdots,b_{2}, b_{1}]$
s.t block $b_{N}$ was mined by an honest user. Let block $b_N$ be mined in round $r^*$. 
If no such block exists then $b_{N}$ is the genesis block and $r^*=0$.
Now consider the interval $S=\{i:r^*<i\leq r\} = [r^*,r]$.  Let $H$ be the number of blocks mined by honest users in the interval $[r^*,r]$ 
and say $H<(1-\mu)k$. Then the number of blocks mined by the adversary in the same interval is at least $N-1-H$. This implies $Z[r^*,r] \geq N-1-H$, so from the chain-growth Lemma \ref{lem:lemma7_bb}, we have $N-1\geq X[r^*,r]$. Putting the last two statements together, we have
\begin{equation}
    Z[r^*,r] > X[r^*,r] - (1-\mu)k. \label{eqn:chain_quality_contradiction1a}
\end{equation}
We assume the event $\tE[r-r',r] =  \tE_1[r-r',r] \cap \tE_2[r-r',r]\cap \tE_3[r-r',r]$ occurs. Under $\tE_2[r-r',r]$, from Lemma \ref{lem:block_to_round}, the $k$-deep block of the chain $\C$, $b_k$, was mined before round $r-r'$, and since block $b_{N}$ was mined before block $b_k$, we have $r^*\leq r-r'$. Under the event $E_1[r-r',r]$, we know that  
$$ 
Y[r - r'-a, r+b] > Z[r - r'-a, r+b] + \frac{(1-2\beta)k}{8}\ \quad \forall a,b\geq 0.
$$
Since $r^*\leq r-r'$ and $X[r^*, r] \geq Y[r^*, r]$, we obtain 
\begin{equation*}
X[r^*, r] > Z[r^*, r] + \frac{(1-2\beta)k}{8},
\end{equation*}
and this contradicts Equation \eqref{eqn:chain_quality_contradiction1a} for $\mu = \frac{7+2\beta}{8}$. Therefore in the interval $[r^*,r]$,  at least $(1-\mu)k$ blocks  on $\C[-N+1]$ were mined by honest users. These blocks must be in $\C[-k]$ by definition of $N$.
\end{proof}

Since the common-prefix, chain-quality and chain-growth properties are all proved assuming the event $\tE[r-r',r]$ occurs, a natural question is how likely is its occurrence? The next lemma shows that the probability of it occurring approaches $1$ exponentially as $r'$ increases. This lemma will be heavily used in our analysis of security and fast confirmation.

\begin{lemma} \label{lem:Event_prob}
Let $\bar{f} \leq \frac{\log (2-2\beta)}{1-\beta}$.\footnote{We will assume this constraint in all our results without stating it explicitly.}
For any $r, \; \mathbb{P}\left(\tE^c\left[r - r', r\right]\right) \leq  4 e^{-\gamma \bar{f}r'}$, where $r'=\frac{k}{2\bar{f}}$ and $\gamma= \frac{1}{36}(1-2\beta)^2$.
\end{lemma}
\begin{proof}
The event $\tE^c\left[r - r', r\right]$ is a union of three events. We will upper bound the probability each of these events separately and then use union bound. 

\begin{lemma}\label{prop:E1}
 For any $r, \; \mathbb{P}\left(\tE^c_1\left[r - r', r\right]\right) \leq 2e^{-\frac{(1-2\beta)^2 \bar{f}r'}{36}}$. Here $r'=\frac{k}{2\bar{f}}$.
\end{lemma}
\begin{proof}
Let us  restate the event $\tE_1\left[r - r', r\right]$ by substituting $k=2r'\bar{f}$:
\begin{equation*}
\tE_{1}\left[r - r', r\right] := \bigcap_{a,b\geq 0}\left\{ Y[r - r'-a, r+b] - Z[r - r'-a, r+b] > \frac{(1-2\beta)\bar{f}r'}{4} \right\}.
\end{equation*}
Observe that the random variable $Y\left[r - r'-a, r+b\right] - Z\left[r - r'-a, r+b\right]$ can be interpreted the position of a $1$-d random walk  (starting at the origin) after $r'+a+b$ steps.
Here $Y\left[r - r'-a, r+b\right]$, $Z\left[r - r'-a, r+b\right]$ are the number of steps taken in right and left direction respectively.
The value of $\bar{f}$ is chosen s.t  the random variables $Y\left[r - r'-a, r+b\right]\sim \bin(r'+a+b,\frac{\bar{f}}{2})$ and as seen before $Z\left[r - r'-a, r+b\right]\sim \pois((r'+a+b)\bar{f} \beta)$; the random walk has $\frac{\bar{f}(1-2\beta)}{2}$ positive bias per step. In this random walk analogy, event $\tE_1\left[r - r', r\right]$ implies that the random walk is to the right of the point $\frac{(1-2\beta)\bar{f}}{4}$ after first $r'$ steps and remains to the right of that point in all the future steps. We analyze this event by breaking in into two events.

Define a new event $\texttt{D}\left[r - r', r\right] = \big\{ Y\left[r - r', r\right] - Z\left[r - r', r\right] < \frac{1}{3}(1-2\beta)\bar{f} r' \big\}$. In our random walk analogy, this event corresponds to a random walk which starts at the origin and is to the left of the point $\frac{1}{3}(1-2\beta)\bar{f} r'$ after $r'$ steps. We upper bound the  probability of the event $\texttt{D}\left[r - r', r\right]$:
\begin{align}
    \mathbb{P}\big(\texttt{D}\left[r - r', r\right]\big)  =& \mathbb{P}\big(Y\left[r - r', r\right] - Z\left[r - r', r\right] <  \frac{1}{3}(1-2\beta)\bar{f} r'\big)\nonumber\\
    =& \mathbb{P}\Big(Y\left[r - r', r\right] - Z\left[r - r', r\right] - \frac{1}{2}(1-2\beta)\bar{f} r' <  -\frac{1}{6}(1-2\beta)\bar{f} r'\Big)& \nonumber\\
    \overset{(a)}{\leq}& e^{-\gamma_1\bar{f}r'}.
\end{align}
The inequality $(a)$ follows by applying Chernoff bound and the value of $\gamma_1$ is $\frac{1}{36}(1-2\beta)$.
We will now use the event $\texttt{D}\left[r - r', r\right]$ to calculate the probability of the event $\tE_1^c\left[r - r', r\right]$:
\begin{align}
    \mathbb{P}\big(\tE_1^c\left[r - r', r\right]\big) &=  \mathbb{P}\big(\tE_1^c\left[r - r', r\right]\cap\texttt{D}\left[r - r', r\right]\big)  +  \mathbb{P}\big(\tE_1^c\left[r - r', r\right]\;\cap\;\texttt{D}^c\left[r - r', r\right] \big)& \nonumber 
    \\ &\leq  \mathbb{P}\big(\texttt{D}\left[r - r', r\right]\big)  + \mathbb{P}\big(\tE_1^c\left[r - r', r\right]\;\big|\;D^c\left[r - r', r\right] \big)& \nonumber 
    \\ &\overset{(a)}{\leq} e^{-\gamma_1\bar{f} r'}  + e^{-\gamma_2\bar{f}r'}\nonumber\\
    &\leq 2e^{-\gamma\bar{f} r'}.
\end{align}
In our random walk analogy, the event $\big\{\tE_1^c\left[r - r', r\right]\;\bigcap\;\texttt{D}^c\left[r - r', r\right]\big\}$
corresponds to a positive biased random walk
$Y\left[r - r', r\right] - Z\left[r - r', r\right]$ starting to the right of the point $\frac{1}{3} \bar{f}(1-2\beta)r'$ and hitting the point $\frac{1}{4} \bar{f}(1-2\beta)r'$ in a future round.
This event is analyzed in Lemma \ref{lemma:RW_3} and using this lemma we obtain inequality (a) with $\gamma = \gamma_2 = \frac{1}{36}(1-2\beta)^2$. 
\end{proof}

\begin{lemma}\label{prop:E2}
For any $r, \; \mathbb{P}\left(\tE^c_2\left[r - r', r\right]\right) \leq e^{-\bar{f}r'}$. Here $k=r'\bar{f}$.
\end{lemma}
\begin{proof}
Let us  restate the event $\tE_2\left[r - r', r\right]$ by substituting $k=2r'\bar{f}$:
\begin{equation*}
\tE_2\left[r - r', r\right] := \left\{ X\left[r - r', r\right] + Z\left[r - r', r\right] <  2\bar{f}r' \right\}.
\end{equation*}
As defined in Section \ref{sec:model}, the total number of block mined by the honest users in interval $[r-r', r]$ is $H[r-r', r] \sim \pois((1-\beta)\bar{f},r') $ and we have $ H[r-r', r] \geq X[r-r', r]$. Using this we have
\begin{align*}
    \mathbb{P}\left(\tE_2^c\left[r - r', r\right]\right) &= \mathbb{P}\left( H\left[r - r', r\right] + Z\left[r - r', r\right] \geq  2\bar{f}r' \right) \\
     &=  \mathbb{P}\left(\pois((1-\beta)\bar{f}r') + \pois(\beta\bar{f}r') \geq  2\bar{f}r' \right) \\
     &=  \mathbb{P}\left(\pois(\bar{f}r')) \geq  2\bar{f}r' \right) \\
     &\leq e^{-\bar{f}r'}. 
\end{align*}
The last inequality follows from Chernoff bound\footnote{https://github.com/ccanonne/probabilitydistributiontoolbox/blob/master/poissonconcentration.pdf}.
\end{proof}

\begin{lemma}
\label{prop:E3}
For any $r, \; \mathbb{P}\left(\tE^c_3\left[r - r', r\right]\right) \leq e^{-\frac{\bar{f}r'}{36}}$. Here $r'=\frac{k}{2\bar{f}}$.
\end{lemma}
\begin{proof}
Let us  restate the event $\tE_3\left[r - r', r\right]$ by substituting $k=2r'\bar{f}$:
\begin{equation*}
    \tE_3\left[r - r', r\right] := \left\{ X\left[r - r', r\right] > \frac{\bar{f}r'}{3} \right\}.
\end{equation*}
We know that $Y\left[r - r', r\right] \leq X\left[r - r', r\right]$ and $Y\left[r - r', r\right]\sim \bin(r',\frac{\bar{f}}{2})$. Thus we have
\begin{align*}
    \mathbb{P}\left(\tE_3^c\left[r - r', r\right]\right) &= \mathbb{P}\left( X\left[r - r', r\right] < \frac{\bar{f}r'}{3}\right) \\
    &\leq  \mathbb{P}\left( Y\left[r - r', r\right] < \frac{\bar{f}r'}{3} \right) \\
     &=  \mathbb{P}\left(\bin(r',\frac{\bar{f}}{2}) < \frac{\bar{f}r'}{3} \right) \\
     &\leq  e^{-\frac{\bar{f}r'}{36}}. 
\end{align*}
The last inequality also follows from Chernoff bound.
\end{proof}
\noindent Combining Lemmas \ref{prop:E1}, \ref{prop:E2} and \ref{prop:E3}, we obtain 
\begin{align*}
    \mathbb{P}\left(\tE^c\left[r - r', r\right]\right) &\leq  \mathbb{P}\left(\tE_1^c\left[r - r', r\right]\right)+ \mathbb{P}\left(\tE_2^c\left[r - r', r\right]\right)+ \mathbb{P}\left(\tE_3^c\left[r - r', r\right]\right)\\
    & \leq 2e^{- \frac{(1-2\beta)^2\bar{f}r'}{36}} + e^{-{\bar{f}r'}} + e^{-\frac{\bar{f}r'}{36}} \\
    & \leq 4e^{- \frac{(1-2\beta)^2\bar{f}r'}{36}}.
\end{align*}
\end{proof}

\section{Total ordering for \scheme\!\!: proofs of Theorems 4.1 and 4.2} 
\label{sec:slow_latency_app}

In Appendix \ref{sec:app_bb},  we proved three chain properties -- chain-growth, common-prefix and chain-quality -- for the \bitcoin backbone under  events defined in Equation \eqref{eqn:single_chain_event}. The voter blocktrees in \scheme also follow the longest chain protocol, hence these three chain properties will directly hold for each of the $m$ voter blocktree under the corresponding events:

Similar to section \ref{sec:modelling_block_generation}, let $H_j[r]$ and $Z_j[r]$ be the number of blocks mined by the honest nodes and by the adversarial node in round $r$ on the $i^{th}$ voter tree for $j\in[m]$.  From section \ref{sec:model} and the sorition technique, we know that $H_j[r], Z_j[r]$ are Poisson random variables with means $(1-\beta)f_v\Delta$ and $\beta f_v \Delta$ respectively. Note that random variables $\{H_j[r]\}_{r\in\{0,r_{max}\}, j \in [m]}$, $\{Z_j[r]\}_{r\in\{0,r_{max}\}, j \in [m]}$  are independent of each other.
We now define auxiliary random variables $X_j[r]$ and $Y_j[r]$ as follows: If at round $r$ an honest party mines at least one block on voter tree $j$, then $X_j[r] = 1$ , otherwise $X_j[r] = 0$. If at round $r$ an honest party mines exactly one block on voter tree $i$, then $Y_j[r] = 1$, otherwise $Y_j[r] = 0$. Let $r' = \frac{k}{2\bar{f}}$.
\begin{align}
    \tE_{1,j}\left[r - r', r\right] :=& \bigcap_{a,b\geq 0}\left\{ Y_j[r - r'-a, r+b] - Z_j[r - r'-a, r+b] > \frac{(1-2\beta)k}{8} \right\}\nonumber\\
    \tE_{2,j}\left[r - r', r\right] :=& \left\{ H_j\left[r - r', r\right] + Z_j\left[r - r', r\right] <  k \right\}\nonumber\\
    \tE_{3,j}\left[r - r', r\right] :=& \left\{ X_j\left[r - r', r\right] > \frac{k}{6} \right\}\nonumber\\
    \tE_j\;[r-r',r] := & \;\tE_{1,j}[r-r',r] \cap \tE_{2,j}[r-r',r]\cap \tE_{3,j}[r-r',r]. \label{eqn:single_chain_event_slow}
\end{align}
Note the similarity between the above events and events defined in Equation \eqref{eqn:single_chain_event}.
We know that  $X_j\left[r - r', r\right]$ and $Y_j\left[r - r', r\right]$ are the number successful and uniquely successful rounds respectively in the interval  $[r-r',r]$ on the blocktree $j$.
Along the same lines, $Z_j\left[r - r', r\right]$ is the number of voter blocks mined by the adversary on the blocktree $j$ in the interval  $[r-r',r]$.
Events $\tE_{1,j}\left[r - r', r\right]$, $\tE_{2,j}\left[r - r', r\right]$ and $\tE_{3,j}\left[r - r', r\right]$ have corresponding interpretation of the events $\tE_{1}\left[r - r', r\right]$, $\tE_{2}\left[r - r', r\right]$ and $\tE_{3}\left[r - r', r\right]$.

\noindent\textbf{Typical event:} For a given $r'$, define the following event:
\begin{equation}
    \tE_j(r') := \bigcap_{\tilde{r}\geq r'}\bigcap_{0\leq r\leq\rmax}\tE_j\left[r - \tilde{r}, r\right]. \label{eqn:slow_typical_events}
\end{equation}

\begin{lemma} \label{lem:Event_j_union_prob}
For any $j$, $\mathbb{P}\left(\tE_j^c(r')\right) \leq  4\rmax^2 e^{-\gamma\fv r'}$,
where $\gamma= \frac{1}{36}(1-2\beta)^2$. 
\end{lemma}
\begin{proof}
Use Lemma \ref{lem:Event_prob} and apply union bound.
\end{proof}
\noindent Let the first proposer block at level $\ell$ appear in round $\Rl$. We will now prove common-prefix and chain-quality for the leader block sequence defined in Equation \eqref{eqn:leader_block_seq}.

\begin{figure}
  \centering
  \includegraphics[width=.7\linewidth]{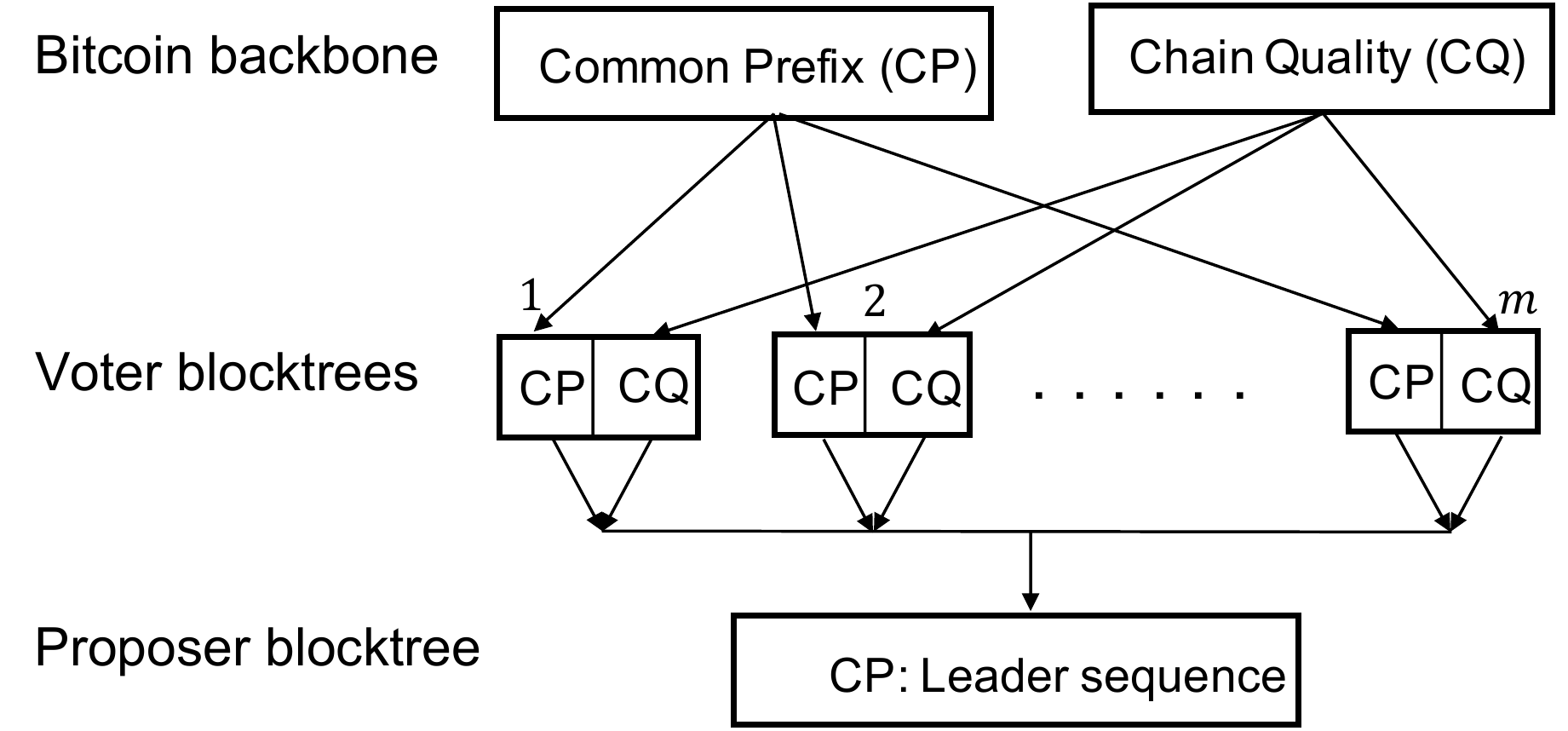}
  \caption{Dependencies of properties required to prove the common-prefix property of the leader sequence.\label{fig:slow_latency_proof}}
 \end{figure}
 
\textbf{Common prefix property:}
The common-prefix property of the leader sequence gives us the \textit{confirmation policy}. We derive this property using the common prefix and the chain-quality properties of the voter blocks.
Refer Figure \ref{fig:slow_latency_proof}.

\begin{lemma}[Common-prefix]\label{lem:slow_common_prefix_1}
At round $r\geq \Rl$, if every voter blocktree has a voter block mined by the honest users after round $\Rl$ which is at least $k$-deep, then w.p $1-\epsilon_{k}$, the leader block sequence up to level $\ell$ is permanent  i.e,
\begin{equation*}
        \texttt{LedSeq}_{\;\ell}(r) = \texttt{LedSeq}_{\;\ell}(\rmax).
\end{equation*}
Here $\epsilon_{k} \leq 4m\rmax^2 e^{-\gamma k/2}$ and $\gamma= \frac{1}{36}(1-2\beta)^2$.
\end{lemma}
\begin{proof}
Fix a voter blocktree $j$ and denote its $k$-deep voter block in round $r$ by $b_j$. From the definition in Equation \eqref{eqn:slow_typical_events} and common-prefix Lemma \ref{thm:common_prefix_single_chain} we know for under the event $\tE_j(r')$, for $r'=\frac{k}{2\fv}$, the $k$-deep voter block and its ancestors permanently remain on the main chain of voter blocktree $j$. From Lemma \ref{lem:Event_j_union_prob} we know that $\mathbb{P}\left(\tE_j^c(r')\right) \leq  \frac{\epsilon_{k}}{m}$. Therefore, the $k$-deep voter block on the voter blocktree $j$ is permanent w.p $1-\frac{\epsilon_{k}}{m}$. On applying union bound we conclude \textit{all} the $k$-deep voter block on the $m$ voter blocktrees are permanent w.p $1-\epsilon_{k}$.
Each of these voter blocks, $b_j$'s, are mined by the honest users after round $\Rl$. Therefore, by the voter mining policy defined in Section \ref{sec:protocol}, the main chain of the voter blocktree $j$ until voter block $b_j$  has votes on proposer blocks on all the levels $\ell'\leq\ell$ and all these votes are permanent w.p $1-\epsilon_{k}$. 
Therefore, for each level $\ell'\leq \ell$ has $m$ permanent votes and this implies that the leader block at level $\ell'$ is also permanent w.p $1-\epsilon_{k}$. 
\end{proof}

Therefore, to confirm leader blocks with $1-\epsilon$ security, votes on all the $m$ voter blocktrees should be at least $k=\frac{2}{\gamma}\log \frac{4m\rmax}{\epsilon}$ deep. The natural question is: how long does it take to have (at least) $k$-deep votes on \textit{all} $m$ voter blocktrees? The next lemma answers this question.

\begin{lemma}\label{lem:slow_latency_1}
By round $\Rl+r_k$, wp  $1-\epsilon'_{k}$,  all the voter blocktrees have an honest voter block mined after round $\Rl$ and is at least $k$-deep, where
$r_k \leq \frac{64 k}{(1-2\beta)\fv}$ and $\epsilon'_{k} \leq 8m\rmax^2 e^{-\frac{\gamma\fv r_k}{8}}$.
\end{lemma}
\begin{proof}
Fix a blocktree $j$.  Using the chain growth Lemma \ref{thm:chain_growth_single_chain} under the event $\tE_j\left(r_k\right)$, we know that the main chain of voter blocktree $j$ grows by $k_1 \geq \frac{r_k\fv}{3}$ voter blocks.
Next, using the chain-quality Lemma  \ref{thm:chain_quality_single_chain}  under the second event $\tE_j\left(\frac{k_1}{2\fv}\right)$, we know that at least $\frac{1-2\beta}{8}$ fraction of these $k_1$ voter blocks are mined by the honest users and the earliest of these voter block, say $b_j$, is at least $k_2$-deep, where $k_2\geq \frac{(1-2\beta)k_1}{8} \geq  \frac{(1-2\beta)\fv r_k}{24}  := k $.
It is important to note that the depth $k_2$ is \textit{observable} by all the users.
The  probability of failure of either of these two events is 
\begin{align}
    \mathbb{P}\left(\tE_j^c\left(r_k\right)\bigcup \tE_j^c\left(\frac{k_1}{2\fv}\right) \right) & \leq \mathbb{P}\left(\tE_j^c\left(r_k\right)\right)+
    \mathbb{P}\left(\tE_j^c\left(\frac{k_1}{2\fv}\right)\right) \nonumber\\
    & \overset{(a)}{\leq} \mathbb{P}\left(\tE_j^c\left(r_k\right)\right)+
    \mathbb{P}\left(\tE_j^c\left(\frac{r_k}{6}\right)\right)\nonumber\\
    & \overset{(b)}{\leq} 
    2\mathbb{P}\left(\tE_j^c\left(\frac{r_k}{6}\right)\right)\nonumber\\
    & \overset{(c)}{\leq} \frac{\epsilon'_{k}}{m}. \label{eqn:single_voter_chain_failure}
\end{align}
From Lemma \ref{lem:Event_j_union_prob}, we see that as $r'$ decreases,  $\mathbb{P}\left(\tE_j^c(r')\right)$ increases, and because  $\frac{k_1}{2\fv} \geq \frac{r_k}{6}$, we have the inequality $(a)$. The inequality $(b)$ also follows by the same logic. The last inequality $(c)$ is  given by Lemma \ref{lem:Event_j_union_prob}. Now applying union bound on Equation \eqref{eqn:single_voter_chain_failure} over $m$ blocktree gives us the required result.
\end{proof}

\noindent
{\bf Proof of Theorem \ref{thm:common_prefix_1}:}
\begin{proof}
From Lemma \ref{lem:slow_latency_1} we know that by round $\Rl+r(\epsilon)$, all the voter blocktrees will have a $k$-deep honest voter blocks wp at least $1-\frac{\epsilon}{2}$\footnote{$8m\rmax^2 e^{-\frac{\gamma\fv}{8}\frac{64}{\gamma\fv(1-2\beta)}\log \frac{8m\rmax^2}{\epsilon}} \hspace{-1mm}= 8m\rmax^2 e^{-\frac{8}{1-2\beta}\log \frac{8m\rmax^2}{\epsilon}} \hspace{-1mm}< \frac{\epsilon}{2}$.} for $k\geq \frac{2}{\gamma}\log \frac{8m\rmax^2}{\epsilon}$.
Now applying Lemma \ref{lem:slow_common_prefix_1} for $k\geq \frac{2}{\gamma}\log \frac{8m\rmax^2}{\epsilon}$, we obtain that all these honest voter blocks are permanent w.p $1-\frac{\epsilon}{2}$. On combining these two, we obtain that by round $\Rl+r(\epsilon)$ the leader block sequence up to level $\ell$ is permanent w.p $1-\epsilon$.
\end{proof}

\textit{Worst Case vs Average Case:} The confirmation policy in Lemma \ref{lem:slow_latency_1} is stated for the worst case adversarial attack: when there are two (or more) proposer blocks at a given level have equal number of votes. Consider an `average case' scenario with two proposer blocks at a level, where the first block has $2m/3$ votes and the second block as $m/3$ votes. In this scenario one can intuitively see that we don't need to guarantee permanence of all the $m$ votes but a weaker guarantee suffices: permanence of $m/6$ of the $2m/3$ votes of first block.
This weaker guarantee can be achieved within a few rounds and translates to short latency in \schemenosp.

\begin{corollary} \label{cor:bitcoin_latency}
\bitcoinnosp's latency is the time required to mine a single honest $\frac{1}{\epsilon}$-deep block on a voter chain of \schemenosp\;and it is lesser than  $\frac{2304}{\fv(1-2\beta)^2}\log \frac{8\rmax^2}{\epsilon}$ rounds to provide $1-\epsilon$ reliability to confirm blocks and the transactions in it.
\end{corollary}



\begin{definition} [Leader-sequence-quality ] \label{def:leader_quality_slow}The $(\mu,k)$-leader-sequence-quality property holds at round $r$ if at most $\mu$ fraction of the last $k$ consecutive leader blocks on the proposer blocktree at round $r$ are mined by the adversary.
\end{definition}
Similar to section \ref{sec:modelling_block_generation}, let $H^p[r]$ and $Z^p[r]$ be the number of blocks mined by the honest nodes and by the adversarial node on proposer tree in round $r$.  From section \ref{sec:model} along with sortition technique, we know that $H^p[r], Z^p[r]$ are Poisson random variables with means $(1-\beta)f_v\Delta$ and $\beta f_v \Delta$\footnote{$f_v=f_p$} respectively. Note that random variables $\{H^p[r]\}_{r\in\{0,r_{max}\}}$, $\{Z^p[r]\}_{r\in\{0,r_{max}\}}$  are independent of each other.
We now define auxiliary random variables $X^p[r]$ and $Y^p[r]$ as follows: If at round $r$ an honest party mines at least one block on the proposer tree, then $X^p[r] = 1$ , otherwise $X^p[r] = 0$. If at round $r$ an honest party mines exactly one block on the proposer tree, then $Y^p[r] = 1$, otherwise $Y^p[r] = 0$. Let $r' = \frac{k}{2\bar{f}}$.

\noindent Let us define the following events on the proposer blocktree:

\begin{align}
    \tE_{1}^p\left[r - r', r\right] :=& \bigcap_{a,b\geq 0}\left\{ Y^p[r - r'-a, r+b] - Z^p[r - r'-a, r+b] > \frac{(1-2\beta)k}{8} \right\}\nonumber\\
    \tE_{2}^p\left[r - r', r\right] :=& \left\{ H^p\left[r - r', r\right] + Z^p\left[r - r', r\right] <  k \right\}\nonumber\\
    \tE_{3}^p\left[r - r', r\right] :=& \left\{ X^p\left[r - r', r\right] > \frac{k}{6} \right\}\nonumber\\
    \tE^p\;[r-r',r] := & \;\tE_{1}^p[r-r',r] \cap \tE_{2}^p[r-r',r]\cap \tE_{3}^p[r-r',r]. \label{eqn:single_chain_event_slow_1}
\end{align}
We know that $X^p\left[r - r', r\right]$ and $Y^p\left[r - r', r\right]$ are the number of users in successful and uniquely successful rounds respectively in the interval  $[r-r',r]$ on proposer blocktree, and $Z^p\left[r - r', r\right]$ is the number of proposer blocks mined by adversary in the interval  $[r-r',r]$.
These events have corresponding interpretation of the events defined in Equations \eqref{eqn:single_chain_event} and  \eqref{eqn:single_chain_event_slow}.

\noindent
\begin{figure}
  \hspace{-10mm}
  \includegraphics[width=\linewidth]{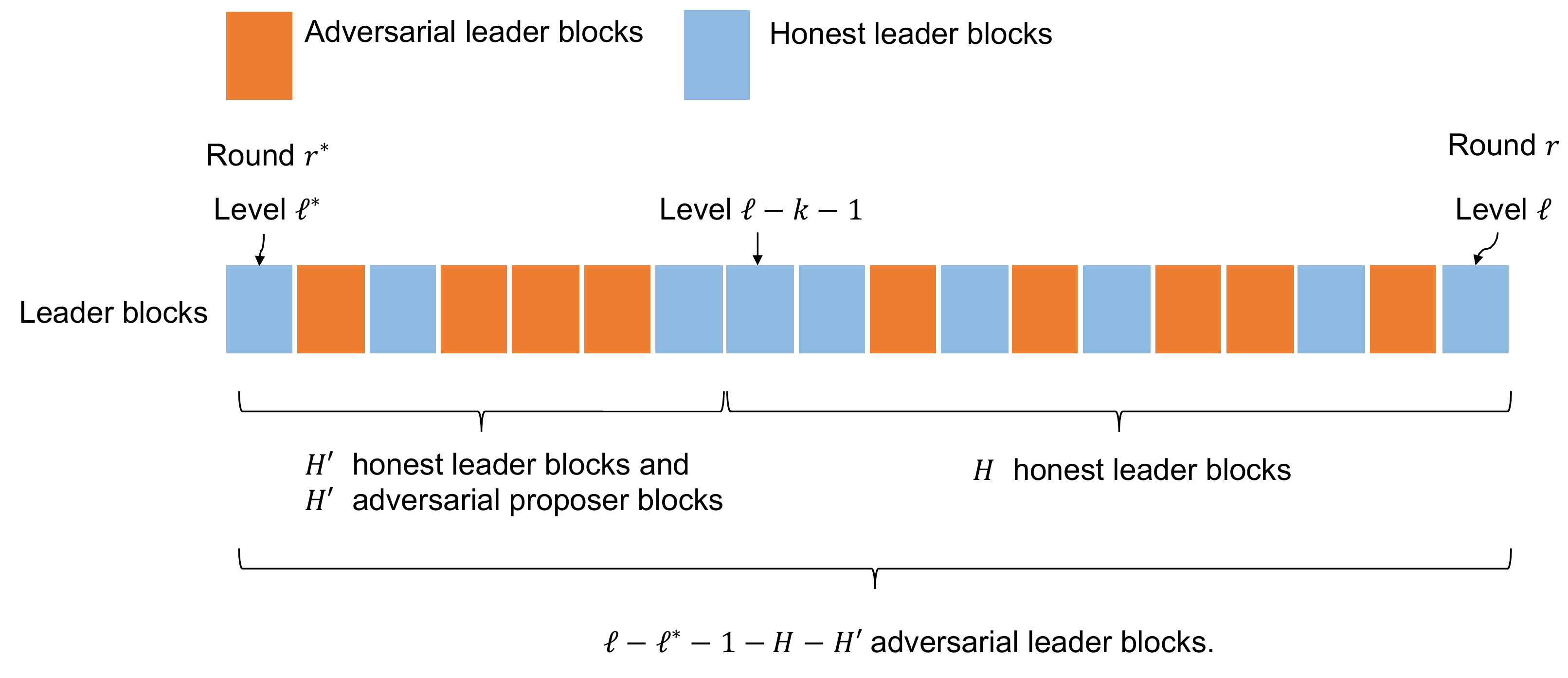}
  \caption{The leader blocks in levels $[\ell^*,\ell]$.\label{fig:slow_chain_quality}}
 \end{figure}

 \begin{lemma}[Leader-sequence-quality] \label{lemma:ls_slow_quality}
The $(\mu, k)$-leader-sequence-quality property holds at round $r$ for $\mu=\frac{7+2\beta}{8}$ w.p at least $1-4\rmax^2 e^{-(1-2\beta)^2 k/72   }$.
\end{lemma}

\begin{proof}
Unlike the longest chain in \bitcoinnosp, the leader sequence in \scheme does not form a chain. Therefore, one cannot directly use Lemma \ref{thm:chain_quality_single_chain} and we need to adapt its proof to prove the required property here.

Let $r$ be the current round and $\ell$ be the last level on the proposer blocktree which has proposer blocks at round $r$.
Consider the $k$ consecutive leader blocks on levels $[\ell-k, \ell] := \{\ell-k+1, \cdots, \ell\}$ on the leader sequence $\LS_{\ell}(r)$ and define:
\begin{align*}
    \ell^* &:= \max \big( \tilde{\ell} \leq \ell-k+1 \;\;s.t\;\; \text{ the honest users mined the first proposer block on level } \tilde{\ell}\big)
\end{align*}
Let $r^*$ be the round in which the first proposer block was mined on level $\ell^*$ and define  the  interval $S:=\{r:r^*< i\leq r\}=[r^*,r]$. From the definition of $\ell^*$ we have the following two observations:
\begin{enumerate}
    \item The adversary has mined at least one proposer block on all levels in $[\ell^*,\ell-k+1]$.
    \item All the proposer blocks on levels $[\ell^*,\ell]$ are mined in the interval $S$ because there are no proposer blocks on level $\ell^*$ before round $r^*$ and hence no user can mine a proposer block on a level greater than $\ell^*$ before round $r^*$.
\end{enumerate}
Let $H$ be  the number of honest leader blocks on the levels $[\ell-k,\ell]$ and say
\begin{equation}
    H < (1-\mu)k. \label{eqn:local_contradition_1}
\end{equation}
Let $H'$ be  the number of honest leader blocks on the levels $[\ell^*,\ell-k]$. The adversary has mined $\ell-\ell^*-1 - H - H'$ leader blocks in the interval $S$. From our first observation, we know that the number of proposer blocks mined by the adversary on the levels $[\ell^*,\ell-k]$ which are \textit{not} leader blocks is at least $H'$, and from our second observation, these proposer blocks are mined in the interval $S$. Therefore, the number of proposer blocks mined by the adversary in the interval $S$ satisfies
\begin{align}
    Z^p[r^*,r]&\geq(\ell-\ell^*-H-H'-1)+H' \nonumber \\
     &\geq \ell-\ell^*- 1-H \nonumber \\
     \left(\text{From Equation } \eqref{eqn:local_contradition_1}\right)& > \ell-\ell^* -1 - (1-\mu)k. \label{eqn:local_454}
\end{align}
Refer Figure \ref{fig:slow_chain_quality} for an illustration. From the chain growth Lemma \ref{thm:chain_growth_single_chain}, we know that $\ell-\ell^* - 1 \geq X^p[r^*,r]$ and combining this with Equation \eqref{eqn:local_454} gives us 
\begin{equation}
Z^p[r^*,r] >  X^p[r^*,r] - (1-\mu)k. \label{eqn:chain_quality_contradiction1}   
\end{equation}

Let $r':=\frac{k}{2\fv}$. Define an event $\tE^p(r') :=\bigcap_{\tilde{r}\geq r'} \bigcap_{r\leq\rmax}\tE^p\left[r - \tilde{r}, r\right]$ and assume the event $\tE^{\;p}(r')$ occurs. Under the event $E_1^p[r-r',r] \supseteq \tE^{\;p}(r')$, we know that  
$$ 
Y^p[r - r'-a, r+b] > Z^p[r - r'-a, r+b] + \frac{(1-2\beta)k}{8}\ \quad \forall a,b\geq 0.
$$
The first proposer block on the level $\ell$ is mined before round $r$. Under the event  $\tE_2^p[r-r',r] \supseteq E^p(r')$, from Lemma \ref{lem:block_to_round}, the first proposer block on level $\ell-k+1$ was mined before round $r-r'$, and hence  $r^* \leq r-r'$. This combined with $X^p[r^*, r] \geq Y^p[r^*, r]$, gives us
\begin{equation*}
X^p[r^*, r] > Z^p[r^*, r] + \frac{(1-2\beta)k}{8},
\end{equation*}
and this contradicts Equation \eqref{eqn:chain_quality_contradiction1} for $\mu = \frac{7+2\beta}{8}$. Therefore on the levels $[\ell-k,\ell]$, at least $\frac{1-2\beta}{8}$ fraction of the leader blocks are mined by honest users. From Lemma \ref{lem:Event_j_union_prob}, we know that the event $\tE^{\;p}(r')$ occurs w.p $1-4\rmax^2 e^{-\gamma k/2}$, where $\gamma=\frac{1}{36}(1-2\beta)^2$, and this completes the proof.
\end{proof}

The  leader sequence quality defined in \ref{def:leader_quality_slow}  is parameterized by two parameters $r$ and $k$, whereas its counterpart definition of chain quality in \cite{backbone}, is parameterized only by a single parameter $k$. Even though our definition of `quality' is a weaker, we show that it is suffices to ensure liveness.
\begin{figure}
  \hspace{-5mm}
  \includegraphics[width=1\linewidth]{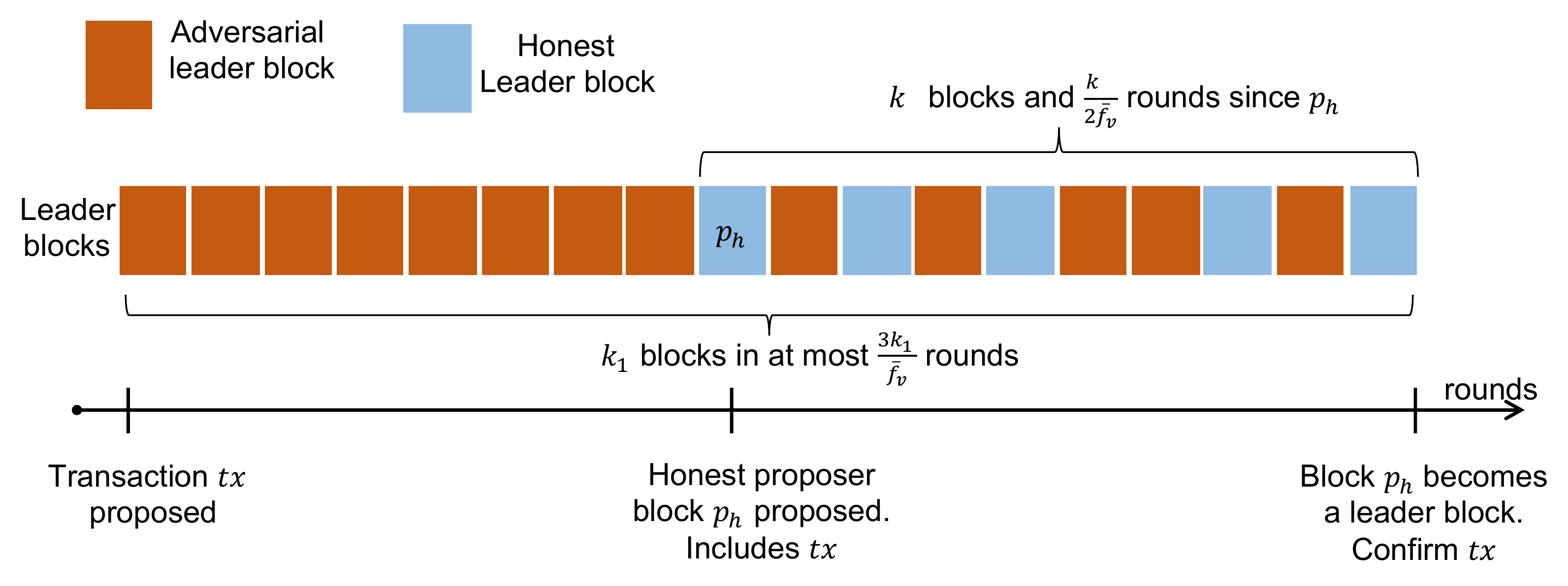}
  \caption{\label{fig:slow_latency_liveness}}
 \end{figure}
\newline
\newline
\noindent{\bf Proof of Theorem \ref{thm:liveness}:}
\begin{proof}
Let $k:=\frac{2048}{(1-2\beta)^3}\log(\frac{32m\rmax}{\eps})$ and  $k_1:=\frac{8k}{1-2\beta}$.
Using Lemma \ref{lemma:ls_slow_quality} we know that w.p at least $1-\eps/4$, the last $k_1$ leader blocks have at least $k$ honest leader blocks. From Lemma \ref{lem:block_to_round} and \ref{lem:Event_j_union_prob}, w.p at least $1-\eps/4$, the deepest of these $k$ honest leader block was proposed before the round  $r-\frac{k}{2\fv}$.  Here $\frac{k}{2\fv} = \frac{1024}{(1-2\beta)^3\fv}\log(\frac{32m\rmax}{\eps})$ and now using Theorem \ref{thm:common_prefix_1}, this deepest honest leader block is permanent w.p $1-\eps/4$. Therefore, the honest transaction will be permanently added to the blockchain after $k_1$ proposer blocks are mined. Using chain growth Lemma \ref{thm:chain_growth_single_chain}, we know that the $k_1$ proposer blocks will be mined in no more than $ \frac{3k_1}{\fv}$ rounds w.p $1-\eps/4$. Therefore, w.p $1-\eps$, the transaction will be part of the permanent leader sequence in 
\begin{equation}
    \frac{3k_1}{\fv} = \frac{3\times 2^{14}}{(1-2\beta)^3\fv}\log\frac{32m\rmax}{\eps} \mbox{ rounds}. \label{eqn:thm_2_local_eqn_1}
\end{equation}
Refer Figure \ref{fig:slow_latency_liveness} for an illustration. Note that the constants in Equation \eqref{eqn:thm_2_local_eqn_1} have not be optimized for the sake of readability. The scaling w.r.t $1-2\beta$, $\fv$ and $\log \frac{1}{\eps}$ is the main take away.
\end{proof}

\section{Fast list confirmation for \scheme\!\!: Proof of Theorem 4.5 and 4.6} 
\label{sec:fast_list_app}

\subsection{Voter chain properties}

\label{appen:sec:latency_proof}

In Appendix \ref{sec:slow_latency_app}, we proved the common-prefix and the leader-sequence-quality properties by requiring the typical event defined in Equation \eqref{eqn:slow_typical_events} to hold for {\em every} voting chain, i.e. at the {\em microscopic} scale. The typicality of each such event was obtained by averaging over rounds and as a consequence the confirmation of leader blocks with $1-\epsilon$ guarantee required  averaging over $O(\log \frac{1}{\epsilon})$ rounds. In this section we obtain faster confirmation time by relaxing the notion of typicality to a notion of {\em macroscopic} typicality, one which concerns the mining processes of a large fraction of the voter chains. This event guarantees macroscopic versions of the chain-growth, common-prefix and chain-quality properties. That is, these properties are guaranteed to be satisfied by a large fraction of the voter chains, but \textit{not} all. These macroscopic properties of the voting chains turn out to be sufficient to allow fast confirmation. 
For this section, we will define:
\begin{eqnarray}
    \gamma &:=& \frac{1}{36}(1-2\beta)^2 \nonumber\\ 
     c_1 &:=& \cone\nonumber\\
    \rmin &:=& \frac{2\log \frac{200}{ \gamma c_1}}{\gamma\fv}\nonumber\\
    \kmin &:=& \frac{4}{\gamma }\log \frac{200}{\gamma c_1}\nonumber\\
 \rho_{r'} &:=& \max \left(\frac{c_1}{1+4\fv r'}, \frac{(1-2\beta)c_1}{1+32\log m} \right)\nonumber\\
     \delta_k & :=&  \max\left(\frac{c_1}{1+2k}, \frac{(1-2\beta)c_1}{1+32\log m} \right)\nonumber\\
    \eps_m &:= & r_{\text{max}}^2e^{-\frac{(1-2\beta)c_1m}{2+64\log m}}. \label{eqn:constants}
\end{eqnarray}


\begin{lemma}[Macroscopic Typicality] \label{lem:typical_events}
The macroscopic typical event $\texttt{T}$ defined below occurs with probability $1-\eps_m$.
\begin{align*}
    \texttt{T}[r-r',r] &:= \left\{\frac{1}{m}\sum_{j=1}^{m}\mathbf{1}\big(\tE_j[r-r',r])\geq 1-\delta_k \right\}\\
        \texttt{T} &:= \bigcap_{0\leq r\leq\rmax, r'\geq\rmin}\texttt{T}[r-r',r],
\end{align*}
where $r'=\frac{k}{2\fv}$.
Note that $\delta_k=\rho_{r'}$.
\end{lemma}
\begin{proof}

For a fixed $r,r'$, the indicator random variables $\mathbf{1}\big(\tE_j^c[r-r',r]\big)$ are identical and independent  $\forall j\in[m]$. The mean of the random variable $\mathbf{1}\big(\tE_j^c[r-r',r]\big)$ is $\mu$, and it is at most $4e^{-\gamma\fv r'}$ by Lemma \ref{lem:Event_prob}. Using Chernoff bound\footnote{\url{http://math.mit.edu/~goemans/18310S15/chernoff-notes.pdf}} for Bernoulli random variables, $\forall a \geq 0$, we have

\begin{align}
    &\mathbb{P} \big\{\frac{1}{m}\sum_{j=1}^{m}\mathbf{1}\big(\tE_j^c[r-r',r]\big) \geq \mu+a \big\} \leq  e^{-\frac{ma^2}{a+2\mu}} \nonumber \\
    \overset{(1)}{\Rightarrow} &\mathbb{P} \big\{\frac{1}{m}\sum_{j=1}^{m}\mathbf{1}\big(\tE_j^c[r-r',r]\big) \geq 4e^{-\gamma\fv r'}+a \big\} \leq  e^{-\frac{ma^2}{a+4e^{-\gamma\fv r'}}} \nonumber\\ 
    \overset{(2)}{\Rightarrow} &\mathbb{P} \big\{\frac{1}{m}\sum_{j=1}^{m}\mathbf{1}\big(\tE_j^c[r-r',r]\big) \geq (\delta+1)4e^{-\gamma\fv r'} \big\} \leq  e^{-4me^{-\gamma\fv r'}\frac{\delta^2}{\delta+2}}. \label{eqn:typ_local_1}
\end{align}
Step $(1)$ follows because $\mu \leq 4e^{-\gamma\fv r'}$, and step $(2)$ is obtained by substituting $a = \delta4e^{-\gamma\fv r'}$. For $r'\geq\rmin$\footnote{The value of $\rmin$ was precisely chosen to satisfy this inequality.}, we have $ \frac{1}{4}e^{\gamma\fv r'}\rho_{r'} > 10$. On substituting $\delta = \frac{1}{4}e^{\gamma\fv r'}\rho_{r'}-1$
in Equation \eqref{eqn:typ_local_1}, for all values of $r'\geq \rmin$, we get

\begin{align*}
     \mathbb{P}(\texttt{T}^c[r-r',r]) = \mathbb{P} \big\{\frac{1}{m}\sum_{j=1}^{m}\mathbf{1}\big(\tE_j^c[r-r',r]\big) \geq \rho_{r'} \big\} &\leq  
     e^{-m\rho_r\frac{\delta^2}{(\delta+1)(\delta+2)}}\\
    & \overset{(a)}{\leq}  e^{-m\rho_r/2}\\
    & \overset{(b)}{\leq} e^{-\frac{(1-2\beta)c_1m}{2+64\log m}}.
\end{align*}
The inequality $(a)$ follows from because $\delta>9$ and inequality $(b)$ follows because $\rho_r'>\frac{(1-2\beta)c_1}{1+32\log m}$.
Since $r,r'$ can take at most $\rmax$ values, the event $\texttt{T}^c$ is a union of at most $\rmax^2$ $\texttt{T}^c[r-r',r]$ events. Using union bound we prove that the event $\texttt{T}^c$ occurs w.p at most $\eps_m = \rmax^2e^{-\frac{(1-2\beta)c_1m}{2+64\log m}}$ and this combined with $\delta_k=\rho_{r'}$ proves the required result.
\end{proof}



\begin{lemma}[Macroscopic Chain-growth]\label{lem:frac_chain_growth}
Under the event $\tT$,  for $k\geq\kmin$ and $r'=\frac{k}{2\fv}$, the longest chain grows by at least $\frac{k}{6}$ blocks in the interval $[r-r',r]$ on at least $1-\delta_k$ fraction of voter blocktrees.
\end{lemma}
\begin{proof}
From the typicality Lemma \ref{lem:typical_events}, we know that under the event $\tT[r-r',r]\supseteq T$,
\begin{equation*}
    \frac{1}{m}\sum_{j=1}^{m}\mathbf{1}\big(\tE_j[r-r',r]\big) \geq  1- \delta_k.
\end{equation*}
Applying Lemma \ref{thm:chain_growth_single_chain} on events $\tE_j[r-r',r]$  for $j\in[m]$ gives us the required result.
\end{proof}

\begin{lemma}[Macroscopic Common-prefix]\label{lem:frac_common_prefix}
Under the event $\tT$,  for $k\geq\kmin$ and $r'=\frac{k}{2\fv}$,  the $k$-deep common-prefix property holds at round $r$ for at least $1-\delta_k$ fraction of voter blocktrees.
\end{lemma}
\begin{proof}
From the typicality Lemma \ref{lem:typical_events}, we know that under the event $\tT[r-r',r]\supseteq T$,
\begin{equation*}
    \frac{1}{m}\sum_{j=1}^{m}\mathbf{1}\big(\tE_j[r-r',r]\big) \geq  1- \delta_k.
\end{equation*}
Applying Lemma \ref{thm:common_prefix_single_chain} on events $\tE_j[r-r',r]$ for $j\in[m]$ gives us the required result.
\end{proof}

\begin{lemma}[Macroscopic Chain-quality]\label{lem:frac_chain_quality}
Under the event $\tT$,  for $k\geq\kmin$ and $r'=\frac{k}{2\fv}$,  the $(\mu, k)$-chain quality property holds at round $r$ for $\mu=\frac{7+2\beta}{8}$ for at least $1-\delta_k$ fraction of voter blocktrees.
\end{lemma}
\begin{proof}
From the typicality Lemma \ref{lem:typical_events}, we know that under the event $\tT[r-r',r]\supseteq T$,
\begin{equation*}
    \frac{1}{m}\sum_{j=1}^{m}\mathbf{1}\big(\tE_j[r-r',r]\big) \geq  1- \delta_k.
\end{equation*}
Applying Lemma \ref{thm:chain_growth_single_chain}  on events $\tE_j[r-r',r]$  for $j\in[m]$ gives us the required result.
\end{proof}
In Appendix \ref{sec:slow_latency_app}, we used microscopic properties of each voter chain to obtain the common-prefix and the leader sequence quality properties for the blocktree. The voter chains require long interval of rounds to individually satisfy the microscopic properties and that results in large latency. Here we change use a different strategy: we use macroscopic properties of the voter chains to obtain the common-prefix and the leader sequence quality properties. The voter chains satisfy macroscopic properties for short interval of rounds and this directly translates to short latency.
\subsection{Fast list confirmation policy} 

We repeat the definitions from Section \ref{subsubsec:fast_list_conf}. $\P_{\ell}(r) = \{p_1,p_2...\}$ is the set of proposer blocks at level $\ell$ at round $r$.
Let $U_{\ell}(r)$ be the number of voter blocktrees which have not voted for any proposer block in the set $\P_{\ell}(r)$.
Let $V^k_n(r)$ be the number of votes at depth $k$ or greater for proposer block $p_n$ in round $r$. Let  $V_{-n}^k(r)$ be the number of votes at depth $k$ or greater for proposer blocks in the subset $\P_{\ell}(r) - \{p_n\}$.  Note that $V^k_n(r)$ and $V_{-n}^k(r)$ are observable quantities.
The following lemma bounds the future number of votes on a proposer block.

\begin{lemma}\label{lem:votes_bounds}
With probability at least $1-\eps_m$, the number of votes on any proposer block $p_n$ in any future round $r_f\geq r$, $V_n(r_f)$, satisfies
\begin{equation}
    \underbar{V}_n(r) \leq V_n(r_f) \leq \overline{V}_n(r),\nonumber
\end{equation}
where 
\begin{align}
    \underbar{V}_n(r) &:= \max_{k \geq \kmin}  (V_n^{k}(r) -\delta_k m)_{+} , \label{eqn:votes_lowerbound}\\
    \overline{V}_n(r)& := 
    {V}_n(r)+\left(V_{-n}(r) - \max_{k \geq \kmin}  (V_{-n}^{k}(r) -\delta_k m)_{+}\right)+U_{\ell}(r).
\end{align}
\end{lemma}
\begin{proof}
From the typicality Lemma \ref{lem:typical_events}, we know that the typical event $T$ occurs w.p $1-\eps_m$. We will use this to prove $V_n(r_f)\geq (V_n^{ k}(r) -\delta_k m)_{+}$ for all values of $k\geq \kmin$. For a fixed $k$, let $r'=\frac{k}{2\fv}$.  Under the event $T$,  from Lemma \ref{lem:frac_common_prefix}, we know that the $k$-deep common-prefix property holds for at least $1-\delta_k$ fraction of voter blocktrees.  Therefore $V_n(r_f)$ is at least $ (V_n^{ k}(r) -\delta_k m)_{+} $  for all $r_f\geq r$ . Since this holds for all values of $k\geq \kmin$, we have $\underbar{V}_n(r) := \max_{k \geq \kmin}  (V_n^{ k}(r) -\delta_k m)_{+} $.

Following the same line of reasoning,
$ \underbar{V}_{-n}(r) := \max_{k \geq \kmin}  (V_{-n}^{ d}(r) -\delta_k m)_{+} $ is a lower bound on $V_{-n}(r')$. Therefore, at most $(V_{-n}(r) - \underbar{V}_{-n}(r))$ votes can be removed from proposer blocks  in the set $\P_{\ell}(r) - \{p_n\}$ and added to the proposer block $p_n$. Also the $U_{\ell}(r)$ voter blocktrees which have not yet voted could also vote on block $p_n$. Combining these both gives us the upper bound on $V_n(r_f)$.
\end{proof}

Any private block $p_{\text{private}} \not \in \P_{\ell}(r)$ by definition has zero votes at round $r$.
The future number of votes on the  proposer block $p_{\text{private}}$  w.p $1-\eps_m$ satisfies
\begin{equation}
V_{\text{private}}(r_f)\leq \overline{V}_{\text{private}}(r) := m-\sum_{p_n \in \mathcal{P}_{\ell}(r)} \underbar{V}_n(r)  \qquad \forall r_f\geq r, \label{eqn:private_block_upper_bound}
\end{equation}
because  each proposer block $p_n$ has $\underbar{V}_n(r)$ permanent votes w.p $1-\eps_m$ and $p_{\text{private}}$ could potentially get the rest of the votes.

\textbf{Fast list confirmation policy}: If $\max_n  \underbar{V}_n(r) > \overline{V}_{\text{private}}(r) $, confirm the following proposer block list at level $\ell$:
\beqa
\Pi_{\ell}(r) := \{p_i: \overline{V}_i(r) > \max_n \underbar{V}_n(r) \}. \label{eqn:fast_confirmation_policy}
\eeqa

Figures \ref{fig:confidence} illustrate one such example. The definition of $\Pi_{\ell}(r)$ is precisely designed to prevent private proposer blocks from becoming the leader blocks in the future rounds. 

\begin{lemma} \label{lem:leader_sequence_list_decoding}
If the proposer lists are confirmed for all levels $\ell' \leq \ell$ by round $r$, then w.p $1- \eps_m$, the final leader sequence up to level $\ell$  satisfies 
\begin{equation}
p_{\ell'}^*(\rmax)\in \Pi_{\ell'}(r) \quad\forall \ell'\leq\ell. \nonumber
\end{equation}
\end{lemma}
\begin{proof}
We prove by contradiction. Say the final leader block at level $\ell'\leq\ell$ is $p_{\ell'}^*(\rmax) = b_i$ and  $b_i \not\in \Pi_{\ell'}(r)$. Without loss of generality, let us assume proposer block $p_1$ has the largest $\underbar{V}_n(r)$ in round $r$.  We have
\begin{equation}
    {V}_i(r_f) \overset{(a)}{\leq} \overline{V}_i(r)\overset{(b)}{<}\underline{V}_1(r) \overset{(c)}{\leq} {V}_1(r_f) \qquad  \forall r_f\geq r. \label{eqn:contradiction_11}
\end{equation}
The inequality (b) is by definition of $\Pi_{\ell'}(r)$, and the inequalities (a) and (c) are due to confidence intervals from Lemma \ref{lem:votes_bounds}.
 Equation \eqref{eqn:contradiction_11} gives us ${V}_i(r_f) < {V}_1(r')$, and therefore the proposer block $b_i$ cannot be the leader block in any future rounds $r_f\geq r$, which includes the final round $\rmax$. Therefore, we have $p_{\ell'}^*(\rmax)\in \Pi_{\ell'}(r) \;\forall \ell'\leq\ell$ and this proves the required result.
\end{proof}

Lemma \ref{lem:leader_sequence_list_decoding} proves that the proposer lists obtained via the fast list confirmation policy contains the final leader blocks. The natural question is: how long does it take to satisfy the constraint for the fast list confirmation? We answer this question next.

\subsection{Latency}
The first proposer block at level $\ell$ appears in round $R_{\ell}$. For ease of calculations we assume that the proposer blocktree mining rate $\fp = \fv$. Define $\Delta_0 := \frac{12\rmin}{1-2\beta}$.
\begin{lemma}\label{lem:fast_latency_1}
By round $r=\Rl+\Dr$, for $\Dr \geq \Delta_0$, w.p  $1-\eps_{m}$,  at least $1-4\rho_{\Dr}$ fraction of the voter blocktrees have an honest voter block which is mined after round $\Rl$ and  is at least $k$-deep on the main chain. Here
$k\geq  \frac{(1-2\beta)\fv \Dr}{24} $ and is also greater than $\kmin$.
\end{lemma}
\begin{proof}
From the typicality Lemma \ref{lem:typical_events}, we know that the  event $T$ occurs w.p $1-\eps_m$.  Using the chain-growth Lemma \ref{lem:frac_chain_growth} under the event $T$, we know  that by round $r$, $1-\rho_{\Dr}$ fraction of the voter blocktree's main chain grows by $k_1 \geq \frac{\Dr\fv}{3}$ voter blocks. Let $r'=\frac{k_1}{2\fv}$.
Next, using chain-quality Lemma  \ref{lem:frac_chain_quality}  under the event $ T$, we know that for at least $1-\delta_{k_1}$ fraction of voter blocktrees,  the deepest of these honest voter blocks, mined after round $R_{\ell}$, is at least $k$-deep, where $k\geq \frac{(1-2\beta)k_1}{8} \geq  \frac{(1-2\beta)\fv \Dr}{24} $.
Therefore, at least $1-\rho_{\Dr}-\delta_{k_1}$ fraction of the blocktrees have an honest voter block mined after round $\Rl$ which is at least $k$-deep on the main chain. The constants satisfy $\delta_{k_1} = \rho_{\frac{\Dr}{3}} < 3\rho_{\Dr} $ and this completes the proof.
It is important to note that the depth of votes on all the $m$ voter blocktree are \textit{observable} by the users.
\end{proof}

Define random variable  $N_{\ell}(r) := |P_{\ell}(r) |$ as the number of proposer blocks on level $\ell$ at round $r$ and let $c_1 = \cone$ and $c_2 := \frac{16}{\bar{f}_v(1-2\beta)^3}$.

\begin{lemma} \label{lem:well_defined_proposer_list}
The proposer list at level $\ell$ can be confirmed w.p $1-\epsilon_m$ in round $r=R_{\ell}+\Dr$ for for $\Dr \geq \Delta_0$ if
\begin{align}
    & \text{Case } 1.\;\;N_{\ell}(R_{\ell}+\Dr)+1 < \frac{c_1}{\rho_{\Dr}}, \label{eqn:case1}  \\
    \text{Or} \qquad
    & \text{Case } 2.\;\;
    \Dr = c_2m \nonumber. 
    \end{align}
\end{lemma}
\begin{proof}
Let us first consider Case $1$. All the events here are $1-\eps_m$ probability events. From Lemma \ref{lem:fast_latency_1}, we know that by round $r=R_{\ell}+\Dr$,
at least $1-4\rho_{\Dr}$ fraction of voter blocktrees have $k$-deep votes on proposer blocks in $\P_{\ell}(r)$ where $k\geq  \frac{(1-2\beta)\fv \Dr}{24} $. This implies
$
    \sum_{p_n \in \P_{\ell}(r) } V_n^k(r) \geq m(1-4\rho_{\Dr})
$
and from Lemma \ref{lem:votes_bounds}, we have
\begin{equation}
    \sum_{p_n \in \P_{\ell}(r) } \underbar{V}_n(r) \geq m(1-4\rho_{\Dr}-\delta_k), \label{eqn:lower_bound_local_1}
\end{equation}
where the constant $\delta_k$ satisfies $\delta_k \leq \frac{12\rho_{\Dr}}{(1-2\beta)}$. Without loss of generality we  assume  $\underbar{V}_1(r)\geq \underbar{V}_i(r) \;\forall p_i\in \P_{\ell}(r)$, and therefore from Equation \eqref{eqn:lower_bound_local_1}, we have 
\begin{equation}
\underbar{V}_1(r) \geq \frac{m}{N_{\ell}(r)}\left(1-\frac{16\rho_{\Dr}}{1-2\beta}\right)\footnote{Note that this inequality is extremely weak.}. \label{eqn:lower_bound_inequality_local}
\end{equation}
On the other hand, the upper bound on the votes on a private proposer block, $p_{\text{private}}$,  by Equation \eqref{eqn:private_block_upper_bound} is :
\begin{align}
\overline{V}_{\text{private}}(r) & < m-\sum_{p_n \in \mathcal{P}_{\ell}(r)} \underbar{V}_n(r)  \nonumber \\
& \overset{(a)}{<} \frac{(1-2\beta)\rho_{\Dr}}{16}, \label{eqn:private_block_upperbound_local}
\end{align}
where the inequality (a) follows from from Equation \eqref{eqn:lower_bound_local_1}.
From Equations \eqref{eqn:lower_bound_inequality_local} and \eqref{eqn:private_block_upperbound_local}, it is easy to see that
\begin{equation}
    N_{\ell}(r)+1 < \frac{16}{(1-2\beta)\rho_{\Dr}} \Longrightarrow \underbar{V}_1(r) > \overline{V}_{\text{private}}(r), \nonumber
\end{equation}
and therefore, the proposer list at level $\ell$ can be confirmed by round $r$. This proves the claim in Case 1.
Now let us consider Case 2. From the proof of Theorem \ref{thm:common_prefix_1}, we know that \textit{all} the $m$ votes are permanent w.p $1-\epsilon$ for
$$
r(\epsilon)=\frac{1024}{\fv(1-2\beta)^3 }\log \frac{8m\rmax}{\epsilon}.
$$
Substituting $\epsilon = \eps_m$ in the above equation, we conclude that for $r(\eps_m)=c_2 m$, the upper bound on the number of votes on private block, $\overline{V}_{\text{private}}(r) = 0$ and $\underbar{V}_1(R_{\ell}+k) \geq 1 >\overline{V}_{\text{private}}(R_{\ell}+k)$ w.p $1-\eps_m$.

\end{proof}

We now use the above Lemma \ref{lem:well_defined_proposer_list} to calculate the expected number of rounds to confirm the proposer block list at level $\ell$. 
For Case $1$ \eqref{eqn:case1} let us define the random variable:

\begin{equation}
    \Rc_{\ell} := \min \Dr >\Delta_0 \; s.t \; N_{\ell}(R_{\ell}+\Dr)+1 < \frac{c_1}{\rho_{\Dr}}. \label{eqn:martingale_1}
\end{equation}
Note that $\Rc_{\ell} = \infty$ if the inequality condition in Equation \eqref{eqn:martingale_1} is not satisfied for any $\Dr$. From Lemma \ref{lem:well_defined_proposer_list}, the proposer list at level $\ell$ can be confirmed in $\min(\Rc_{\ell},c_2 m)$ rounds and the next lemma calculates its expectation.

\begin{figure}
  \centering
  \includegraphics[width=.8\linewidth]{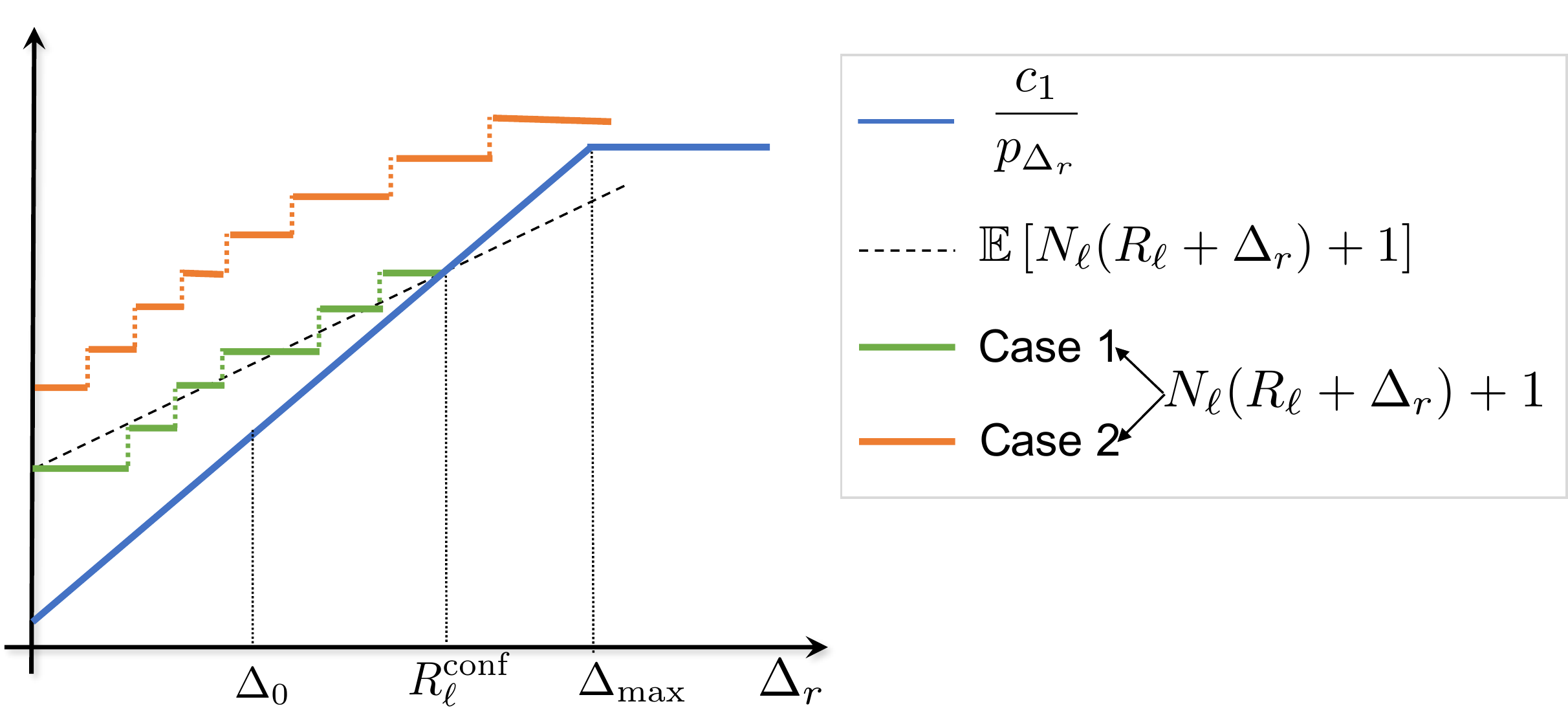}
  \caption{ Sample paths of rv $N_{\ell}(R_{\ell}+\Dr)+1$ falling under Case 1 and Case 2 in Lemma \ref{lem:well_defined_proposer_list}. The values of $\frac{c_1}{\rho_{\Dr}} =   \min \left(1+4\fv\Dr, 1+32\log m \right),\;\Delta_0 = \frac{12}{1-2\beta}\rmin,\;\Dmax = 8\log m$.
 \label{fig:latency_proof_prop}}
 \end{figure}


\begin{lemma} \label{lem:List_decoding_level_1}
The proposer list at level $\ell$ can be  confirmed by round $R_{\ell}+\min(\Rc_{\ell},c_2 m)$ and we have 
$$\mathbb{E}[\min(\Rc_{\ell}, c_2 m)] \leq \frac{13}{(1-2\beta)}\rmin+\frac{48}{\fv(1-2\beta)^3m^3}.$$
\end{lemma}

\begin{proof}
The honest users do not mine new proposer blocks on level $\ell$ after round $R_\ell$, however, the adversary could potentially mine new proposer blocks on level $\ell$ after round $R_\ell$. Therefore, the random variable $N_{\ell}(R_{\ell}+\Dr)$ satisfies
$$N_{\ell}(R_{\ell}+\Dr) \leq H^p[R_{\ell}] + W^p_{\ell}(R_{\ell})+Z^p_{\ell}[R_{\ell},R_{\ell}+\Dr].$$ 
\begin{enumerate}
    \item $H^p[R_{\ell}]$ corresponds to the number of  proposer blocks mined by the honest users on level $\ell$. From Section \ref{sec:model}, we know that  $H^p[R_{\ell}] \sim \pois((1-\beta)\fv)$. 
    \item $W^p_{\ell}(R_{\ell})$ denotes the upper bound on number of proposer blocks at level $\ell$ in store by the adversary by round $R_{\ell}$. It is shown in Appendix \ref{sec:reserve_blocks} that $W^p_{\ell}(R_{\ell}) \sim \geom(1-2\beta)$.  
    \item  $Z^p_{\ell}[R_{\ell},R_{\ell}+\Dr]$ denotes the number of proposer blocks mined by the adversary at level $\ell$ in the interval $[R_{\ell},R_{\ell}+\Dr]$. From Section \ref{sec:model}, we know that  $Z^p_{\ell}[R_{\ell},R_{\ell}+\Dr] \sim \pois(\fv\beta \Dr)$. 
\end{enumerate}
The mean of random variable $N_{\ell}(R_{\ell}+\Dr)$ is affine in $\Dr$, and $\frac{c_1}{\rho_{\Dr}}$  is also affine in $\Dr$ with a higher slope (by design).
Therefore, intuitively the expected value of $\Rc_{\ell}$ defined in Equation \eqref{eqn:martingale_1} should be constant which depends only on $\beta$. Two examples are illustrated in Figure \ref{fig:latency_proof_prop}. We now formalize this intuition.
Let us define $\Dmax =  \frac{8\log m}{\fv(1-2\beta)}$. We will calculate $\Pr(\Rc_{\ell} > \Dr)$ separately for three intervals: $[0,\Delta_0], (\Delta_0,\Dmax), [\Dmax,\infty)$.
\begin{enumerate}
\item \textit{Interval $[0,\Delta_0)$:} Since $\Rc_{\ell} \geq \Delta_0$ by definition, we have
\begin{equation}
\Pr(\Rc_{\ell} > \Dr)=1 \quad \forall\Dr \leq \Delta_0.  \label{eqn:local_minus_1}
\end{equation}

\item \textit{Interval $[\Dmax,\infty)$:} For $\Dr\geq \Dmax$, we have
\begin{align}
\{\Rc_{\ell} > \Dr \} 
&=  \bigcap_{x\leq\Dr}\left\{N_{\ell}(R_{\ell}+x)+1 >\frac{c_1}{\rho_{\Dr}} \right\} \nonumber \\
&\subseteq \bigcap_{x\leq\Dmax} \left\{N_{\ell}(R_{\ell}+\Dmax)+1 >\frac{c_1}{\rho_{\Dmax}} \right\} \nonumber\\
& = \{\Rc_{\ell} > \Dmax \}.
\end{align}
\noindent This implies 
\begin{equation}
\Pr\{\Rc_{\ell} > \Dr \} \leq \Pr\{\Rc_{\ell} > \Dmax \} \qquad \forall \Dmax\leq\Dr. \label{eqn:local_0}    
\end{equation}
\item \textit{Interval $(\Delta_0,\Dmax)$:}
Using Equation \eqref{eqn:constants}, we have
\begin{align}
\frac{c_1}{\rho_{\Dr}} = 1+4\fv \Dr \quad \forall \Dr<\Dmax. \label{eqn:local_1}
\end{align}
\noindent 
For $\Delta_0 < \Dr < \Dmax $, we bound the tail event:
\begin{align}
\{\Rc_{\ell} > \Dr \} 
&=  \bigcap_{x\leq\Dr}\left\{N_{\ell}(R_{\ell}+x)+1 >\frac{c_1}{\rho_{\Dr}} \right\} \nonumber \\
&\subseteq  \left\{N_{\ell}(R_{\ell}+\Dr)+1 >\frac{c_1}{\rho_{\Dr}} \right\} \nonumber \\
&\subseteq  \left\{H^p[R_{\ell}] + W^p_{\ell}(R_{\ell})+Z^p_{\ell}[R_{\ell},R_{\ell}+\Dr]+1 >\frac{c_1}{\rho_{\Dr}} \right\}\nonumber\\
&=  \bigg\{\left(H^p[R_{\ell}]  - \E[H^p[R_{\ell}]]\right) + W^p_{\ell}(R_{\ell})+\left(Z^p_{\ell}[R_{\ell},R_{\ell}+\Dr]-\E[Z^p_{\ell}[R_{\ell},R_{\ell}+\Dr]]\right)  \nonumber \\
&\quad\quad >\frac{c_1}{\rho_{\Dr}} - \big(1+\E[H^p[R_{\ell}]]+\E[Z^p_{\ell}[R_{\ell},R_{\ell}+\Dr]]\big) \bigg\} \nonumber\\
&\overset{(a)}{=}  \bigg\{\left(H^p[R_{\ell}]  - \E[H^p[R_{\ell}]]\right) + W^p_{\ell}(R_{\ell})+\left(Z^p_{\ell}[R_{\ell},R_{\ell}+\Dr]-\E[Z^p_{\ell}[R_{\ell},R_{\ell}+\Dr]]\right)  \nonumber \\
&\quad\quad >1+4\fv\Dr - \big(1+(1-\beta)\fv+\beta\fv\Dr\big) \bigg\} \nonumber\\
&\subseteq \bigg\{\left(H^p[R_{\ell}]  - \E[H^p[R_{\ell}]]\right) + W^p_{\ell}(R_{\ell})+\left(Z^p_{\ell}[R_{\ell},R_{\ell}+\Dr]-\E[Z^p_{\ell}[R_{\ell},R_{\ell}+\Dr]]\right)  \nonumber \\
&\quad\quad >\big(\fv\Dr+\fv\Dr +\fv \Dr\big) \bigg\}  \nonumber\\
\Rightarrow \{\Rc_{\ell} > \Dr \} &\subseteq \tF_1\cup\tF_2\cup\tF_3, \label{eqn:local_2}
\end{align}
where the events are:
\begin{align*}
    \tF_1 &:= \left\{ H^p[R_{\ell}]  - \E[H^p[R_{\ell}]] \geq \fv\Dr  \right\}\\
    \tF_2 &:= \left\{ W^p_{\ell}(R_{\ell})  \geq \fv\Dr  \right\}\\
    \tF_3 &:= \left\{Z^p_{\ell}[R_{\ell},R_{\ell}+\Dr]-\E[Z^p_{\ell}[R_{\ell},R_{\ell}+\Dr]] > \fv \Dr\right\}.
\end{align*}
The equality (a) follows from Equation \eqref{eqn:local_1}. Using Chernoff bounds, we upper bound the probabilities of events the $\tF_1$, $\tF_2$ and $\tF_3$:
\begin{align}
    \Pr(\tF_1) &\leq e^{-\frac{\fv\Dr}{2}}\label{eqn:local_3}\\
    \Pr(\tF_2) &\leq (2\beta)^{\fv\Dr} \leq e^{-\frac{(1-2\beta)\fv\Dr}{2}}\label{eqn:local_4}\\
    \Pr(\tF_3) &\leq e^{-\frac{\fv\Dr}{2}}.\label{eqn:local_5}
\end{align}
From  Equations \eqref{eqn:local_2}, \eqref{eqn:local_3}, \eqref{eqn:local_4} and \eqref{eqn:local_5}, for $\Delta_0< \Dr < \Dmax$, we have 
\begin{align}
    \Pr(\{\Rc_{\ell} > \Dr \}) &\leq e^{-\frac{\fv\Dr}{2}}+e^{-\frac{(1-2\beta)\fv\Dr}{2}}+e^{-\frac{\fv\Dr}{2}}\nonumber\\
    & \leq 3e^{-\frac{(1-2\beta)\fv\Dr}{2}}.
     \label{eqn:local_6}
\end{align}

\end{enumerate}
From  Equations \eqref{eqn:local_minus_1}, \eqref{eqn:local_0} and \eqref{eqn:local_6}, we have
\begin{equation}
    \Pr(\Rc_{\ell} > \Dr) \leq
     \begin{cases}
     1 &\Dr \leq \Delta_0\\
     3e^{-\frac{(1-2\beta)\fv\Dr}{2}} & \Delta_0< \Dr < \Dmax\\
     3e^{-\frac{(1-2\beta)\fv\Dmax}{2}} \quad & \Dmax\leq \Dr
     \end{cases} \label{eqn:prop_prob_bound}
\end{equation}
\begin{figure}[H]
  \centering
  \includegraphics[width=.6\linewidth]{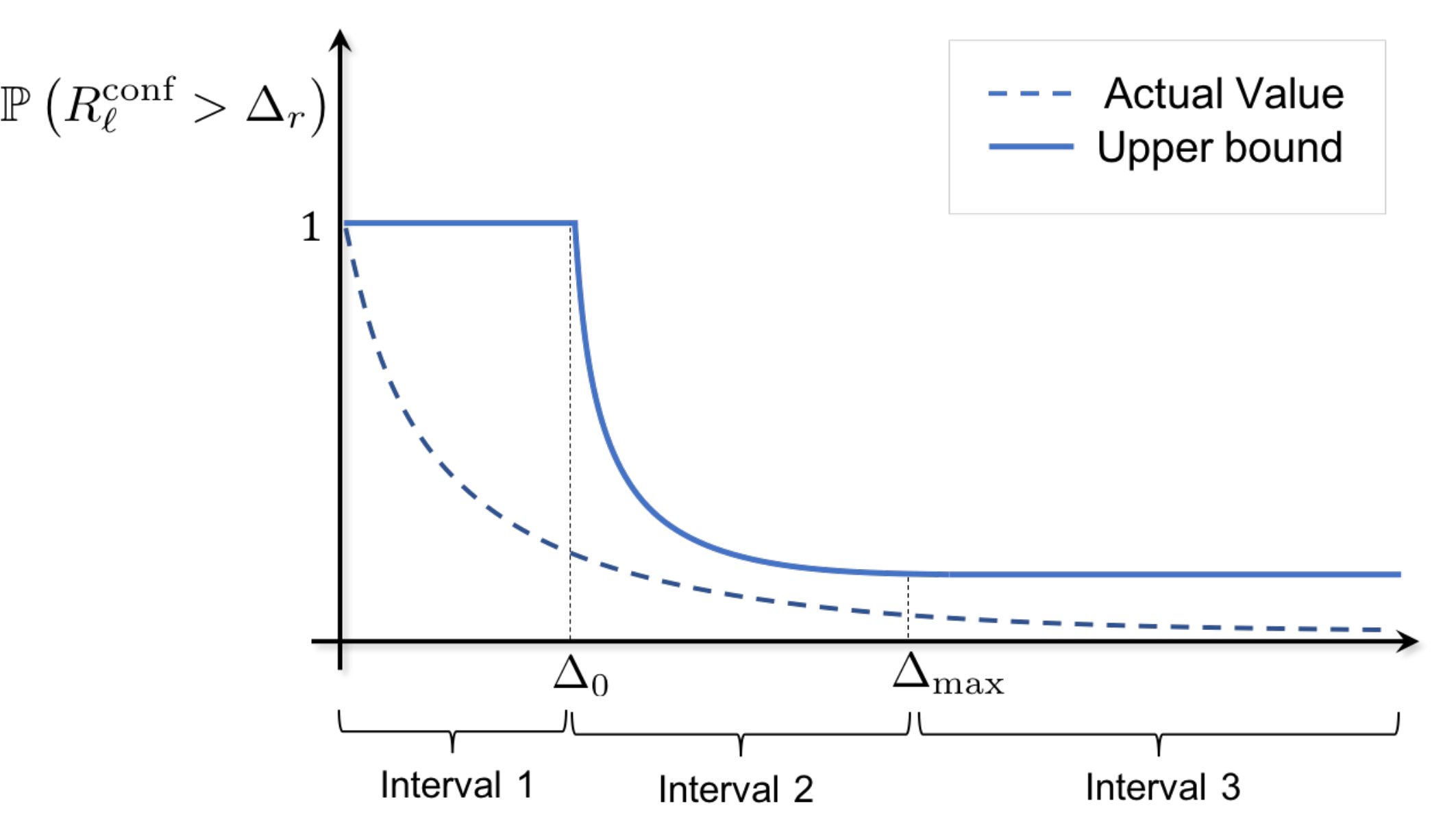}
  \caption{   \label{fig:single_level_proof}}
 \end{figure}
 
Using the above expression, the expectation of $\max(\Rc_{\ell},c_2m)$ is given by
\begin{align}
    \E[\min(\Rc_{\ell},c_2m)] &= \sum_{\Dr=0}^{\Dmax} \Pr(\Rc_{\ell} > \Dr)+c_2m\nonumber\Pr(\Rc_{\ell} > \Dmax)\\
    &\leq \Delta_0+\sum_{\Dr=\Delta_0}^{\Dmax} \Pr(\Rc_{\ell} > \Dr)+3c_2me^{-(1-2\beta)f_v\Dmax/2}\nonumber\\
    &\leq \Delta_0+\sum_{\Dr=\Delta_0}^{\infty}\left(3e^{-\frac{(1-2\beta)\fv\Dr}{2}}\right)+ 3c_2me^{-4\log m} \nonumber\\
     & = \Delta_0 + \frac{6e^{-\frac{(1-2\beta)\fv\Delta_0}{2}}}{(1-2\beta)\fv}+\frac{3c_2}{m^3}\nonumber \\
    & \leq \frac{12}{(1-2\beta)}\rmin+ \frac{6}{(1-2\beta)\fv}+\frac{3c_2}{m^3} \nonumber\\
      \E[\min(\Rc_{\ell},c_2m)] & \leq \frac{13}{(1-2\beta)}\rmin+\frac{48}{\fv(1-2\beta)^3m^3}  \ \label{eqn:local_final}.
\end{align}
\end{proof}

\noindent  Lemma \ref{lem:List_decoding_level_1} upper bounds the expected number of rounds to confirm the proposer block list at level $\ell$.  However, our goal in Lemma \ref{lem:leader_sequence_list_decoding} is to conform proposer list for \textit{all} the levels $\ell'\leq \ell$.  From Lemma \ref{lem:well_defined_proposer_list}, we know that the proposer list at level $\ell'$ is confirmed by round $R_{\ell'}+\min(\Rc_{\ell},c_2m)$.  Therefore, all the proposer list up to level $\ell$ are confirmed in the following number of rounds:
\begin{align}
    \Rconf_{\ell} :=&\max_{\ell'\leq\ell}\big(R_{\ell'}+\min(\Rc_{\ell'}, c_2 m)-R_{\ell}\big)\nonumber\\
    =& \max_{\ell'\leq\ell}\big(\min(\Rc_{\ell'}, c_2 m)-\Dll\big), \label{eqn:ultra_local_1}
\end{align}
where $\Dll = R_{\ell}-R_{\ell'}$. Expression \eqref{eqn:ultra_local_1} is a maximum of  random variables associated with each level up to level $\ell$. It turns out max is dominated by random variable associated with level $\ell$ and in fact it's expectation, calculated in the next lemma, is very close to expectation of $\min(\Rc_{\ell}, c_2 m)$. We now calculate the expectation of the random variable in expression \eqref{eqn:ultra_local_1}.
\begin{lemma} \label{lem:conf_upto_level}
All the proposer lists up to level $\ell$ will get confirmed in the following number of rounds in expectation:
\begin{align*}
    \E[\Rconf_{\ell}]  &\leq  \frac{13}{(1-2\beta)}\rmin+\frac{256}{(1-2\beta)^6\fv m^2}  \\
    &\leq  \frac{2808}{(1-2\beta)^3\fv}\log \frac{50}{(1-2\beta)}+\frac{256}{(1-2\beta)^6\fv m^2}.
\end{align*}
\end{lemma}
\begin{proof}
Let us define
\begin{align}
    F(\Dllf) := & \E\left[\Rconf_{\ell} | \big\{\Dll\big\}_{\ell'\leq \ell}\right] \nonumber\\ =&\E\left[\max_{\ell'\leq\ell}\big(\min(\Rc_{\ell'}, c_2 m)-\Dll\big)\big|\big\{\Dll\big\}_{\ell'\leq \ell}\right]  \nonumber \\
    \leq &\Delta_0+\E\left[\max_{\ell'\leq\ell}\big(\min(\Rc_{\ell'} - \Delta_0, c_2m)-\Dll\big)\big|\big\{\Dll\big\}_{\ell'\leq \ell}\right]  \nonumber \\
    \leq& \Delta_0+\sum_{\ell'\leq\ell}\E\left[\left(\min(\Rc_{\ell'}-\Delta_0, c_2 m)-\Dll\right)_+\big|\Dll\right]. \label{eqn:2_local_1}
\end{align}
We bound each term in the summation the Equation \eqref{eqn:2_local_1} similar to steps used to Equations \eqref{eqn:local_final}.

\begin{align}
    &\quad \E\left[\left(\min(\Rc_{\ell'}-\Delta_0, c_2 m)-\Dll\right)_+\big|\Dllf\right]  \\
    &=\sum_{\Dr=\Dll+\Delta_0}^{\Dmax} \Pr(\Rc_{\ell'} > \Dr|\Dllf)+(c_2m-\Dll)_+\nonumber\Pr(\Rc_{\ell'} > \Dmax|\Dllf)\\
    &\overset{(a)}{=} \sum_{\Dr=\Dll+\Delta_0}^{\Dmax} \Pr(\Rc_{\ell'} > \Dr)+(c_2m-\Dll)_+\nonumber\Pr(\Rc_{\ell'} > \Dmax)\\
    &\leq \sum_{\Dr=\Dll}^{\infty}\left(3e^{-\frac{(1-2\beta)\fv\Dr}{2}}\right)+ 3(c_2m-\Dll)_+e^{-4\log m} \nonumber\\
    & \leq \frac{6e^{-\frac{(1-2\beta)\fv\Dll}{2}}}{(1-2\beta)\fv}+\frac{3(c_2m-\Dll)_+}{m^4} 
       \label{eqn:2_local_final}.
\end{align}
The inequality (a) follows because the random variable $\Rc_{\ell'}$ is independent of proposer block mining on levels other than $\ell'$ and
depends only on the mining on voting blocktrees and proposer blocks on level $\ell'$. Using Equation \eqref{eqn:2_local_final} in Equation \eqref{eqn:2_local_1} we get
\begin{align}
     F\left(\Dllf\right) \leq &\Delta_0+\sum_{\ell'\leq\ell}  \frac{6e^{-\frac{(1-2\beta)\fv\Dll}{2}}}{(1-2\beta)\fv}+\frac{3(c_2m-\Dll)_+}{m^4}. \label{eqn:2_local_2}
\end{align}
Intuitively, if the first proposer block on every level is mined by the honest users then $\Dll$ is a geometric random variable with mean $\frac{2(\ell-\ell')}{\fv}$ i.e, linear in $\ell-\ell'$. Taking expectation on Equation \eqref{eqn:2_local_2} and substituting $\Dll$ with $\frac{2(\ell-\ell')}{\fv}$ would give us a finite bound. However this intuition is incorrect because the adversary could  present proposer blocks on \textit{multiple} levels in the same round and thus the value of $\Dll$ depends on the adversarial strategy. We overcome this problem by showing that irrespective of the adversary's strategy, the honest users will propose the first proposer blocks for sufficient number of levels.

Let levels $\{L_1, L_2,\cdots,L_{i},\cdots,L_n\}$ be the levels lesser than $\ell$ on which the honest users  presented the first  proposer block. Let $L_{n+1}=\ell$. Here $L_i$'s are a random variables and the first proposer block at level $L_i$ is produced in round $R_{L_{i}}$.  If the adversary produces the first proposer block at level $\ell'$ for $L_{i} < \ell' < L_{i+1}$, then from the monotonicity of the growth of the proposer blocktree, we have the following constraint $R_{L_{i}} \leq R_{\ell'} \leq R_{L_{i+1}}$. Let us use this in Equation \eqref{eqn:2_local_2}.
\begin{align}
    &\quad F\left(\Dllf\right) \nonumber\\
    &\leq \Delta_0+\sum_{\ell'\leq\ell}  \frac{6e^{-\frac{(1-2\beta)\fv\Dll}{2}}}{(1-2\beta)\fv}+\frac{3(c_2m-\Dll)_+}{m^4}.\nonumber\\
  \nonumber  \\
    &\leq \Delta_0+\sum_{i \in [n] } \sum_{L_i < \ell' \leq L_{i+1}}  \frac{6e^{-\frac{(1-2\beta)\fv \Dll}{2}}}{(1-2\beta)\fv}+\frac{3(c_2m-\Dll)_+}{m^4} \nonumber\\
    &\overset{(a)}{\leq} \Delta_0+ \sum_{i \in [n] } (L_{i+1}-L_i)\left(  \frac{6e^{-\frac{(1-2\beta)\fv D_{L_{n+1}, L_{i+1}}}{2}}}{(1-2\beta)\fv}+\frac{3(c_2m-{D_{L_{n+1}, L_{i+1}}})_+}{m^4}\right).\label{eqn:final_latency_random_2}
\end{align}
\noindent The inequality $(a)$ follows because $R_{\ell'} \leq R_{L_{i+1}}$.
Let $G_j$  be i.i.d random variables s.t $G_j \sim \geom({\fv})$. Since the levels $L_i$ and $L_{i+1}$  are mined by the honest users, we have $D_{L_{i+1},L_{i}} \geq \sum_{j=L_i}^{L_{i+1}}G_j$ and $D_{L_{n+1},L_{i}} = \sum_{j=L_{i+1}}^{L_{n+1}}G_j$. Using this in Equation \eqref{eqn:final_latency_random_2}, we get  
\begin{equation*}
     F\left(\Dllf\right) \leq \Delta_0+
     \sum_{i \in [n] } (L_{i+1}-L_i)\left( \frac{6e^{-\frac{(1-2\beta)\fv\sum_{j=L_{i+1}}^{L_{n+1}}G_j}{2}}}{(1-2\beta)\fv}+\frac{3(c_2m-\sum_{j=L_{i+1}}^{L_{n+1}}G_j)_+}{m^4}\right).
\end{equation*}
We now take expectation over $G_j$'s gives us
\begin{align*}
    \mathbb{E}\left[F\left(\Dllf\right) \big| \{L_i\}_{i=1}^{n} \right] 
    &\leq \Delta_0+\sum_{i \in [n] } (L_{i+1}-L_i) \bigg( \frac{6}{(1-2\beta)\fv}\Big(\frac{1}{1+(1-2\beta)}\Big)^{\frac{2(L_{n+1}-L_{i+1})}{\fv}}\\
    & \qquad \qquad +\frac{3(c_2 m-\frac{L_{n+1}-L_{i+1}}{\fv})_+}{m^4}\bigg).\\ 
\end{align*} 
Since the honest user have $1-\beta$ fraction of mining power, we have $(L_{i+1}-L_i) \sim  \geom(1-\beta)$ and on taking expectation over  $L_i$'s we get:
\begin{align*}
    \E\left[\Rconf_{\ell}\right] &= \mathbb{E}\left[F\left(\Dllf\right) \right] \\
    &\leq  \Delta_0+\frac{1}{1-\beta}\sum_{i \in [n] } \left( \frac{6}{(1-2\beta)\fv}\Big(\frac{1}{1+(1-2\beta)}\Big)^{\frac{(n-i)}{\fv}} +\frac{(c_2 m-\frac{(n-i)}{\fv})_+}{m^4}\right) \\
    &\leq \Delta_0+\frac{1}{1-\beta}\sum_{i =0}^{\infty} \Big( \frac{6}{(1-2\beta)\fv}\Big(\frac{1}{1+(1-2\beta)}\Big)^{\frac{i}{\fv}} +\frac{(c_2 m-\frac{i}{\fv})_+}{m^4}\Big) \\
    & \leq  \frac{2}{(1-2\beta)}\rmin+2\Big( \frac{6}{(1-2\beta)^2\fv}+\frac{128}{(1-2\beta)^6\fv m^2}\Big) \nonumber \\
    & \leq  \frac{13}{(1-2\beta)}\rmin+\frac{256}{(1-2\beta)^6\fv m^2} \nonumber \\
    & \leq  \frac{2808}{(1-2\beta)^3\fv}\log \frac{50}{(1-2\beta)}+\frac{256}{(1-2\beta)^6\fv m^2}
    .
       \label{eqn:final_latency_bound}
\end{align*}
\end{proof}

\section{Fast confirmation for honest transactions: proof of Theorem 4.7}
\label{thm:honest_tx_latency_proof}
This section uses ideas from the proof of Lemma \ref{lemma:ls_slow_quality}.
Let the transaction $tx$ enters the system\footnote{As a part of a transaction block.} in round $r$ and let $\ell$ be the last level on the proposer blocktree which has proposer blocks at round $r$. Define
\begin{align*}
    \ell^* &:= \max \big( \tilde{\ell} \leq \ell \;\;s.t\;\; \text{ the honest users mined the first proposer block on level } \tilde{\ell}\big)
\end{align*}
Let $r^*$ be the round in which the first proposer block was mined on level $\ell^*$.
From the definition of $\ell^*$ we have the following two observations:
\begin{enumerate}
    \item All the proposer blocks on levels greater or equal to  $\ell^*$ are mined on or after round $r^*$ because by definition there are no proposer blocks on level $\ell^*$ before round $r^*$ and hence no user can mine a proposer block on a level greater than $\ell^*$ before round $r^*$.
    \item The adversary has mined at least one proposer block on all levels in $[\ell^*,\ell]$.
\end{enumerate}
Define $\Delta_0 := \frac{12\rmin}{1-2\beta}$. For $r_f\geq r$, let us define the following event:
\begin{equation}
       A_{r_f} = \big\{Y^p[r^*,r_f-\Delta_0] -Z^p[r^*,r_f] > 0 \big\}.
\end{equation}
\begin{lemma} \label{lem:A_r_f}
If event $A_{r_f}$ occurs, then the transactions $tx$ is included in a block $b$ which is proposed in round $r(b)\leq r_f-\Delta_0$ and confirmed as a leader block by round $r_f$. 
\end{lemma}
\begin{proof}

From our first observation, $Y^p[r^*,r_f-\Delta_0] > Z^p[r^*,r_f]$ implies that by round $r_f$ there exists a level $\tilde{l} \geq \ell^*$ which has only one honest proposer block proposed in interval $[r^*, r_f-\Delta_0]$.
Our second observation says that the adversary has mined a proposer block on all levels in $[\ell^*,\ell]$ and therefore, we have $\tilde{\ell}>\ell$.
From Lemma \ref{lem:List_decoding_level_1}, the single proposer block at level $\tilde{\ell}$ is confirmed as a final leader block of its level w.p $1-\eps_m$ by round $r_f$. Since this proposer block was mined after round $r$, it will include the transaction $tx$. 
\end{proof}
\noindent Let us define the following random variable:
\begin{equation*}
    R_f := \min r_f \geq r \;\;s.t\;\; A_{r_f} \;\;\text{occurs}.
\end{equation*}
\begin{lemma} \label{lem:R_f-r}
\begin{equation}
 \mathbb{E}[R_f-r] \leq   \frac{24(1-\beta)\rmin}{(1-2\beta)^2}  \leq  \frac{2592}{(1-2\beta)^3\fv}\log \frac{50}{(1-2\beta)}.
\end{equation}
\end{lemma}
\begin{proof}
Consider the following random walk
\begin{align}
    W_{r_f} 
    &:= Y^p[r^*+\Delta_0,r_f] -Z^p[r^*,r_f-\Delta_0].
\end{align}
and a random variable $V\sim \bin(\Delta_0, \fv/2)$ which is independent of $W_{r_f}$. It is easy to see that $Y^p[r^*+\Delta_0,r_f] -Z^p[r^*,r_f-\Delta_0] \overset{d}{=}W_{r_f} - V$ in distribution. Therefore, event $A_{r_f}$ implies $W_{r_f}>V$ and we have
\begin{equation*}
    R_f = \min r_f \geq r \;\;s.t\;\; W_{r_f}>V \;\;\text{occurs}.
\end{equation*}
The random walk $W_{r_f}$ has a positive drift of $\frac{(1-2\beta)\fv}{2}$.
For a fixed value of $V$, the conditional expectation is
\begin{align*}
    \mathbb{E}[R_f-r^*|V] = \Delta_0+\frac{2V}{(1-2\beta)\fv}.
\end{align*}
Taking expectation on $V$, we get
\begin{equation}
    \mathbb{E}[R_f-r^*] = \Delta_0+\frac{\Delta_0}{1-2\beta} = \frac{24(1-\beta)\rmin}{(1-2\beta)^2}.
\end{equation}
Since $r^*\leq r$, we have $\mathbb{E}[R_f-r] \leq \frac{24(1-\beta)\rmin}{(1-2\beta)^2}$. Therefore,  the transaction $tx$ is included in all the ledgers in less than $\frac{24(1-\beta)\rmin}{(1-2\beta)^2}$ rounds in expectation. Substituting $\rmin$ from \eqref{eqn:constants} give us the required result.
\end{proof}
From Lemma \ref{lem:A_r_f} and \ref{lem:R_f-r}, we conclude that a transaction, which is part of a transaction block mined in round $r$, is referred by a proposer block at level (say) $\ell$ and the leader block at this level confirmed before round $r+\frac{2592}{(1-2\beta)^3\fv}\log \frac{50}{(1-2\beta)}$ in expectation. This proves the main claim of Theorem \ref{thm:honest_tx_latency}.

\section{Others}

\begin{figure}
\hspace{-.65in}
   \includegraphics[width=0.8\linewidth]{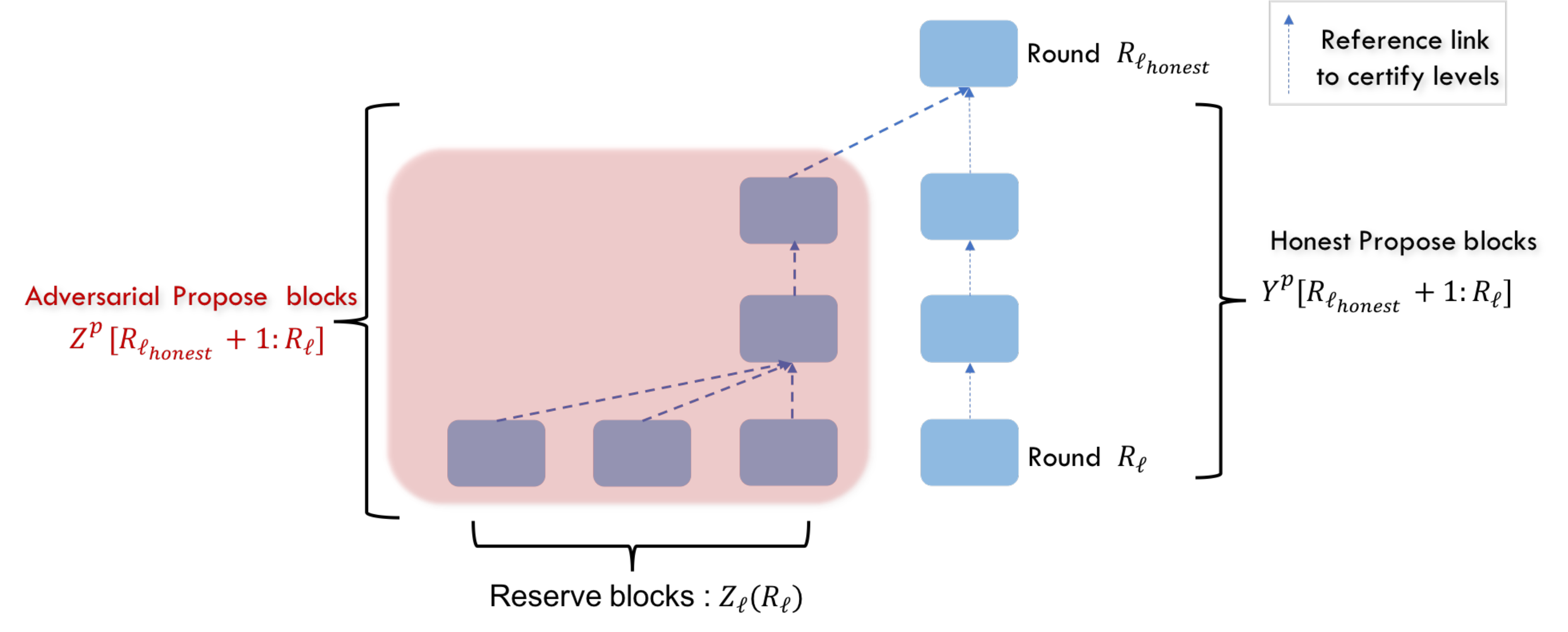}
   \caption{Number of reserved blocks by the adversary on level $\ell$ in round $\R_{\ell}.$ \label{fig:reserve}}
 \end{figure}
 
\subsection{Reserve proposer blocks by the adversary}\label{sec:reserve_blocks}
Say the honest users mine the first proposer block at level $\ell$ in round $R_{\ell}$. Let $W_{\ell}^p(R_{\ell})$ denote that the number of hidden proposal blocks blocks on level $\ell$ by the adversary. In order to maximize  $W_{\ell}^p(R_{\ell})$, all these hidden proposer blocks should have a common honest parent proposer block at level (say) $\ell_{honest}$ linked via private proposal blocks as shown in the Figure \ref{fig:reserve}. The total number of reserve blocks is given by
\begin{equation}
W_{\ell}^p(R_{\ell}) = \max_{\ell_{honest}\leq \ell} Z^p[R_{\ell_{honest}}+1, R_{\ell}]-Y^p[R_{\ell_{honest}}+1, R_{\ell}]+1.
\end{equation}

The random variable $Y^p[R_{\ell_{honest}}, R_{\ell}]-Z^p[R_{\ell_{honest}}, R_{\ell}]$ is a random walk in the variable $\ell_{honest}$ with a net drift of $\frac{(1-2\beta)f_v}{2}$. The ratio of left drift to the right drift is $2\beta$ and from \cite{randomwalk}, we have 

\begin{align*}
    \mathbb{P}(W_{\ell}^p(R_{\ell}) > k) &=  \mathbb{P}\big(\max_{\ell_{honest}\leq \ell}Z^p[R_{\ell_{honest}}+1, R_{\ell}]-Y^p[R_{\ell_{honest}}+1, R_{\ell}]\geq k\big) \\
     &=(2\beta)^{k}.
\end{align*}
Therefore $W_{\ell}^p(R_{\ell}) \sim \geom(1-2\beta)$.

\subsection{Random walk proofs}
Consider the following event from Equation \eqref{eqn:single_chain_event}
\begin{align*}
    \tE_1\left[r - r', r\right] :=& \bigcap_{a,b\geq 0}\left\{ Y[r - r'-a, r+b] - Z[r - r'-a, r+b] > \frac{(1-2\beta)k}{8} \right\},\nonumber
\end{align*}
for $r'=\frac{k}{2\fv}$. The random variable $W[r-r',r] =  Y[r-r',r] - Z[r-r',r]$ is  a random walk with drift $\frac{(1-2\beta)f}{2}$. 

\begin{lemma} \label{lem:RW_1}
If $W[r-r',r] > c_1k$, for $c_2<c_1$ we have
\begin{align*}
    \mathbb{P}\big(W[r-r', r+a] \geq c_2k \; \forall\; a >0\;\big|\; W[r-r',r] > c_1k) &= 1 - (2\beta)^{(c_1-c_2)k}\\
    &= 1  - e^{\log (2\beta)(c_1-c_2)k}.
\end{align*}
\end{lemma}
\begin{proof}
Refer \cite{randomwalk}.
\end{proof}

If the random walk is to the right of $c_1k$ after $r'$ steps, 
the above lemma calculates the probability of that the random walk remains to the right of $c_2k$ in all future rounds.

\begin{lemma} \label{lem:RW_2}
If $W[r-r',r] > c_1k$, for $c_2<c_1$, then we have
\begin{align*}
    \mathbb{P}\big(W[r-r'-b, r] \geq c_2k \; \forall\; b >0 \;\big|\; W[r-r',r] > c_1k) &= 1 - (2\beta)^{(c_1-c_2)k}\\
    &= 1  - e^{\log (2\beta)(c_1-c_2)k}.
\end{align*}
\end{lemma}
\begin{proof}
Refer \cite{randomwalk}.
\end{proof}
The above lemma is mathematically characterizing the same event as Lemma \ref{lem:RW_1}.

\begin{lemma} \label{lemma:RW_3}
If $W[r-r',r] > c_1k$,then for $c_3<c_1$, then we have
\begin{align*}
    \mathbb{P}\left(W[r-r'-b, r+a] \geq c_3k \; \forall\; a >0\;\big|\; W[r-r',r] > c_1k \right) &\geq 1 - 2(2\beta)^{(c_1-c_3)k/2}\\
    &= 1  - 2e^{\log (2\beta)(c_1-c_3)k/2}\\
    &\overset{(a)}{\geq} 1  - 2e^{-(1-2\beta)(c_1-c_3)k/2}.
\end{align*}
\end{lemma}
\begin{proof}
Using $c_2=(c_1-c_3)/2$ in the above two Lemmas \ref{lem:RW_1} and \ref{lem:RW_2}, we get the required result. The inequality $(a)$ uses $\log 2\beta < 2\beta -1$ for $\beta>0$.
\end{proof}

\vspace{-2mm}
\section{Throughput of \bitcoin }
\label{app:btc_thruput}

\subsection{For $\beta \approx 0.5$}
Indeed, in order for \bitcoin to be secured against Nakamoto's private attack \cite{bitcoin} in that regime, it is necessary that $f\Delta$, the expected number of blocks mined per network delay round, approaches $0$ so that very little forking occurs among the honest nodes and the honest nodes can grow the longest chain faster than the adversary. Note that for a given block size $B$,  the throughput is bounded by:
\begin{eqnarray*}
fB  = f\Delta \cdot B/\Delta 
 =   f\Delta \cdot B/(B/C + D)  <  f\Delta C \;\;\mbox{tx/second}
\end{eqnarray*}
Hence, in the regime where $\beta \rightarrow 0.5$, \bitcoin can only achieve a {\em vanishing} fraction of the network capacity. 

\subsection{For general $\beta<0.5$}
An upper bound on the worst case throughput (worst case over all adversary actions) is  the rate at which the longest chain grows when no adversary nodes mine. The longest chain grows by one block in a round exactly when at least one honest block is mined. Hence the rate of growth is simply $\Pr( \text{\#blocks mined in round } r >0)$, i.e.  
\begin{equation}
1 - e^{-(1-\beta) f \Delta} \;\; \mbox{blocks per round},
\label{eq:btc_chain_growth}
\end{equation}
Notice that \eqref{eq:btc_chain_growth} is monotonically increasing in $f$; hence to maximize throughput, we should choose as high a mining rate as possible. 

However, we are simultaneously constrained by security.
For \bitcoin\!\!'s security, \cite{backbone} shows  that the main chain must grow faster in expectation than any adversarial chain, which can grow at rates up to $\beta f \Delta$ in expectation.
Hence we have the following (necessary) condition for security:
\begin{equation}
1 - e^{-(1-\beta) f \Delta} > \beta f \Delta. 
\label{eq:bitcoin_condition}
\end{equation}
Equation \eqref{eq:bitcoin_condition} gives the following upper bound on $f\Delta$, the mining rate per round:
$$
f\Delta <  \bar{f}_{\text{BTC}}(\beta), 
$$
where $\bar{f}_{\text{BTC}}(\beta)$ is the unique solution to the equation:
\beq
\label{eq:fp}
1 - e^{-(1-\beta) \bar{f}} = \beta \bar{f}.
\eeq
This yields an upper bound on the throughput, in transactions per second, achieved by \bitcoin as:
\begin{eqnarray}
\lambda_{\text{BTC}} & \le & [1 - e^{-(1-\beta) \bar{f}_{\text{BTC}}(\beta)}] B/\Delta \nonumber \\
& = & \beta \bar{f}_{\text{BTC}}(\beta) B/\Delta,
\label{eq:bc_ub}
\end{eqnarray}
where the last equality follows from \eqref{eq:fp}. Substituting in $\Delta = B/C + D$ and optimizing for $B$, we get the following upper bound on the maximum efficiency of \bitcoin:
$$
\bar \lambda_{\text{BTC}} \leq \beta \bar{f}_{\text{BTC}}(\beta), 
$$
achieved when $B \gg CD$ and $\Delta \gg D$.

Another upper bound on the throughput is obtained by setting $f$ at the capacity limit: $f = C/B$ 
(cf. section \eqref{sec:model}). Substituting into \eqref{eq:btc_chain_growth} and optimizing over $B$, this yields 
$$
\label{eq:ub2}
\bar \lambda_{\text{BTC}} \leq 1 - e^{\beta - 1}, 
$$
achieved when $f\Delta = 1$,  $B \gg CD$ and $\Delta \gg D$.

Combining the above two bounds, we get:
$$
\bar \lambda_{\text{BTC}} \leq \min \left \{ \beta \bar{f}_{\text{BTC}}(\beta),  1 - e^{\beta -1} \right \}
$$
This is plotted in Figure \ref{fig:throughput_comparision}. Note that for large values of $\beta$, the first upper bound is tighter; this is a {\em security-limited} regime, in which the throughput efficiency goes to zero as $\beta \to 0.5$. This is a manifestation of the (well-known) fact that to get a high degree of security, i.e. to tolerate $\beta$ close to $0.5$, the mining rate of \bitcoin must be small, resulting in a low throughput. \bitcoin currently operates in this regime, with the mining rate one block per $10$ minutes; assuming a network delay of $1$ minute, this corresponds to a tolerable $\beta$ value of $0.49$ in our model.

For smaller $\beta$, the second upper bound is tighter, i.e. this is the {\em communication-limited} regime. The crossover point is the value of $\beta$ such that 
$$ 1 - e^{\beta - 1} = \beta,$$
i.e.,  $ \beta \approx 0.43$.

\section{Throughput of ghost}
\label{app:ghost}

The \ghost \cite{ghost}  protocol uses a different fork choice rule, which uses the heaviest-weight subtree (where weight is defined as the number of blocks in the subtree), to select the main chain. 
To analyze the throughput of \ghostnosp, we first observe that when there are no adversarial nodes working, the growth rate of the main chain of \ghost is upper bounded by the growth rate of the main chain under the longest chain rule. Hence, the worst-case throughput of \ghostnosp, worst-case over all adversary actions, is bounded by that of \bitcoinnosp, i.e.
\beq \label{eq:ghost_chain_growth}
1 - e^{-(1-\beta) f \Delta} \quad \mbox{blocks per round},
\eeq
(cf. \eqref{eq:btc_chain_growth}). Notice that once again, this bound is monotonically increasing in $f$ and we would like to set $f$ largest possible subject to security and network stability constraints. The latter constraint gives the same upper bound as \eqref{eq:ub2} for \bitcoinnosp:
\beq
\label{eq:ub2}
\bar \lambda_{\ghost} \leq 1 - e^{\beta - 1}. 
\eeq
\ref{sec:seq_to_ledger}
We now consider the security constraint on $f$. Whereas our security condition for \bitcoin throughput was determined by a Nakamoto private attack (in which the adversary builds a longer chain than the honest party), a more severe attack for \ghost is a balancing attack, analyzed in the next subsection.  As shown in that analysis, the balancing attack implies that a necessary condition on $f$ for robustness against an adversary with power $\beta$ is given by:
\beq
\label{eq:bal}
\beta f \Delta < \mathbb{E}[|H_1[r] - H_2[r]|] ,
\eeq
where $H_1[r], H_2[r]$ are two independent Poisson random variables each with mean $(1-\beta)f\Delta/2$.
Repeating the same analysis as we did for \bitcoinnosp, we get the following upper bound on the maximum efficiency of \ghostnosp:
\beq
\bar \lambda_{\ghost} \leq \beta \bar{f}_{\ghost}(\beta),
\eeq
where $\bar{f}_{\ghost}(\beta)$ is the value of $f\Delta$ such that \eqref{eq:bal} is satisfied with equality instead of inequality. 

Combining this expression with the network stability upper bound, we get:
\beq
\bar \lambda_{\ghost} \leq \min \left \{ \beta \bar{f}_{\ghost}(\beta),  1 - e^{\beta - 1} \right \}. 
\eeq
The throughput is plotted in Figure \ref{fig:throughput_comparision}. As in  \bitcoinnosp, there are two regimes, communication-limited for $\beta$ small, and security-limited for $\beta$ large. Interestingly, the throughput of \ghost goes to zero as $\beta$ approaches  $0.5$, just like \bitcoinnosp. So although \ghost was invented to improve the throughput-security tradeoff of \bitcoinnosp, the mining rate $f$ still needs to vanish as $\beta$ gets close to $0.5$. The reason is that although \ghost is indeed secure against Nakamoto private attacks for any mining rate $f$ \cite{ghost}, it is not secure against balancing attacks for $f$ above a threshold as a function of $\beta$. When $\beta$ is close to $0.5$, this threshold goes to zero.

\subsection{Mining rate constraint}

Similar to the balancing attack in \cite{ghost_attack}, we would like to analyze its constraint on the mining rate $f$ which in turns constrains the throughput.
The adversary strategy is to divide the work of honest users by maintaining two forks:
\begin{enumerate}[leftmargin=*]
    \item Say \textbf{two} blocks $b_1, b_2$ are mined over the main chain block $b_{0}$ in the first round. Say the adversary mines $b_1$ and the honest nodes mine $b_2$. The adversary will broadcast both these blocks (and all previous blocks) to all the honest users. This is when the attack starts.
    \item At this time instance (say $r=1$) all the honest nodes have the same view of the blocktree -- which has two main chains ending at blocks $b_1$ and $b_2$.
    \item The honest users are divided into two equal groups $G_1$ and $G_2$, mining over $b_1$ and $b_2$ respectively, each at average rate $(1-\beta)f\Delta/2$ blocks per round each.
    \item The adversary's goal is to maintain the \textbf{forking} - make sure that $G_1$ chooses block $b_1$ in its main chain, whereas $G_2$ chooses block $b_2$ in its main chain.
    To do this, it divides its own resources into two equal parts $A_1$ and $A_2$, each with average mining rate $f\Delta/2$ blocks per round.
    The first part $A_1$  mines only (direct) children of block $b_1$ and second part mines  $A_2$ (direct) children of block $b_2$. Suppose at round $r$, $H_1[r], H_2[r] \sim \pois(1-\beta)f\Delta/2)$ honest blocks are
    mined in subtree 1 (below $b_1$) and subtree 2 (below $b_2$)  respectively. 
    \item \textit{Attack Strategy:} 
    \begin{itemize}[leftmargin=*]
        \item If $H_1[r] = H_2[r]$, then the adversary does nothing.
        \item If say $H_1[r]$ is larger, then adversary releases $H_1[r] - H_2[r]$ blocks that it has mined in subtree $2$  (either in private or just mined in this round). Vice versa for the case when $H_2[r]$ is larger.  This (re)balances the weight of the two subtrees and the honest work is again split in the next round.
    \end{itemize}
    \item \textit{Analysis:}
     The expected number of blocks the adversary needs to release in subtree $1$ per round is $\mathbb{E}[(H_2[r] - H_1[r])^+]$. A necessary condition for this attack to not be able to continue indefinitely with non-zero probability is
     $$ \beta f \Delta < \mathbb{E}[(H_2[r] - H_1[r])^+]  /2,$$

\end{enumerate}

\section{Additional Simulations}
\label{app:simulations}
Here we show our simulations from Section \ref{sec:sims} under additional parameter settings. 
First, we consider an active adversary of hash power $\tilde \beta=0.3$ and $\tilde \beta = 0.15$ with confirmation reliability $\epsilon=e^{-20}$ in Figure \ref{fig:active_e_20}.
Notice that for $\tilde \beta=0.3$, the latency of confirming double-spent transactions exceeds that of the longest-chain protocol, as explained in Section \ref{sec:analysis}. 
The numeric latency values of the double-spent transaction curve are colored green to clarify which numbers belong to which curve.


\begin{figure*}[htb!]
\minipage{0.48\textwidth}
\includegraphics[width=\linewidth]{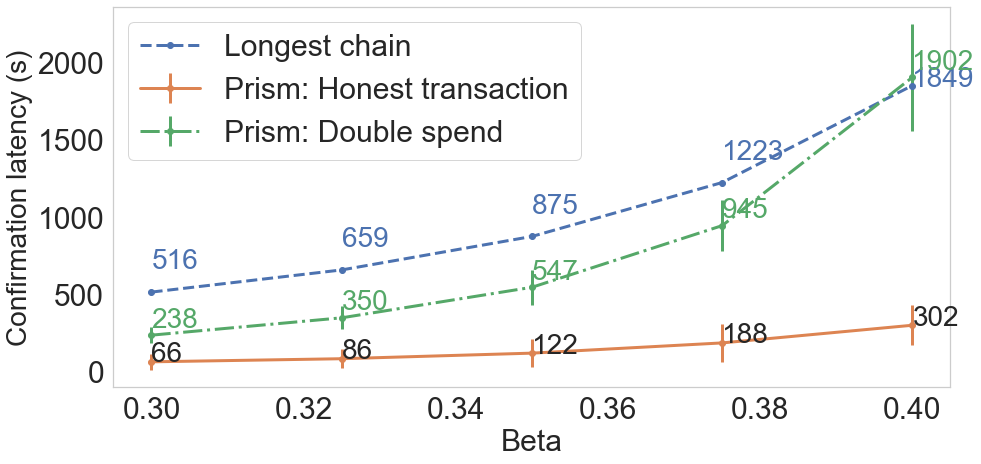}
\put(-120,95){Balancing attack}
    \label{fig:balancing}
\endminipage\hfill
\minipage{0.48\textwidth}%
\includegraphics[width=\linewidth]{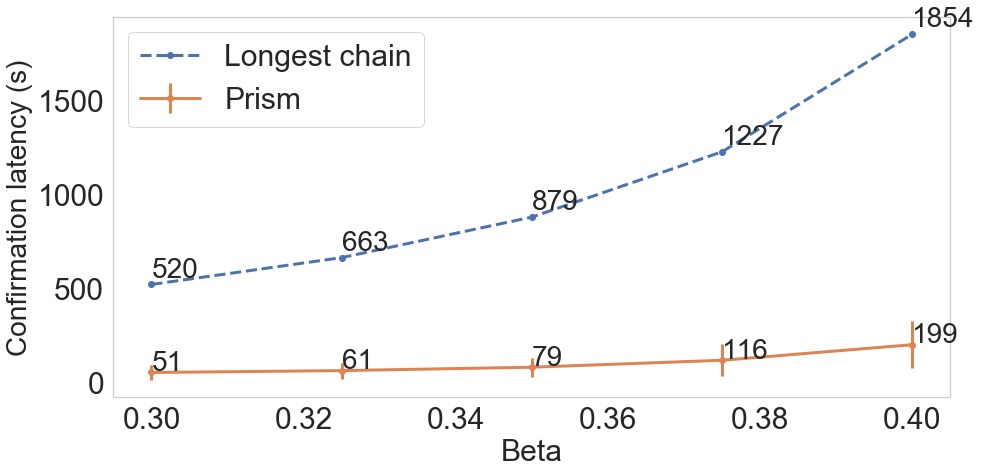}
\put(-120,95){Censorship attack}
    \label{fig:censor}
\endminipage
\put(-250,-60){\textbf{Active adversary $\tilde \beta=0.3$}}
\newline
\minipage{0.48\textwidth}
\includegraphics[width=\linewidth]{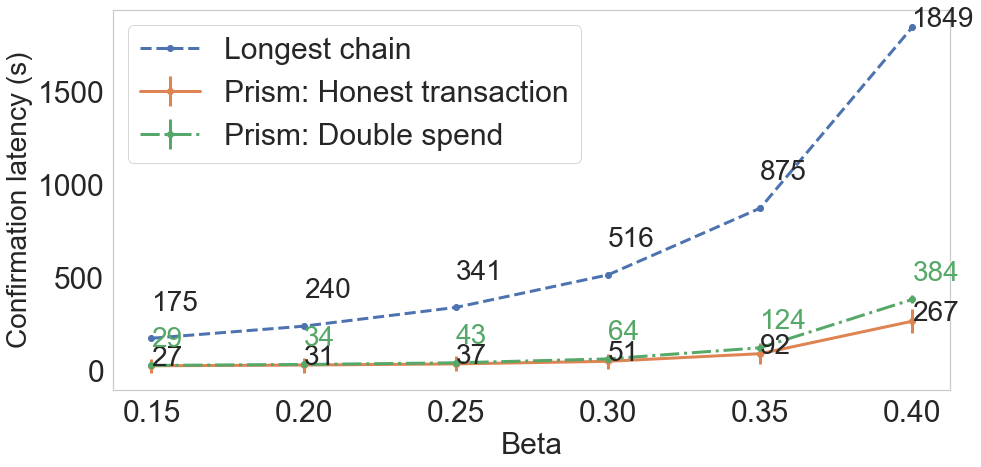}
\put(-120,95){Balancing attack}
    \label{fig:balancing}
\endminipage\hfill
\minipage{0.48\textwidth}%
\includegraphics[width=\linewidth]{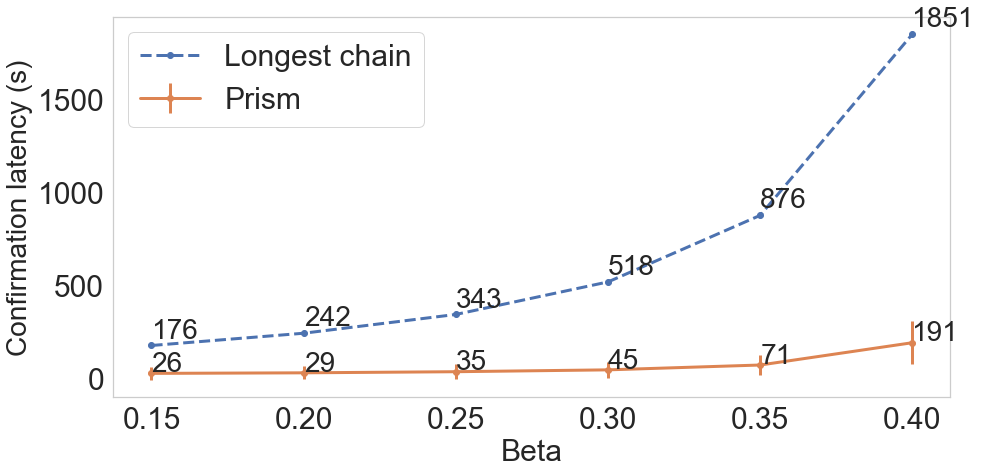}
\put(-120,95){Censorship attack}
    \label{fig:censor}
\endminipage
\put(-250,-60){\textbf{Active adversary $\tilde \beta=0.15$}}
\caption{Active adversary at confirmation reliability $\epsilon=e^{-20}$.}
\label{fig:active_e_20}
\end{figure*}

Since $\epsilon=e^{-20}\approx 2.1 \times 10^{-9}$ is fairly conservative (this corresponds to a latency on the order of 1 day at Bitcoin's current settings of 1 block every 10 minutes with $\beta=0.4$), we also consider a weaker confirmation guarantee of $\epsilon=e^{-10}\approx 4.5\times 10^{-5}$.
The results for this weaker confirmation reliability are shown in Figure \ref{fig:active_e_10}.
As expected, all confirmation times are reduced, both for Bitcoin and for Prism. 
Another key difference relates to double-spent transaction latency under balancing attacks. 
As $\epsilon$ grows, Prism's latency overtakes that of the longest-chain protocol for smaller values of $\beta$. 
(Recall that $\tilde \beta$ denotes the current fraction of hash power that is actively launching the attack, while $\beta$ is the maximum tolerable fraction of adversarial hash power against which the system is secure.)
This observation is expected. 
Prism's latency does not grow significantly as $\epsilon$ changes;
notice the similarity in Prism's numeric latency values between Figures \ref{fig:active_e_20} and \ref{fig:active_e_10}. 
However, as $\epsilon$ decreases, it significantly increases the latency of the longest-chain protocol.

\begin{figure*}[htb!]
\minipage{0.48\textwidth}
\includegraphics[width=\linewidth]{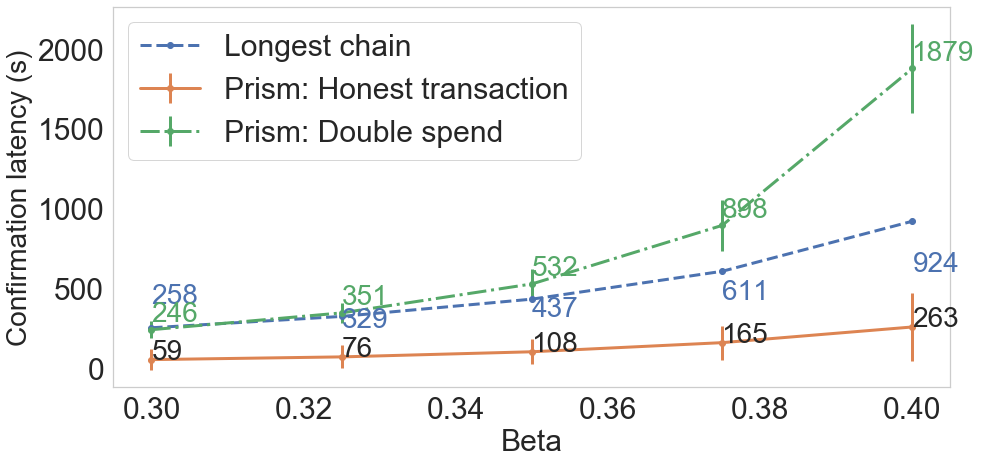}
\put(-120,95){Balancing attack}
    \label{fig:balancing}
\endminipage\hfill
\minipage{0.48\textwidth}%
\includegraphics[width=\linewidth]{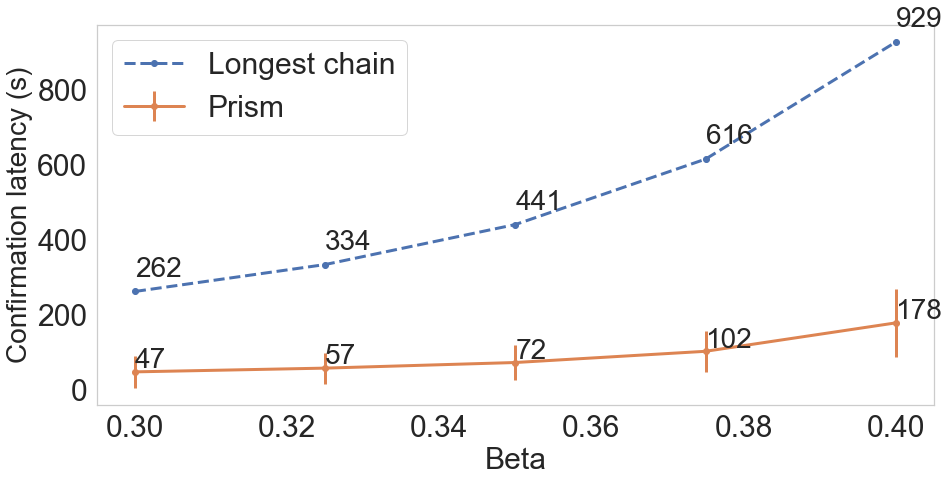}
\put(-120,95){Censorship attack}
    \label{fig:censor}
\endminipage
\put(-250,-60){\textbf{Active adversary $\tilde \beta=0.3$}}
\newline
\minipage{0.48\textwidth}
\includegraphics[width=\linewidth]{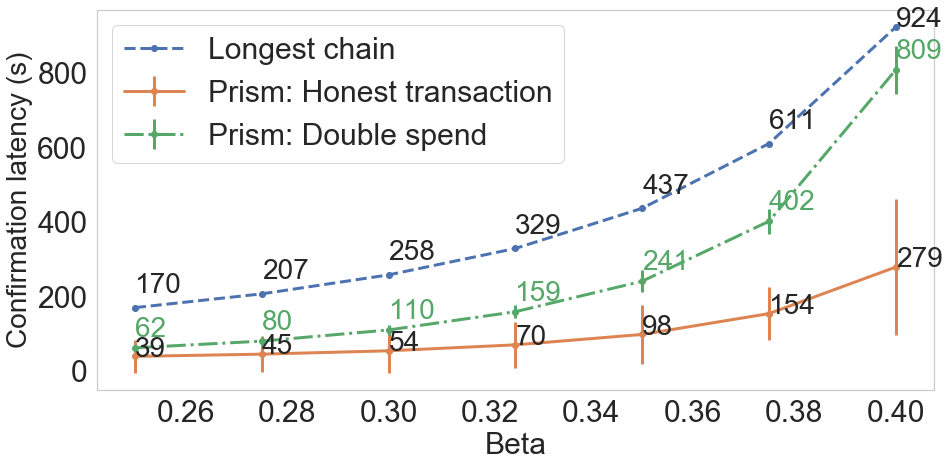}
\put(-120,95){Balancing attack}
    \label{fig:balancing}
\endminipage\hfill
\minipage{0.48\textwidth}%
\includegraphics[width=\linewidth]{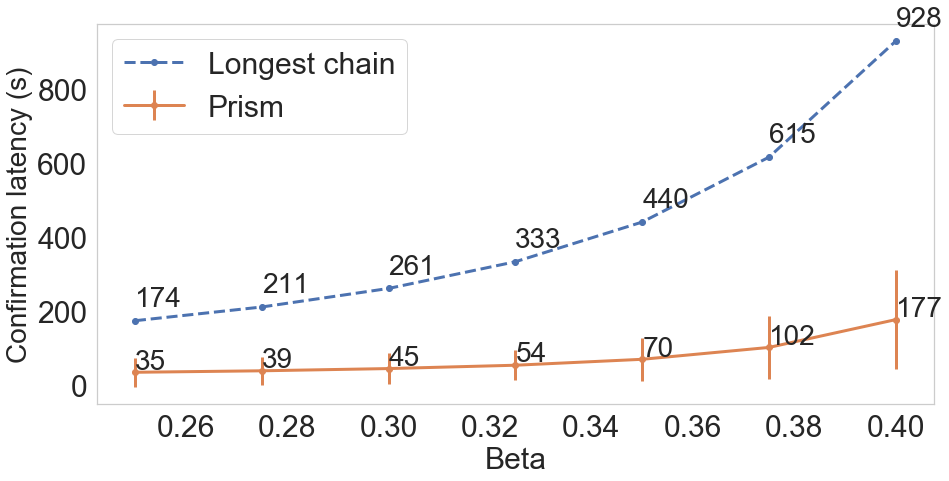}
\put(-120,95){Censorship attack}
    \label{fig:censor}
\endminipage
\put(-250,-60){\textbf{Active adversary $\tilde \beta=0.25$}}
\newline
\minipage{0.48\textwidth}
\includegraphics[width=\linewidth]{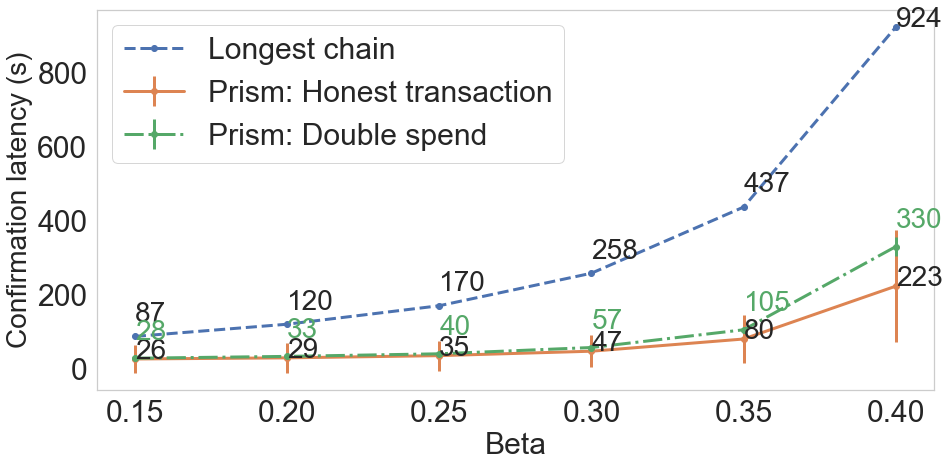}
\put(-120,95){Balancing attack}
    \label{fig:balancing}
\endminipage\hfill
\minipage{0.48\textwidth}%
\includegraphics[width=\linewidth]{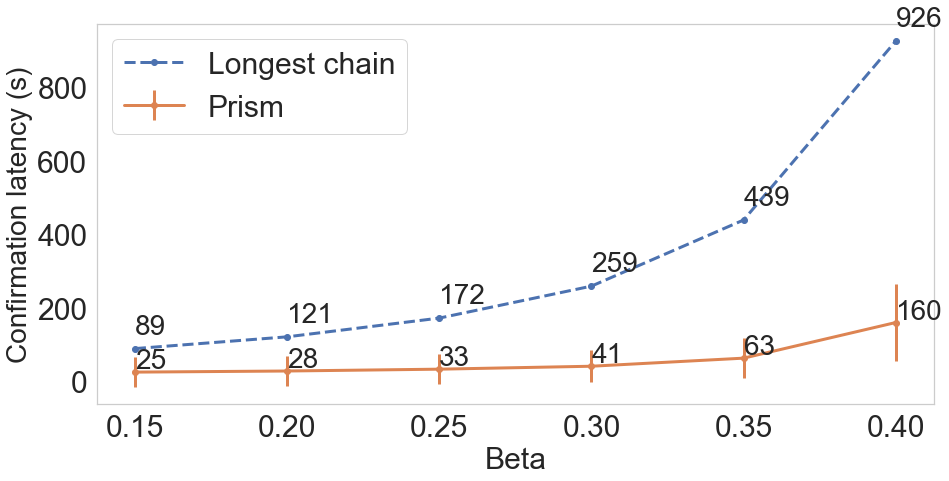}
\put(-120,95){Censorship attack}
    \label{fig:censor}
\endminipage
\put(-250,-60){\textbf{Active adversary $\tilde \beta=0.15$}}
\caption{Active adversary at confirmation reliability $\epsilon=e^{-10}$.}
\label{fig:active_e_10}
\end{figure*}

To illustrate this effect more explicitly, Figure \ref{fig:varying_epsilon} shows the latency of Prism and the longest-chain protocol as we vary $\epsilon$ for a fixed $\beta=0.4$ and $\tilde \beta=0.3$.
Notice that indeed, Prism's latency changes very little as $\epsilon$ scales, whereas the longest-chain latency grows much faster. 
Additionally, we observe that censorship attacks do not appear to significantly affect latency  compared to the non-adversarial setting, whereas balancing attacks can incur a much higher latency than the longest-chain protocol, with an especially pronounced difference at small values of $\epsilon$. 
This disparity becomes less notable for smaller values of $\tilde \beta$, as seen in the earlier plots. 

\begin{figure*}[htb!]
\minipage{0.32\textwidth}
\includegraphics[width=\linewidth]{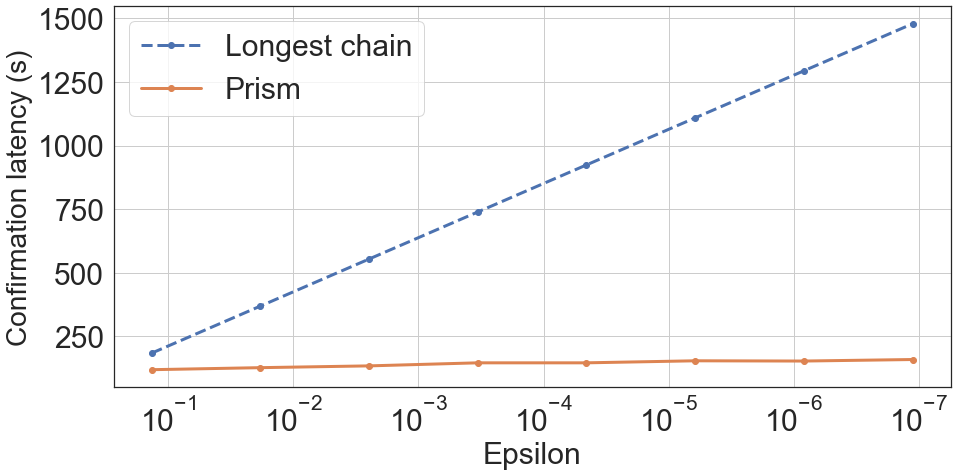}
\put(-100,65){Non-adversarial setting}
    \label{fig:balancing}
\endminipage\hfill
\minipage{0.32\textwidth}%
\includegraphics[width=\linewidth]{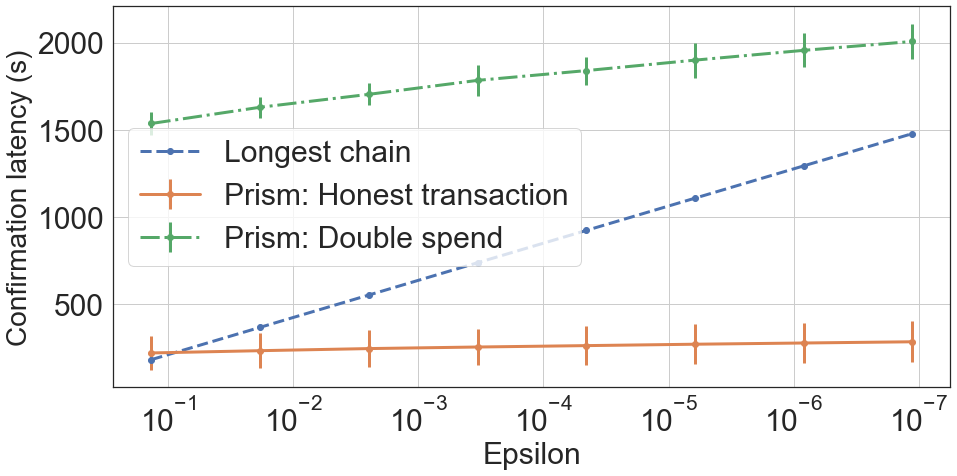}
\put(-90,65){Balancing attack}
    \label{fig:balancing_eps}
\endminipage\hfill
\minipage{0.32\textwidth}%
\includegraphics[width=\linewidth]{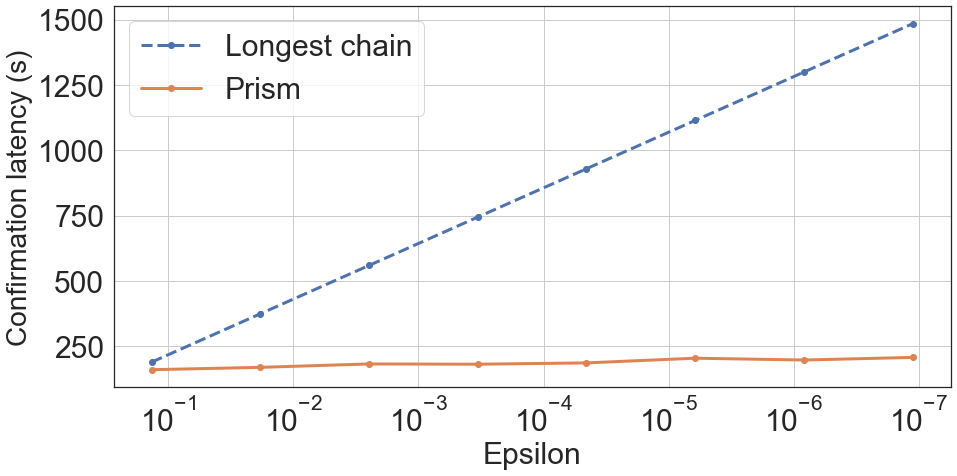}
\put(-95,65){Censorship attack}
    \label{fig:censor_eps}
\endminipage
\caption{Latency vs. $\epsilon$ for a non-adversarial setting, a balancing attack, and a censorship attack when $\beta=0.4$ and the active adversarial hash power $\tilde \beta=0.3$.}
\label{fig:varying_epsilon}
\end{figure*}

\restoregeometry 
\end{document}